\documentclass[12pt]{amsart}
\usepackage[pdftex=true]{hyper ref}

\def\smallsection#1{\smallskip\noindent\textbf{#1}.}

\usepackage{mathrsfs}
\usepackage{url} 

\usepackage{color}
\usepackage{amsthm}
\usepackage{amsxtra}
\usepackage{amssymb}
\usepackage{graphicx}  
\newcommand{\xqedhere}[1]{%
    \rlap{%
         \hbox to#1{%
           \hfil
           \llap{%
               \ensuremath{\square}
           }%
       }%
   }%
}

\newcommand{\nwc}{\newcommand}

\nwc{\blds}{\boldsymbol}

\newcommand{\bep}{\blds{\epsilon}}

\newcommand{\hph}{{\gamma}}
\newcommand{\hG}{{ G_0}}
\newcommand{\Op}{{\operatorname{Op}^{{w}}_h}}
\newcommand{\Opt}{{\operatorname{Op}^{{w}}_{\tilde h}}}

\newcommand{\be}{\begin{equation}}
\newcommand{\ee}{\end{equation}}

\newcommand{\ra}{\rangle}
\newcommand{\la}{\langle}

\newcommand{\CI}{{\mathcal C}^\infty }
\newcommand{\CIc}{{\mathcal C}^\infty_{\rm{c}} }

\newcommand{\cH}{{\mathcal H}}
\newcommand{\cG}{{\mathcal G}}
\newcommand{\cK}{\mathcal K}
\newcommand{\cL}{\mathcal L}

\newcommand{\cO}{{\mathcal O}}
\newcommand{\Oo}{{\mathcal O}} 

\newcommand{\cS}{{\mathscr S}}
\newcommand{\cT}{{\mathcal T}}

\newcommand{\cU}{{\mathcal U}}
\newcommand{\VV}{{\mathcal V}}

\newcommand{\CC}{{\mathbb C}}
\newcommand{\NN}{{\mathbb N}}
\newcommand{\IN}{{\mathbb N}}

\newcommand{\RR}{{\mathbb R}}
\newcommand{\IR}{{\mathbb R}}
\newcommand{\TT}{{\mathbb T}}

\newcommand{\ZZ}{{\mathbb Z}}

\newcommand{\ta}{\widetilde a}

\newcommand{\tF}{\varphi_{t_0} }
\renewcommand{\th}{{\tilde h}}

\newcommand{\tN}{\widetilde N}

\newcommand{\tQ}{q}
\newcommand{\tS}{\widetilde S}
\newcommand{\tT}{\widetilde T}
\newcommand{\tU}{\widetilde U}

\newcommand{\tkappa}{{\widetilde \kappa}}

\newcommand{\tchi}{\widetilde \chi}

\newcommand{\tPsi}{\widetilde \Psi}
\newcommand{\tP}{\widetilde P}

\newcommand{\trho}{{\tilde \rho}}

\newcommand{\teta}{{\tilde \eta}}

\newcommand{\ty}{{\tilde y}}
\newcommand{\tw}{{\tilde w}}

\newcommand{\bj}{\boldsymbol{j}}
\newcommand{\bT}{\boldsymbol{T}}
\newcommand{\bu}{\boldsymbol{u}}

\newcommand{\defeq}{\stackrel{\rm{def}}{=}}

\newcommand{\sgn}{\operatorname{sgn}}
\newcommand{\vol}{\operatorname{vol}}

\newcommand{\supp}{\operatorname{supp}}

\renewcommand{\dim}{\operatorname{dim}}
\newcommand{\WF}{\operatorname{WF}}
\newcommand{\WFh}{\operatorname{WF}_h}

\newcommand{\comp}{{\operatorname{comp}}}

\newcommand{\rest}{\!\!\restriction}

\renewcommand{\Re}{\mathop{\rm Re}\nolimits}
\renewcommand{\Im}{\mathop{\rm Im}\nolimits}
\newcommand{\ad}{\operatorname{ad}}

\newcommand{\neigh}{\operatorname{neigh}}

\theoremstyle{plain}

\newtheorem{thm}{Theorem}
\newtheorem{prop}{Proposition}[section]
\newtheorem{cor}[thm]{Corollary}
\newtheorem{lem}[prop]{Lemma}

\theoremstyle{definition}

\newtheorem{rem}[prop]{Remark}

\numberwithin{equation}{section}

\def\bbbone{{\mathchoice {1\mskip-4mu {\rm{l}}} {1\mskip-4mu {\rm{l}}}
{ 1\mskip-4.5mu {\rm{l}}} { 1\mskip-5mu {\rm{l}}}}}

\def\squarebox#1{\hbox to #1{\hfill\vbox to #1{\vfill}}}

\newcommand{\Id}{{ I}}

\newcommand{\eps}{{\epsilon}}
\newcommand{\vareps}{{\varepsilon}}

\renewcommand{\emptyset}{\varnothing}

\usepackage{amsxtra}

\ifx\pdfoutput\undefined
  \DeclareGraphicsExtensions{.pstex, .eps}
\else
  \ifx\pdfoutput\relax
    \DeclareGraphicsExtensions{.pstex, .eps}
  \else
    \ifnum\pdfoutput>0
      \DeclareGraphicsExtensions{.pdf}
    \else
      \DeclareGraphicsExtensions{.pstex, .eps}
    \fi
  \fi
\fi

\title[Decay of correlations]
{Decay of correlations for normally hyperbolic trapping}

\author[S. Nonnenmacher]
{St\'ephane Nonnenmacher}
\author[M. Zworski]
{Maciej Zworski}
\address{Institut de Physique Th\'eorique\\
CEA/DSM/IPhT, Unit\'e de recherche associ\'ee au CNRS\\
CEA-Saclay\\
91191 Gif-sur-Yvette, France}
\email{snonnenmacher@cea.fr}
\address{Mathematics Department, University of California \\
Evans Hall, Berkeley, CA 94720, USA}
\email{zworski@math.berkeley.edu}

\begin{document}    

\begin{abstract}
We prove that for evolution problems with
normally hyperbolic trapping in phase space, 
correlations decay exponentially in time. Normally
hyperbolic trapping means that the trapped
set is smooth and symplectic and that the flow is 
hyperbolic in directions transversal to it.  
Flows with this structure include 
contact Anosov flows \cite{f-sj},\cite{Tsu10},\cite{Ts}, classical flows in 
molecular dynamics \cite{GeSj2},\cite{GSWW10}, 
and null geodesic flows for black holes metrics \cite{Dya01},\cite{Dya02},\cite{WuZ}. 
The decay of correlations is a consequence of the
existence of resonance free strips for Green's functions
(cut-off resolvents) and polynomial bounds on the growth
of those functions in the semiclassical parameter.

\end{abstract}

\maketitle

\section{Statement of results}
\label{sor}

\subsection{Introduction}
\label{int}

We prove the existence of resonance free strips for general
semiclassical problems with normally hyperbolic trapped sets. 
The width of the strip is related to certain Lyapunov exponents and,
for the spectral parameter in that strip, the Green's function (cut-off resolvent) is polynomially bounded.
Such estimates are closely related to exponential decay of 
correlations in classical dynamics and in scattering problems.
The framework to which our result applies covers both settings. 

To illustrate the results consider 
\begin{equation}
\label{eq:PP}
 P = - h^2 \Delta + V ( x) , \ \ \ V \in \CIc ( \RR^n ; \RR )
.\end{equation}
The classical flow $ \varphi_t : ( x ( 0 ) , \xi ( 0 ) ) \mapsto ( x (
t ) , \xi ( t ) ) $ is obtained by solving Newton's equations $ x' ( t
) 
( t ) = 2 \xi ( t ) $, $  \xi' ( t ) = - \nabla V ( x ( t ) )
$. The trapped set at energy $ E $, $ K_E $,  is defined as the
set of $ ( x, \xi ) $ such that $p(x,\xi)\defeq \xi^2 + V ( x ) = E $ and $
\varphi_t ( x, \xi ) \not \to \infty $, as $ t \to \infty $ and as 
$ t \to - \infty $.

The flow $ \varphi_t $ is said to be {\em normally hyperbolic near
  energy $ E $}, if for some $ \delta > 0 $, 
\begin{gather}
\label{eq:NN}
\begin{gathered}
K^\delta \defeq \bigcup_{ |  E - E' | < \delta } K_{E'}  \ \text{ is a
  smooth symplectic manifold, and } \\
\text{ the flow $\varphi_t $ is hyperbolic in the directions
  transversal to $ K^\delta $,}
\end{gathered}
\end{gather}
see  \eqref{eq:NH} below for a precise definition, and 
\cite{GSWW10} for physical motivation for considering such dynamical 
setting.
A simplest consequence of Theorems \ref{t:1} and \ref{t:2'} is
the following result about decay of correlations.

\begin{thm}
\label{t:0}
Suppose that $ P $ is given by \eqref{eq:PP} and that \eqref{eq:NN}
holds, that is the classical flow is normally hyperbolic near energy $
E $. Then for $ \psi \in \CIc ( ( E - \delta/2 , E + \delta/ 2 ) )$, 
and any $ f , g \in  L^2 ( \RR^n )  $,  with 
$ \| f \|_{L^2 } = \| g \|_{L^2} = 1 $,  $ \supp f , \supp g \subset B
( 0 , R ) $, 
\begin{equation}
\label{eq:t0}  \left| \langle  e^{ - i t P / h } \psi (  P )  f , g \rangle_{
    L^2 ( \RR^n  ) } \right|   \leq  \frac {C_R  \log  ( 1/h) } { h^{1
    + \gamma c_0 }}    e^{ -  \gamma t} +   C_{R, N }  h^N  , \ \ t > 0 , 
\end{equation}
for any $ \gamma < \lambda_0/2 $ and for all $ N $. Here $ \lambda_0 $ and
$ c_0 $ are the same as in \eqref{eq:t1} and $ C_R $, $
C_{R, N } $ are 
constants depending on $ R $ and on $ R $ and $ N $, respectively.
\end{thm}

This means that the correlations decay rapidly in the semiclassical
limit: we start with a state localized in space (the support
condition) and energy, $ \psi ( P ) f $, propagate it, and test it
against another spatially localized state $ g$.
The estimate \eqref{eq:t0}  is a consequence of the existence of 
a band without scattering resonances and estimates on cut-off
resolvent given in Theorem \ref{t:1}.
When there is no trapping, that is when $ K_E = \emptyset $,  then 
the right hand side in \eqref{eq:t0} can be replaced by 
$ {\mathcal O} ( (h/t)^\infty ) $, provided that $ t > T_E $, for some
$ T_E $ -- see for instance \cite[Lemma 4.2]{NSZ}. On the other
hand when strong trapping is present, for instance when the 
potential has an interaction region separated from infinity by a 
barrier, then the correlation does not decay -- see \cite{NSZ} and references given there.

More interesting quantitative results can be obtained for
the wave equation or for decay of classical correlations: see \S \ref{mot} for motivation and 
\cite[Theorem 3]{WuZ} and Corollary~\ref{t:4} below for examples. When 
the outgoing and incoming sets at energy $E$,
\[  \Gamma_E^\pm \defeq \{ ( x , \xi ) : p(x,\xi)=E,\ 
\varphi_t ( x, \xi ) \not \to \infty ,  t \to \mp \infty \} , \]
 are
sufficiently regular and of codimension one, Theorem \ref{t:0} and Theorem \ref{t:1} below (without the specific constant
$ \lambda_0 $) are already a  consequence of 
earlier work by Wunsch--Zworski 
\cite[Theorem 2]{WuZ}\footnote{Recently Dyatlov \cite{Dya2} provided a much simpler proof of 
that result, including the optimal size of the gap
established in this paper and the optimal resolvent bound $ o ( h^{-2} )$, for smooth and orientable stable 
and unstable manifolds.}
and, in the case of closed trajectories, Christianson \cite{Chr1},\cite{Chr2}.
For a survey of other recent results
on resolvent estimates in the presence of weak trapping 
we refer to \cite{Wu}.

When normal hyperbolicity is strengthened
to $ r $-normal hyperbolicity for large $ r $ (which implies that $
\Gamma_E^\pm $ are $ C^r $ manifolds) and provided a certain
pinching condition on Lyapunov exponents is satisfied, much 
stronger results have been obtained by Dyatlov \cite{Dya1}. 
In particular, \cite{Dya1} provides an asymptotic counting law for scattering
resonances {\em below} the band without resonances given in 
Theorems \ref{t:1}, \ref{t:3} and \ref{t:2'}. It shows the optimality 
of the size of the band in a large range of settings, for instance, 
for perturbations of Kerr--de Sitter black holes.

Similar results on asymptotic counting laws in strips have
  been proved by Faure--Tsujii
in the case of Anosov diffeomorphisms \cite{f-ts}, and recently
announced in the case of contact Anosov flows
\cite{f-ts2}. In the latter situation, described in
Theorem~\ref{t:3} below, 
the trapped set is a normally hyperbolic smooth symplectic manifold,
but the dependence of the stable and unstable subspaces on points
on the trapped set is typically nonsmooth, but $C^1$ or
H\"older continuous (see Remark \ref{rem:regularity} below). 
For compact manifolds of constant negative curvature 
Dyatlov--Faure--Guillarmou \cite{DFG} have provided a precise
description of Pollicott--Ruelle resonances in terms of 
eigenvalues of the Riemannian Laplacian acting on section
of certain natural vector bundles.

In this paper we do {\em not} assume any regularity on $ \Gamma_E^\pm $ 
and provide a quantitative estimate on the resonance free strip.
For operators with analytic coefficients this result was already 
obtained by G\'erard--Sj\"ostrand \cite{GeSj2} with even
weaker assumptions on $ K^\delta $. A new component
here, aside from dropping the analyticity assumption,  is the
polynomial bound on the Green's function/resolvent that
allows applications to the decay of correlations.

The proof is given first for an operator with a complex
absorbing potential. This allows very general assumptions
which can then be specialized to scattering and dynamical 
applications.

Finally we comment on the comparison between the resonance free regions in this paper and the results of \cite{NZ3,NZ4} where the
existence of a resonance free strip was given for {\em hyperbolic} trapped
sets, provided a certain {\em pressure condition} was satisfied. 
In the setting of \cite{NZ3} the trapped set is typically very 
irregular but, the assumptions of \cite{NZ3}  also include the situation where
$ K^\delta $ is a smooth symplectic submanifold, and  the flow is 
hyperbolic both transversely to $K^\delta$ and along each $ K_E $. 
In that case the resonance gap obtained in \cite{NZ3} involves a
topological pressure associated with the full (that is, longitudinal and
transverse) unstable Jacobian, namely 
\be\label{e:P(1/2)}
\mathcal{P}\big(-\frac12 (\log J^+_\parallel + \log J^+_\perp)\big)=\sup_{\mu}\Big(H(\mu)-\frac12\int  (\log J^+_\parallel + \log J^+_\perp)\,d\mu \Big)\,,
\ee
where the supremum is taken over all flow-invariant probability measures on $K^\delta$ and $H(\mu)$ is the Kolmogorov--Sinai entropy of the measure $\mu$ with respect to the flow. The bound is nontrivial only if this pressure is negative.
In the case of mixing Anosov flows discussed in \S\ref{decay} the transverse
and longitudinal unstable Jacobians are equal to each other; the
above pressure is then equal to the pressure $\mathcal{P}(-\log J^+_\parallel)$, equivalent with the pressure $\mathcal{P}(-\log J^u)$ of the Anosov flow,
which is known to vanish \cite[Proposition 4.4]{BoRu75}, and hence 
gives only a trivial bound. For
this situation, our spectral bound (Theorem~\ref{t:3}) is thus sharper than the pressure bound. 
On the other hand, one can construct 
examples where the longitudinal and transverse unstable Jacobians are 
independent of one another, and such that the pressure \eqref{e:P(1/2)} is more negative - hence sharper - than the value $-\lambda_0$ given in \eqref{eq:t11}, which may be expressed as
$-\lambda_0= \sup_{\mu} \Big(- \frac12 \int  \log J^+_\perp\,d\mu\Big)$.

\smallsection{Notation} We use the following notation $ g =  \cO_k ( f )_V $ means that
$ \|g \|_V  \leq C_k  f $ where the norm (or any seminorm) is in the
space $ V $, an the $ C_k $ depends on $ k$. When either $ k $ or
$ V $ are dropped then the constant is universal or the estimate is
scalar, respectively. When $ F = \cO_k ( f )_{V\to W} $ then 
the operator $ F : V \to W $ has its norm bounded by $ C_k f $.

\subsection{Motivation}
\label{mot} 
To motivate the problem we consider the following elementary 
example. Let $ X = \RR $ and $ P = - \partial_x^2 $. A
wave evolution is given by 
$  U ( t ) \defeq   { \sin (\sqrt P t)  }/ { \sqrt P} $.
Then for $ f , g \in \CIc ( \RR) $ and any time $t\in\IR$ we define the wave {\em correlation
function} as 
\begin{equation}
\label{eq:LP0} 
C ( f , g ) ( t ) \defeq 
  \int_\RR  [U ( t ) f](x)\, g(x) \, dx 
\end{equation}
In this 1-dimensional setting, the correlation function becomes very
simple for large times. Indeed, for a certain $T>0$ depending on the
support of $ f $ and $ g $, it satisfies
$$
\forall t\geq T,\qquad C ( f , g ) ( t ) = \frac12 \int_{\IR} f(x)\,\int_{\IR} g(x)\, dx 
$$
This particular behaviour is due to the fact that the resolvent
of $ P $,
\[ R ( \lambda ) \defeq ( P - \lambda^2 )^{-1} : L^2 ( \RR) \to L^2 ( \RR
) , \ \ \Im \lambda > 0 , \]
continues meromorphically to $ \CC $ in $ \lambda $ as an operator
$ L^2_{\rm comp} \to L^2_{\rm loc} $ and has a pole at $ \lambda = 0 $. 
In this basic case we see this from an explicit formula, 
\[   [R ( \lambda ) f] ( x) = \frac{ i } { 2 \lambda} \int_\RR e^{ i
  \lambda | x - y | } f( y ) dy . \]

More generally, we can consider $ P = - \partial^2_x  + V ( x ) $, $ V
\in L_{\rm{c}} ^\infty ( \RR ) $,  with $ V \geq 0 $, for simplicity. 
 With the same definition of $  U ( t ) $ we now have the Lax--Phillips
 expansion generalizing \eqref{eq:LP0}:
\begin{equation}
\label{eq:LP}
 C ( f , g ) ( t ) = \int_\RR U ( t ) f\, g \, dx  = \sum_{ \Im \lambda_j > - A } 
e^{ - i \lambda_j t } \int_\RR f\, u_j \, dx\, \int_\RR g\, u_j \, dx  + 
{\mathcal O} ( e^{ - A t } ) , \end{equation}
where $ \lambda_j $ are the poles of the meromorphic 
continuation of $ R ( \lambda ) = ( P - \lambda^2)^{-1} $
(for simplicity assumed to be simple), and $ u_j $ are
solutions to $ ( P - \lambda_j^2 ) u_j = 0 $ satisfying
$ u_j ( x ) = a_{\sgn{x} } e^{ i \lambda | x| } $ for $ | x| \gg 1 $.
Since $ u_j $ are 
not in $ L^2 $ their normalization is a bit subtle: they appear
in the residues of $ R ( \lambda ) $ at $ \lambda_j $.

\begin{figure}
\begin{center}
\includegraphics{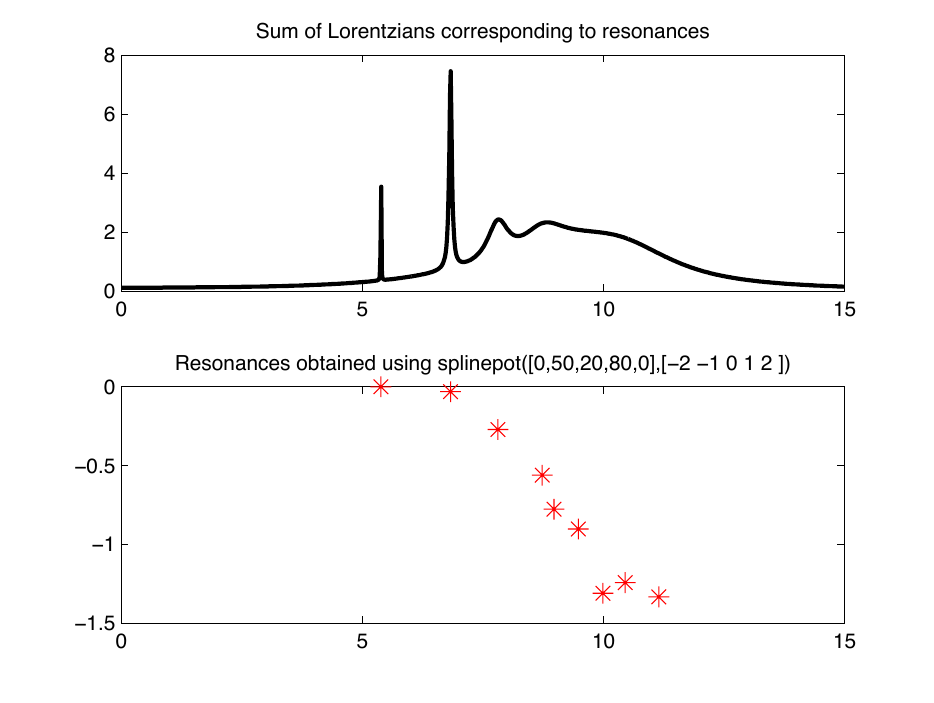}
\end{center}
\caption{The effect of resonances on on the Fourier transform
of correlations as described in \eqref{eq:FTC}. The resonances are
computed using the code {\tt scatpot.m} \cite{BZ}.}
\end{figure}

The expansion \eqref{eq:LP} makes sense since the number
of poles of $ R ( \lambda ) $ with $ \Im \lambda > - A $ is 
finite for any $ A $. If we define $ C ( f , g ) $ to be $ 0 $ for $
t \leq 0 $, the Fourier transform of \eqref{eq:LP} gives (provided $ 0 $ is not a pole of 
$ R ( \lambda ) $), 
\begin{gather}
\label{eq:FTC}
\begin{gathered}
\widehat{ C ( f, g ) } ( - \lambda ) =
  \sum_{ \Im \lambda_j > - A }  \frac{ c_j} { \lambda_j - \lambda } 
+ {\mathcal O} \left(\frac 1  A \right), \ \ \
c_j \defeq  - i  \int_\RR f u_j \, dx \, \int_\RR g u_j \, dx .
\end{gathered}
\end{gather}
The Lorentzians 
\[ \frac{ |\Im \lambda_j |}{  | \lambda - \lambda_j |^2 }  = - 2 \Im \frac 1 {
  \lambda_j - \lambda } , \]
 peak
at $ \lambda = \Re \lambda_j $ and are more pronounced for $ \Im
\lambda_j $ small. This stronger response in the spectrum of
correlations is one of the reasons for calling $ \lambda_j $ (or 
$ \lambda_j^2 $) {\em scattering resonances}. 

In more general situations, to have a finite expansion of type
\eqref{eq:LP}, 
modulo some exponentially decaying error $ {\mathcal O} ( e^{ - \gamma
  t} ) $, we need to know that the number of poles of $ R ( \lambda ) $
is finite in a strip $ \Im \lambda > - \gamma $.  Hence exponential
decay of correlations is closely related to {\em resonance free
  strips}.

This elementary example is related through our approach to 
recent results of Dolgopyat \cite{Dol}, Liverani \cite{Liv}, and Tsujii
\cite{Tsu10},\cite{Ts} on the decay of correlations in classical 
dynamics. 

Let $ X $ be a compact contact manifold of (odd) dimension $
n$, and let $ \gamma_t $ be an Anosov flow on $ X $ preserving 
the contact structure -- see \S \ref{decay} for details. 
The standard example is the geodesic flow on the cosphere 
bundle $ X = S^* M $, where $ ( M ,g ) $ is a smooth negatively curved
Riemannian manifold. Let $ U( t ) : \CI ( X ) \to \CI ( X ) $ 
be defined by $ U ( t ) f = \gamma_t ^*f = f\circ \gamma_t $ and let $ d x $ be 
the measure on $X$ induced by the contact structure and normalized so that 
 $ \vol ( X ) = 1 $. The results
of \cite{Dol},\cite{Liv} show that, for any test functions $f,g\in
C^\infty(X)$, the correlation function satisfies the following
asymptotical behavior for large times:
\begin{equation}
\label{eq:LP1} 
C ( f , g ) ( t ) \defeq 
  \int_X  [U ( t ) f](x)\, g(x) \, dx =   \int_X  f \, dx \,
\int_X  g \, dx + \cO ( e^{ - \Gamma t } ) ,  \ \ t \to \infty \,,
\end{equation}
and the exponent $\Gamma$ is independent of $f,g$.
In other words, the Anosov flow is exponentially mixing with respect
to  the invariant measure $dx$.

From the microlocal point of view of Faure--Sj\"ostrand \cite{f-sj},
this result is related to a resonance free strip for the generator of
the  flow $ \gamma_t $.  The resonances in this setting are called
{\em Pollicott--Ruelle resonances}.

In this paper we consider general semiclassical operators 
modeled on $ P $ given in \eqref{eq:PP}, for which 
the classical flow has a normally hyperbolic trapped set. 
Schr\"odinger operators for which \eqref{eq:NN} holds 
appear in molecular dynamics --- see the recent review \cite{GSWW10}
for an introduction and references. In particular, \cite[Chapter
5]{GSWW10} discusses the resonances in some model cases 
and the relation between the size of the resonance free strip
and the transverse Lyapounov exponents. As reviewed in 
\S \ref{decay}, the setting can be extended such as to include the
generator of the Anosov flow of \eqref{eq:LP1}, namely the operator $ P(h) $ on $ X $
such that $ U ( t ) = \gamma_t^*  = \exp ( - i t P / h ) $.

\subsection{Assumptions and the result}
\label{aatr}
The general result, Theorem \ref{t:1},  is proved for operators modified using a {\em
  complex absorbing potential} (CAP). Results about such 
operators can then be used
for different problems using {\em resolvent gluing techniques}
of 
Datchev--Vasy
\cite{DatVas10_1}
--- 
see Theorems \ref{t:2} and \ref{t:3}. The assumptions on the manifold
$ X $, operator $ P$,
and the complex absorbing potential may seem unduly general, they are justified by the 
broad range of applications.

Let $ X $ be a smooth compact manifold with a density $ dx $ and let
 \[  P = P ( x, h D ) \in \Psi^m ( X ) , \ \ m > 0 , 
 \]
be an unbounded self-adjoint semiclassical pseudodifferential operator 
on $ L^2 (X , dx) $ 
(see \S \ref{10c} and \cite[\S 14.2]{e-z} for background and notations), with 
principal symbol $p(x,\xi)$ independent of $ h $. Let 
\[  W = W ( x, h D ) \in \Psi^k ( X ) 
, \quad \  0 \leq
k \leq m , \quad \ W \geq 0 , \]
be another operator, also self-adjoint and with $h$-independent principal symbol $w(x,\xi)$, 
which we  call a (generalized) complex absorbing potential (CAP). 
We should stress that $ W $ plays a purely {\em auxiliary} role 
and can be chosen quite freely.

If the principal symbols $ p ( x, \xi ) \in S^m ( T^* X )  $ and $ w ( x , \xi ) \in S^k (
T^*X )  $, we assume that, for some fixed $ C_0>0 $ and for any phase space point
$( x, \xi )\in T^*X$,
\begin{equation}
\label{eq:CAPe} 
\begin{gathered}
  |  p ( x, \xi ) - i w ( x , \xi ) | \geq \langle \xi \rangle^m /C_0
- C_0 \,, \qquad
 1 + w ( x, \xi )  \geq \langle \xi \rangle^k / C_0 , \\
 \exp ( t H_p ) ( x, \xi ) \text{ is defined for all $ t \in \RR $.} 
\end{gathered}
\end{equation}
Here, for $ \xi \in T_x^* X $ we have denoted $  \langle \xi \rangle ^2 = 1  +\| \xi
\|^2_x $  for some smoothly varying metric on $ X $, $x \mapsto \| \bullet
\|_x^2 $, and by $ H_p $ the Hamilton vector field of $ p $. The map
$\exp(t H_p):T^*X\to T^*X$ is the corresponding flow at time $t$.
This flow will often be denoted by
$\varphi_t$, the Hamiltonian $p(x,\xi)$ being clear from the context.

For technical reasons (see Lemma~\ref{l:iw}) we will need an additional smoothness assumption on $ w $:
\begin{equation}
\label{eq:aux}   | \partial^\alpha w ( x  , \xi ) | \leq C_\alpha w ( x , \xi )^{
  1 - \gamma }  , \ \  \ 0 < \gamma < \textstyle{\frac12} ,  \end{equation}
when $ w ( x, \xi) \leq 1$. This can be easily arranged and is invariant under
changes of variables.

We call the operator
\begin{equation}
\label{eq:CAP}
\widetilde P  = P - i W \in \Psi^m ( X ), 
\end{equation}
the CAP-modified $ P $. 
The condition \eqref{eq:CAPe} means that the CAP-modified $ P $ is classically elliptic and that for any fixed $
z \in \CC $  
\[ \{ ( x , \xi ) \; : \; \widetilde p(x,\xi) - z = p(x,\xi) - i w(x,\xi) - z  = 0 \}
\Subset T^* X . \]
We define the {\em trapped set} at energy $ E $ as  
\begin{equation}
\label{eq:trapped} K_E \defeq \{ \rho=( x, \xi ) \; : \; \rho \in p^{-1} ( E ) , \ \
\varphi_{\RR} (\rho) \subset  w^{-1} (  0 )  \}\,. 
\end{equation} 
$ K_E $ is compact and consists of  points in $ p^{-1} ( E ) $ which never
reach the {\em damping region}  $\{\rho\in T^*X\,:\, w(\rho) > 0\} $ in backward or forward
propagation by the flow $ \varphi_t$.  

We illustrate this setup  with two simple examples: 

\medskip
\noindent
{\bf Example 1.} Suppose that $ P_0 = -h^2 \Delta + V $, $ V \in \CIc
( \RR^n ; \RR ) $, $ \supp V \Subset B ( 0 , R_0 ) $.  Define the
torus $ X = 
\RR^n / ( 6 R_0 \ZZ )^n $, and $ W \in \CI ( X; [ 0 , \infty ) ) $,  satisfying
\[  
W (x) = 0 , \quad  x \in  B  ( 0 , R_0 ), \qquad  W( x )  = 1 , \quad
 x \in X \setminus B ( 0 , 2 R_0 ) , \quad \partial^{\alpha } W = 
{\mathcal O}_\alpha ( W^{2/3 } ), 
\]
(here we identified the balls in $ \RR^n $ with subsets of the
torus).  The last condition can be arranged by 
taking $ W( x) = \chi ( |x|^2 - R_0^2 )  \psi( x) $
 where $ \chi ( x ) = \exp ( - x^{-1}) \bbbone_{\IR_+}(x) $, and $ \psi 
\in \CI ( X , ( 0 , \infty ) ) $ is suitably chosen. The power of $ W
$ on the right hand side can be any number greater than $ \frac12$.

Because of the support properties of $ V $, $ P \defeq -h^2 \Delta
+ V \in \Psi^2 ( X ) $ and $ P - i W $ satisfy all the properties 
above. The trapped set $ K_E $ can be identified with a subset
of $ T^*_{ B( 0 , R_0) } \RR^n $ and is then equal to the trapped set of
scattering theory:
\[  
K_E = \{ ( x , \xi ) \in T^* \RR^n \; : \;  \xi^2 + V ( x ) = E, \
\  x(t) \not \to \infty , \ \ t \to \pm \infty\}\,.
\]

\begin{rem}
\label{rem:NHIM}Normally hyperbolic trapped sets occur 
in the semiclassical theory of chemical reaction dynamics, where they are usually called
Normally Hyperbolic Invariant Manifolds (NHIM). They are of fundamental importance to quantitatively understand the kinetics of the chemical reaction. See for instance \cite{uzer-etal02} for a description of the classical phase
space structure, and \cite{GSWW10} and references given there
for the adaptation to the quantum framework. 
The focus there is on examples for which the Hamiltonian flow 
exhibits a 
\[ \text{saddle $\times$ saddle $\times $ \ldots $\times$ center \ldots 
$\times$ center} \]
fixed point: after an appropriate linear symplectic change of coordinates, the quadratic expansion of the Hamiltonian $p(x,\xi)$ near the fixed point (set at the origin) reads as:  
$$
p_{\rm{quad}}(x,\xi)= \frac{1}{2}
\sum_{i=1}^{d- d_\perp }(\xi_i^2 + \omega_i^2x_i^2) + \sum_{i= 
d-d_\perp+1}^d \frac{1}{2}(\xi_i^2 - \lambda_i^2 x_i^2)\,.
$$
For this quadradic model the NHIM at a positive energy $E>0$\footnote{For the distribution of resonances at the fixed point energy
$ E= 0 $ see
\cite{KaKe} and \cite{SjLNM}.}, is given by 
\[ p^{-1}(E)\cap\{\xi_{d-d_\perp+1}=x_{d-d_\perp+1}= \ldots = x_{d} = \xi_{d} = 0 \} \]
which is a $2d-2d_\perp -1$-dimensional sphere. 
The stable/unstable distributions are $d_\perp$-dimensional (see 
\eqref{eq:NH} below), and are generated by the vectors $\{{\partial}/{\partial \xi_i}\pm \lambda_i {\partial}/{\partial x_i}\}_{i=d-d_\perp +1 }^{d} $.
For this quadratic model the flow along the NHIM is completely integrable. This implies that the latter is structurally stable to perturbations (it is then $r$-normally hyperbolic for any $r\in\IN$), meaning that for any given regularity $r>0$, a small enough perturbation of $p_{\rm{quad}}$ will still lead to the presence of a NHIM of regularity $C^r$ 
\cite{HPS}. 
However, the flow on the perturbed NHIM is generally not integrable.
This situation occurs if one considers the full Hamiltonian $p$ with quadratic expansion $p_{\rm{quad}}$: for small positive energies $p$ will still exhibit a NHIM,
which is a deformed sphere. 

Physical systems featuring this type of fixed point are presented in the literature: for instance 
the isomerization of hydrogen cyanide \cite{WBW04} or the
 quantum dynamics of the nitrogen-nitrogen exchange \cite{GSWW10}. 
Strictly speaking the potentials appearing in these physical
models are more complicated than the ones allowed here. However,
the behaviour near the NHIM determines the
phenomena which are studied here and which are relevant in physics.

We conclude this remark by recalling that when $ d_\perp =1 $
(most relevant from the point of view of \cite{GSWW10})
and when the system is $r$-normally hyperbolic for sufficiently 
large $ r $ very precise results on the distribution of 
resonances have been obtained by 
Dyatlov \cite{Dya1},\cite{Dya2}. 
\end{rem}

\medskip

\noindent
{\bf Example 2.} Suppose that $ X $ is a compact manifold with a
volume form $ dx $ and a vector field $ \Xi $ generating 
a volume preserving flow ($ {\mathcal L}_\Xi dx = 0 $).  Then 
$ P = -i h \Xi $ is a selfadjoint operator on $ L^2 ( X, dx ) $, and
the corresponding propagator $\exp ( - i t P / h )$
is the push-forward of the flow
$\gamma_t=\exp ( t \Xi )$ generated by $\Xi$ on functions $f\in L^2(X,dx)$:  
$\exp ( - i t P / h )f = f\circ \gamma_{-t}$.

To define the CAP in this setting we choose a Riemannian
metric $ g $ on $ X $, and a function
\begin{equation}
\label{eq:fCAP}  \begin{gathered}
f \in \CI ( \RR , [0, \infty ) ),  \ \ \  
 | f^{(k)} ( s ) | \leq C_k f( s )^{1-\gamma} , \ \text{ for some
   $\gamma\in (0,1/2)$} , \\
f^{-1}(0) =
[-\infty,M] \ \text{ for some $M>0$,} \ \ \  f ( s ) = \sqrt s , \ \
s > 2M .
\end{gathered}
\end{equation}
If $ \Delta_g $ is the corresponding
Laplacian on $X$, we set $ W ( x, h D ) = f ( - h^2 \Delta_g ) $.

Then the operator $ P - i W \in \Psi^1 ( X ) $ satisfies the
assumptions above. The principal symbols read $p(x,\xi)=\xi ( \Xi_x)$,
$w(x,\xi)=f(\|\xi\|_x^2)$, where the norm $\|\bullet\|_{x}$ is associated
with the metric $g$.

At a given energy $E\in\IR$, the trapped set is given by the points which never enter the 
absorbing region:
\[ K_E = \{ ( x, \xi ) \in T^* X \; : \; \xi ( \Xi_x ) = E , \ \ 
\| (\gamma_{-t})_* \xi \|_g \leq M , \ \ \forall t \in \RR \}. \]
At this stage the trapped set seems to depend on the choice of $M$. Below we
will be concerned with $\exp(t\Xi)$ being an Anosov flow, in which
case this explicit dependence will disappear, as long as we choose $M$
large enough compared with the energy $E$ (see the second assumption~\eqref{eq:trapped1} below).

\medskip

Returning to general considerations we also define 
\begin{equation}
\label{eq:Kdel}  K^\delta \defeq \bigcup_{ |E|\leq \delta } K_E , \end{equation}
which is a compact subset to $ T^*X $ and 
assume that 
\begin{equation}
\label{eq:trapped1} 
 dp \rest_{ K^\delta } \neq 0 , \qquad
K^\delta \cap \WF_h ( W ) =  \emptyset . 
\end{equation}
The first assumption implies that for $ | E | \leq
\delta $, the energy shell $ p^{-1} ( E ) $  is a smooth hypersurface close to $ w^{-1}( 0 )$. The 
second assumption is consistent with the definition \eqref{eq:trapped} of $K_E$. It implies that
the latter is contained in the interior of the region $w^{-1}(0)$, a property which is stable when enlarging 
$K_E$ to $K^\delta$, or when slightly modifying the support of $w$.

We now make the following {\em normal hyperbolicity} assumption on $
K^\delta $:
\begin{equation}
\label{eq:NH0}
\text{ $K^\delta $ is a smooth symplectic submanifold of $ T^*X $,}
\end{equation}
and there exists a continuous distribution of
linear subspaces 
\[ K^\delta \ni \rho \longmapsto  E^{\pm}_\rho
\subset T_\rho  ( T^* X ) , \]
invariant
under the flow, 
\[ \forall t \in \RR,\quad  (\varphi_t)_* E^\pm_{ \rho } = E^\pm_{ \varphi_t ( \rho ) }  ,
\]
and satisfying, for some $\lambda>0$, $C>0$ and any point $\rho\in K^\delta$,
\begin{gather}
\label{eq:NH}
\begin{gathered}
 T_\rho K^\delta \cap E^\pm_\rho =  E^+_\rho
\cap E^-_\rho = \{ 0 \} \,, \quad \dim E^\pm_\rho = d_{\perp}  , 
\ \ \ 
T_{\rho}(T^{*}X) = T_{\rho}K^{\delta} \oplus  E_{\rho}^{+} \oplus 
E_{\rho}^{-} , \\
\forall v \in E^\pm_\rho,\ \ \forall t>0,\quad \| d\varphi_{\mp t}(\rho) v \|_{\varphi_{\mp t} ( \rho ) }  \leq C
e^{ - \lambda t } \| v \|_\rho \,.
\end{gathered}
\end{gather}
Here $ \rho \mapsto  \| \bullet \|_\rho $ is any smoothly varying norm on $ T_\rho ( T^*X ) $, $ \rho
\in K^\delta $. The choice of norm may affect
$ C $ but not $ \lambda $. 

\begin{rem} \label{rem:regularity}
A large class of examples for which the distributions $ 
\rho \mapsto E_\rho^\pm $ are not smooth is provided by 
considering contact Anosov flows on compact manifolds
--- see \cite{f-sj},\cite{Ts} and \S \ref{gest} below for the 
natural appearance of normally hyperbolic trapping for 
the flow lifted to the cotangent bundle of the 
manifold. The regularity is inherited from the regularity 
of the stable and unstable distributions tangent to the manifold, which in general are only 
known to be H\"older continuous \cite{An67}. 
More is known on the regularity of these distributions when the manifold is 3-dimensional (and preserves a contact structure). 
In this situation, Hurder-Katok showed \cite{HuKa} that there is a dichotomy (or ``rigidity"): 
either the stable/unstable distributions are $C^{2-\eps}$ for any $\eps>0$ but not $C^2$ 
(this is due to a certain obstruction, namely the {\em Anosov cocycle} is not cohomologous to zero), or the distributions are as smooth as the flow. 
If that 3-dimensional flow is the geodesic flow on a surface of negative curvature, then following Ghys \cite{Gh} they show (Corollary.~3.7) that the latter case imposes a metric of {\em constant} negative curvature. Hence, for the geodesic flow on a surface of {\em nonconstant} negative curvature, the stable/unstable distributions, and hence their lifts $E_\rho^\pm $, are not $C^2$.

We do not know
of an example of a Schr\"odingier operator (that
is of a classical Hamiltonian of the form 
$ p ( x, \xi) = |\xi|^2 + V(x)$) 
for which the trapped set is smooth --- or 
sufficiently regular: as with all microlocal results
a certain high level of regularity, depending on the 
dimension, is sufficient --- and the distributions 
$ \rho \mapsto E_\rho^\pm $ are irregular. However there is 
no general result which prevents that possibility.
Interesting regular examples of $ E_\rho^\pm $ of any dimension 
$ 1 \leq d_\perp \leq d-1 $ were discussed in Remark~\ref{rem:NHIM}.

We also remark that higher dimensional distributions can lead to complicated topological issues, which 
would make the global approach of \cite{Dya1},\cite{Dya2},\cite{WZ}
difficult. 
This is visible already for flows on constant curvature manifolds
for which smooth foliations may have nontrivial topology 
\cite[\S 2.2]{DFG}. 

Except for the construction of the
escape function, for which we need to use \cite{NSZ2} and \cite{SZ10},
the analysis in \S\S \ref{ats} and \ref{micr} would  not be simplified
by a smoothness assumption on the distributions.
\end{rem} 

We can now state our main result.

\begin{thm}
\label{t:1}
Suppose that $ X $ is a smooth compact manifold and that
 $ P $ and $ W $ satisfy the assumptions above. If the 
trapped set $ K^{\delta } $ given by
\eqref{eq:trapped},\eqref{eq:Kdel}
is normally hyperbolic, in the sense that \eqref{eq:NH0} and
\eqref{eq:NH} hold, then for any $ \epsilon_0 > 0 $ there exists $ h_0
$, $ c_0 $, 
$ C_1 $, 
such that for $ 0 < h < h_0 $, 
\begin{gather}
\label{eq:t1} 
\begin{gathered}   \|  ( P - i W - z )^{-1} \|_{L^2 \to L^2 } \leq C_1
  h^{-1  + c_0 \Im z / h } \log (1/h) ,\\ 
\text{ for } \ \  z \in [ - \delta + \epsilon_0 , \delta - \epsilon_0 ] - 
i h [ 0 , \lambda_0/ 2 - \epsilon_0 ] , 
\end{gathered}
\end{gather}
 where $\lambda_0>0$ is the minimal transverse unstable expanding rate:
\begin{equation}
\label{eq:t11} 
\lambda_0 \defeq \liminf_{ t \to \infty } \frac 1 t \inf_{ \rho \in K^\delta } \log
\det\left(d\varphi_t\rest_{ E^+_ \rho } \right)\, . 
\end{equation}
Here $ \det $ is taken using any fixed volume form on $ E^+_ \rho$,
the value of $\lambda_0>0$ being independent of the choice of volume forms. 
\end{thm}

This theorem will be proved in \S \ref{micr} after preparation in \S\S
\ref{cld},\ref{ats}. The bound $ \log (1/h) /h $ 
on the real axis is optimal as shown in \cite{BBR}.
Using the methods of \cite{DatVas10_1} 
the estimate \eqref{eq:t1} almost immediately
applies to the setting of scattering theory. As an example 
we present an application to scattering on asymptotically 
hyperbolic manifolds, which will be proved in \S \ref{recap}:
\begin{thm}
\label{t:2} 
Suppose $ ( Y , g ) $ is a conformally compact $n$-manifold with even 
power metric: $ Y $ is compact, $ \partial Y = \{ x = 0 \} $, $
dx\rest_{\partial Y } \neq 0 $, $ g = ( dx^2 + h ) /x^2 $ where $ h $
is a smooth 2-tensor on $ Y $ with only even powers of $ x $ appearing
in its Taylor expansion at $ x = 0 $.  
If the trapped set for the geodesic flow on $ Y $ is
normally hyperbolic, then the following resolvent estimate holds:
\[   \| x^{ k_0}  \big( - \Delta_g - ( n - 1)^2/4 - \lambda^2  \pm i 0 \big)^{-1} x^{k_0} \|_{ L^2
  \to L^2 } \leq C_0 \frac {\log \lambda } \lambda, \quad \lambda > 1
. \]
\end{thm}

The next application is a rephrasing of a recent theorem of Tsujii
\cite{Tsu10,Ts}; it will be proved in \S\ref{decay}. We take the point of view
of Faure--Sj\"ostrand \cite{f-sj}, see also
\cite{dadyz}.

\begin{thm}
\label{t:3}
Suppose $ X $ is a compact manifold and $ \gamma_t : X \to X $
a contact Anosov flow on $ X$. Let $ \Xi $ be the vector field generating
$\gamma_t $, and $ P = - i h \Xi $ the corresponding semiclassical
operator, self-adjoint on $L^2(X,dx)$ for $dx$ the volume form
derived from the contact structure. 

Define the minimal asymptotic unstable expansion rate
\begin{equation}
\label{eq:gammd}  \lambda_0 \defeq \liminf_{ t \to \infty } \frac 1 t \inf_{ x \in X  } \log
\det \left( d\gamma_t \rest_{ E_u ( x ) } \right) ,
\end{equation}
with $ E_u ( x ) \subset T_x X $ the unstable subspace of the flow at
$ x $.

For any $t>0$ there exists a Hilbert space, $H_{t\cG}$  (see \eqref{eq:defHG}), 
\[     \CI ( X ) \subset H_{t\cG} ( X ) \subset {\mathcal D}'  ( X)\,,  \]
such that $ ( P - z )^{-1} : H_{t\cG} \to H_{t\cG } $ is meromorphic in the half-space
$ \{ \Im z > - t h\}$. 

Then for any small $ \epsilon_0,\delta > 0 $,  there exist  $ h_0, c_0>0$ and $
C_1 >0$ such that, taking any $t>\lambda_0/2$ and any $ 0 < h < h_0 $,
\begin{gather}
\label{eq:t3}
\begin{gathered}   \|  ( P - z )^{-1}  \|_{ H_{t\cG} \to H_{t\cG} }  \leq C_1
  h^{-1   + c_0 \Im z / h } \log (1/h) ,  \qquad z \in  [ \delta , \delta^{-1}   ] -  i h [ 0 , \lambda_0/ 2 -
 \epsilon_0 ]\,.
\end{gathered}
\end{gather}
\end{thm}
The Hilbert space $H_{t\cG}$ in the above theorem is not optimal
as far as sharp resolvent estimates are concerned\footnote{We are grateful to
Fr\'ed\'eric Faure for this remark.}. It is obtained by
applying a microlocal weight $e^{t\cG^w}$ on $L^2$, with a function
$\cG(x,\xi)$ vanishing in a fixed neighbourhood of the trapped set.
In \cite{Tsu10} Tsujii constructed Hilbert spaces $B^\beta$ leading to resolvent
estimates $\|( P - z )^{-1}\|_{B^\beta}\leq C_1\,h^{-1}$ in the same
region. A similar resolvent estimate could be obtained in our
framework, by further modifying $H_{t\cG}$ using the ``sharp'' escape function $G$ presented
in \S\ref{out} (see the estimate \eqref{e:resolvent-tPG}).

Under a pinching condition on the Lyapunov exponents, the recent 
results announced by Faure--Tsujii \cite{f-ts2} provide a much more
precise description of the spectrum of $P=-ih\Xi$ on $H_{t\cG}$: 
the Ruelle--Pollicott resonances are localized in horizontal strips
below the real axis, and the number of resonances in each strip
satisfies a Weyl's law asymptotics. That is analogous to the result
proved by 
Dyatlov \cite{Dya1}, which was motivated by quasinormal modes for black holes.

Theorems \ref{t:2} and \ref{t:3} have applications to the decay of 
correlations, respectively for the wave equation and for contact Anosov flows. As an
example we state a refinement of the decay of correlation result \eqref{eq:LP1}
of Dolgopyat \cite{Dol} and Liverani \cite{Liv}. 
\begin{cor}
\label{t:4}
Suppose that $ \gamma_t : X \to X $ is a contact Anosov flow on a 
compact manifold $ X $ (see \S \ref{gest} for the definitions) and 
that $ \lambda_0 $ is given by \eqref{eq:gammd}. 

Then there exist a sequence of complex numbers, $ \mu_j $,  
\[  0 > \Im \mu_{j} \geq \Im \mu_{j+1} , \]
and of distributions $ u_{j,k}, v_{j,k} \in {\mathscr D}' ( X ) $, $ 0 \leq k
\leq K_j $, 
such that, for any $ \epsilon_0 > 0 $, there exists $J(\eps_0)\in\IN$ such that  
for any $ f , g \in \CI ( X ) $, 
\begin{equation}
\label{eq:t4}
\begin{split}
& \int_X f(x)  \,  \gamma_t^* 
g
 ( 
x ) 
\,  d x = \\
& \ \ \ \ \ \int_X f   dx \, \int_X g
  dx + \sum_{ j=1}^{ J ( \epsilon_0 ) }  \sum_{ k=1}^{K_j } t^{k}
e^{ - i t \mu_{j} } u_{j,k} ( f ) v_{j,k} ( g ) + \cO_{f,g} ( e^{ - t(
  \lambda_0 - \epsilon_0 ) /2 } ), 
\end{split} 
\end{equation}
for $ t > 0 $. 
Here $ dx $ is the measure on $ X $ induced by the contact form and
normalized so that $ \vol ( X) = 1 $, and $ u ( f ) $, $ u \in
{\mathcal D'} ( X) $, $ f \in \CI ( X) $ denotes the distributional pairing.
\end{cor}
The exponential mixing estimate \eqref{eq:t4} has
been obtained by Tsujii \cite[Corollary 1.2]{Tsu10} in the more general
case of contact Anosov flows of regularity $C^r$.
We restate it here to stress its analogy with 
resonance expansions in wave scattering, see for instance \cite{TZ2}.

For information about microlocal structure of the distributions 
$ u_{j,k} $ and $ v_{j,k}$ the reader should consult \cite{f-sj}. Here
we only mention that (with the standard wave front set of \cite{Hor2})
\[   
\WF ( u_{j,k} ) \subset E_s^* , \quad  \WF ( v_{j,k} ) \subset
E_u^* , 
\]
where $ E_\bullet^* = \bigcup_{x\in X} E_\bullet^* ( x ) $, and $ 
E_\bullet ^*  ( x)\subset T_x^* X  $ is the annihilator of $ \IR\Xi_x  + E_\bullet ( x
) \subset T_x X $, $ \bullet = u,s$. The spaces $ E_\bullet ( x ) $
appear in the Anosov decomposition of the tangent space \eqref{eq:Ande}

\section{Outline of the proof of Theorem \ref{t:1}}
\label{out}

The proof proceeds via the analysis of the propagator for the operator
\[  
\tP_{  G } \defeq e^{ - G^w ( x, h D ) }  ( P - i W ) e^{ G^w ( x, h D )} \, , 
\]
where the function $ G(x,\xi;h) $ belongs to a certain exotic class of symbols.
Our $ G $ is closely related to the escape
function constructed in \cite{NSZ2}, it depends on an additional
small parameter, $ \tilde h $, which will be chosen
independently of $ h$.

For a large $ t_0 $, any fixed $ \Gamma >0 $ and
$\eps>0$,
we can construct $ G $ so that, for some constant $ C_0$, the following
holds uniformly in $ 0 < h < h_0$, $ 0 < \tilde h < \tilde h_0 $:
\begin{equation}
\label{eq:G} 
\begin{split} 
&  G ( \rho ) = \cO ( \log ( 1/ h ) ) , 
\quad    G (\rho ) - G ( \varphi_{ - t_0} ( \rho ) ) \geq - C_0 ,
  \quad \rho \in T^* X , 
  \\
& G ( \rho ) - G ( \varphi_{ - t_0} ( \rho ) )  \geq 2 \Gamma , \quad \rho \in  
p^{-1} ( [ - \delta, \delta ] ), \quad d ( \rho, K^\delta ) > (h/\tilde 
h )^{\frac12} , \ \ w ( \rho ) < \epsilon \,,
\end{split}
\end{equation}
where $d(\bullet,\bullet)$ is any given distance function in $T^*X$.

The proof of Theorem \ref{t:1} is based on the following
estimate. For some $ \epsilon_1 > 0 $, take an operator $ A \in \Psi^0
( X ) $ such that $ \WFh ( A ) \subset   p^{ -1} ( ( - \delta, \delta
) ) \cap w^{-1} ([ 0 , \epsilon_1 ) ) $. 
We will prove the following norm estimate: for any $ \epsilon_0 $ 
and $ M $ there 
exists $ M_{\epsilon_0 }$ and $ \th_0>0,\ h_0 > 0$ such that for any $\th<\th_0$, $h<h_0$, we have the estimate
\begin{gather}
\label{eq:main}
\begin{gathered}
\| \exp ( - i t \tP_G /h ) A \|_{L^2 ( X ) \to L^2 ( X ) } \leq 
 e^{ - t ( \lambda_0 - \epsilon_0) / 2 } ,\\
\text{uniformly for times} \ \   M_{\epsilon_0} \log \frac 1 {\th } \leq t \leq 
\max(M, M_{\epsilon_0} ) \log \frac 1 \th .
\end{gathered}
\end{gather}
%
%

As a result,  for $ \Im z >  -  ( \lambda_0 - 2 \epsilon_0 )/ 2 $, 
\begin{equation}
\label{eq:intt} \begin{split}  ( \tP_G - z ) \frac i h \int_0^T e^{ - i t ( \tP_G - z )
    / h } A dt &  = 
( I - e^{ - i T( \tP_G  - z ) / h }) A   =  A - {\mathcal O} ( e^{ - T  \epsilon_0 } ) _{L^2 \to
  L^2} \,.
\end{split}\end{equation}
Hence, by taking $ T $ large enough and using the ellipticity
of $ \tP_G - z $ away from $p^{ -1} ( ( - \delta, \delta
) ) \cap w^{-1} ([ 0 , \epsilon_1 ) )$, we obtain
\be\label{e:resolvent-tPG}
( \tP_G - z)^{-1} ={\mathcal O}  ( h^{-1} ) _{L^2 \to L^2} , \quad 
 \Im z >  -  ( \lambda_0 - 2 \epsilon_0 )/ 2  \,. 
\ee
Since $ e^{ \pm G^w } = {\mathcal O} ( h^{-M_0}  ) _{L^2 \to L^2}  $ from
the growth condition on $ G $, a polynomial bound for $ ( P - i W - z
)^{-1} $ follows. The more precise bound \eqref{eq:t1} follows from
a semiclassical maximum principle.

\bigskip

To prove the estimate \eqref{eq:main} we proceed in a number of steps:

\medskip

\noindent
{\bf Step 1.} The most delicate part of the argument concerns the
evolution near the trapped set. 
For some fixed $R>1$, we introduce a cut-off function $ \chi \in \tS_{\frac12} $ supported in 
the set 
\[ \{ \rho \in p^{ -1} ( ( - \delta, \delta ) ) \; : \; 
d ( \rho, K^\delta ) \leq 2 R ( h /\tilde h)^{\frac12} \} \,.
\]
This cut-off is quantized into an operator $ \chi^w \defeq \chi^w ( x, h D ) $.

We then claim that 
for any $ \epsilon_0 > 0 $  and $ M > 0 $, 
there exists  $ C>0 $ such that, for $\th<\th_0$ and $h<h_0(\th)$, 
\begin{gather}
\label{eq:tricky}
\begin{gathered}
\| \chi^w  e^{ - i t P / h } \chi^w  \|_{ L^2 \to L^2 } \leq C \tilde h^{ -d_{\perp}/2} e^{
-  t ( \lambda_0 - \epsilon_0 / 2 ) / 2 }  
, \\ 
\text{uniformly for $ 0 \leq t \leq M \log \frac 1 {\tilde h } $. }
\end{gathered}
\end{gather}
The proof of this bound is provided in \S \ref{ats}.

\medskip
\noindent
{\bf Step 2.}  For the weighted operator we obtain an 
improved estimate, now with a fixed large time $  t_0 $ related to the 
construction of $ G $,  and for $ \chi $ which in addition 
satisfies
\[  
\chi ( \rho ) = 1 \  \text{ for  $ d ( \rho , K^\delta ) \leq R  (
  h /\tilde h)^{\frac12} $, $ | p(\rho) | \leq \delta /2 $.}  
\]
Using Egorov's theorem and (from \eqref{eq:G}) the positivity of $G-G\circ\varphi_{-t_0}$
on the set $\supp(1-\chi)\cap \WFh (A) $, we get following the weighted estimate:
\begin{equation}
\label{eq:wei}  \| ( 1 -\chi^w  ) e^{ - i t_0 \tP_G  / h }  A \| \leq  e^{- \Gamma } , 
\end{equation}
When constructing the function $G$ it is essential to choose $\Gamma$ such that
\[ \Gamma > \frac { t_0 \lambda_0 } 2 . \]
We also show that 
\begin{equation}
\label{eq:wei0} 
\|  e^{ - i t_0 \tP_G / h } A \| \leq e^{2  C_0 } , 
\end{equation}
for a constant $ C_0 $ independent of $h,\th$.
Formally, these results follow from Egorov's theorem but 
care is needed as $ G $ is a symbol in an exotic class.
To obtain \eqref{eq:wei} and \eqref{eq:wei0} we proceed
as in the proof of \cite[Proposition 3.11]{NSZ2}.  This is 
done in \S \ref{micr}.

\medskip
\noindent
{\bf Step 3.} The last step combines the two previous estimates,
by decomposing
\begin{gather*}   e^{ - i  n t_0 \tP_G / h } = (  U_{G, +} + U_{G,-} )^{n}, \\ 
U_{G,+} \defeq e^{ - i t_0 \tP_G /h  } \chi^w  , \ \ 
U_{G,-} \defeq  e^{ - i t_0 \tP_G /h  } ( 1 -\chi^w ) .
\end{gather*}
In order to apply \eqref{eq:tricky} we use the fact that
\[  \chi^w e^{-G^w } e^{ - i t (P - i W )/h  } e^{ G^w } \chi^w  =
\chi_{G,1}^w e^{ -i t P / h} \chi_{G,2}^w + \cO (\th^\infty ) +  \cO (h^\frac 12) \,, 
\]
where the symbols $ \chi_{G,i} $ have the properties required in \eqref{eq:tricky}.
A clever expansion of $ ( e^{ in t_0 \tP_G/ h} )^n $ into terms involving 
$ U_{G, \pm } $ and an application of Steps 1 and 2 lead to the estimate
\eqref{eq:main}
for $ t = n t_0 $.
The argument is presented in \S \ref{pr}. 

\section{Preliminaries}
\label{prel}

In this section we will briefly recall basic concepts of semiclassical 
quantization on manifolds with detailed references to previous papers.

\subsection{Semiclassical quantization}
\label{10c}
The semiclassical pseudodifferential operators on a compact
manifold $ X$ are quantizations of functions belonging to 
the symbol classes $ S^m  $ modeled on symbol
classes for $ \RR^n $:
\[ \begin{split} 
S^{m} ( T^*\RR^n ) =\big\{ & a \in \CI( T^* \RR^n \times (0, 1]_h
  ) :\ \forall \alpha,\beta\in\NN^n,
|\partial_x ^{ \alpha } \partial _\xi^\beta a ( x, \xi ;h ) | \leq
C_{\alpha \beta} ( 1 + |\xi|)^{m -|\beta| } \big\}  \,, 
\end{split} \]
see \cite[\S 14.2.3]{e-z}. The Weyl quantization, 
which we informally write as
\[   S^m ( T^* X )  \ni a ( x, \xi ) \longmapsto a^w ( x , h D) 
\in \Psi^m ( X ), \]
maps symbols to pseudodifferential operators.
It is modeled
on the quantization on $ \RR^n $: 
\be\label{eq:weyl}
\begin{split}
[a^w u](x) & = a^w(x,hD) u(x)  =  [\Op (a) u ] ( x) \\
& \defeq \frac1{ ( 2 \pi h )^d } 
  \int \int  a \Big( \frac{x + y  }{2}  , \xi \Big) 
e^{ i \la x -  y, \xi \ra / h } u ( y ) dy d \xi \,,  \ \ u \in
\cS  ( \RR^n ) . 
\end{split} 
\ee

 The symbol map
\[ \sigma : \Psi^m ( X 
)   \to S^m ( T^* X ) / h
S^{m-1} ( T^* X ) , \]
is well defined as an equivalence class and its kernel is
$ h \Psi^{m-1} ( X) $ -- see 
\cite[Theorem 14.3]{e-z}. If $ \sigma ( A ) $ 
has a representative independent of $ h $ we call that
invariantly defined element of $ S^m ( T^* X ) $ the 
{\em principal symbol} of $A$.

Following \cite{Da-Dy} we define the class of 
{\em compactly microlocalized operators}
\[ 
\Psi^{\rm{comp}} ( X) \defeq \{ 
a^w ( x , h D ) : a \in ( S^0 \cap \CIc)  ( T^* X )  \} + h^\infty
\Psi^{-\infty } ( X ) .
\]
These operators have well defined semiclassical wave front sets:
\[   \Psi^{\rm{comp}} ( X)  \ni A \longmapsto 
\WFh ( A ) \Subset  T^* X ,\]
see \cite[\S 3.1]{Da-Dy} and \cite[\S 8.4]{e-z}.

Let $ u=u(h) $, $ \| u ( h ) \|_{ L^2} = {\mathcal O} ( h^{-N} ) $
(for some fixed $ N$) be a wavefunction microlocalized in 
a compact set in $ T^* X $, in the sense that for some
$ A \in \Psi^{\rm{comp}}  $, one has $ u = A u + 
\Oo_{\CI } ( h^\infty ) $. The semiclassical wavefront set of $u$ is
then defined as:
\be\label{eq:defWF}
\WFh ( u ) =   \complement \big\{  \rho \in  T^*X 
\; : \; \exists \, a \in S^0 ( T^* X ) \,, \  \ a ( x, \xi ) =1
\,, \ \| a^w \, u \|_{L^2} = \Oo ( h^\infty) \big\}
 \,.
\ee
When $ A \in \Psi^\comp ( X )$ we also define 
\[  \WFh ( I - A ) := \bigcup_{ B \in \Psi^\comp ( X )} \WFh ( B ( I - A ) ), \]
and note that $ \WFh ( B ) \cap \WFh ( A ) $ is defined for any $ 
B \in \Psi^m ( X ) $ as 
\[  \WFh ( B ) \cap \WFh ( A ) := \WFh ( C B ) \cap \WFh ( A ) , \ \
C \in \Psi^\comp, \ \ \WFh ( I - C ) \cap \WFh ( A ) = \emptyset. \]

Semiclassical Sobolev spaces, $ H^s_h ( X) $ are defined using 
the norms
\begin{gather}
\label{eq:ssn}
\begin{gathered}
\| u \|_{H^s_h (X ) } = \| ( \Id - h^2 \Delta_{g} )^{s/2} u
\|_{L^2 ( X) } \,,  
\end{gathered}
\end{gather}
for some choice of Riemannian metric $ g $ on $X$ 
(notice that $ H^s_h ( X) $ represents the same vector space as the usual Sobolev space $H^s(X)$).

\subsection{\boldmath$S_{\frac12}$ calculus with two parameter}
\label{12c}
Another  standard space of symbols  $ S_{\delta} ( \RR^{2n}) $, $0<\delta\leq 1/2$, 
is defined  by demanding that $ \partial^\alpha a = { \mathcal O} (
h^{-|\alpha|\delta} ) $. The quantization procedure $ a \mapsto \Op a $ 
gives well defined operators and $ \Op a \circ \Op b = \Op c $
with $ c \in S_{\delta} $. 

For $0<\delta<1/2$ we still have a pseudodifferential calculus, with asymptotic  expansions in powers of $h$.
However, for $\delta=1/2$ we are at the border of
the uncertaintly principle, and there is no asymptotic calculus - 
see \cite[\S 4.4.1]{e-z}.
To obtain an asymptotic calculus the standard $ S_{\frac12} $ 
spaces is replaced by a symbol space where a second 
asymptotic parameter is introduced:
\[ \widetilde S_{\frac12} (  \RR^{2n} ) \defeq 
\big\{ a = a ( \rho, h , \tilde h ) \in \CI ( \RR^{2n}_\rho \times (0, 1]_{h}
\times ( 0, 1]_{\tilde h}  ) : |\partial_\rho^\alpha a | \leq C_\alpha
( h/\tilde h)^{-|\alpha|/2 } \big\}. \]
Then the quantization $ a \mapsto a^w ( x, h D) \in \widetilde
\Psi_{\frac12} ( \RR^n ) $ is 
unitarily equivalent to 
\begin{equation}
\label{eq:stre}  \tilde a \mapsto \tilde a ^w ( \tilde x , 
\tilde h D ) = \Opt(\ta) , \ \ \  \tilde a ( \rho ) = a ( (h/\tilde
h)^{\frac12} \rho ) . \ \ \ \tilde a \in S(\IR^{2n}) , 
\end{equation} -- see \cite[\S\S
4.1.1,4.7.2]{e-z}. Hence, we now have expansions
in powers of $ \th $, as in the standard calculus, 
with better properties (powers of $(h\th)^{\frac 12}$) when operators in $ \widetilde \Psi_{\frac12} $
and $ \Psi $ are composed -- see \cite[Lemma 3.6]{SZ10}. 

For the case of manifolds we refer to 
\cite[\S 5.1]{Da-Dy}  which generalizes and clarifies
the presentations in \cite[\S 3.3]{SZ10} and \cite[\S 3.2]{WuZ}. 
The basic space of symbols, and the only one needed here, is
\[ \begin{split} \widetilde S_{\frac12}^{\rm{comp} } ( T^*X )  = & \big\{ a \in \CIc ( T^*X ) : 
V_1 \cdots V_k a = {\mathcal O} ( (h/\tilde h)^{-\frac k 2} ) , \ \  \forall \; k ,\\
&
V_j \in \CI ( T^*X , T ( T^* X ) ) \big \} + h^\infty S^{-\infty } ( T^* X )
. 
\end{split} \]
The quantization procedure
\[   \widetilde S_{\frac12}^{\rm{comp} } ( T^*X ) \ni a \to \Op ( a ) \in \widetilde \Psi_{\frac12}^{\rm{comp} } ( X  )
 \]
defines the class of operators $\widetilde \Psi_{\frac12}^{\rm{comp} } ( X ) $ modulo $h^\infty\Psi^{-\infty}(X)$,
and the symbol map:
\begin{equation}
\label{eq:wisi}  \widetilde \sigma : 
\widetilde \Psi_
{\frac12}^{\rm{comp} } ( X  ) \longrightarrow 
 \widetilde S_{\frac12}^{\rm{comp} } ( T^*X ) / h^{\frac12} \tilde
 h^{\frac12}  \widetilde S_{\frac12}^{\rm{comp} } ( T^*X )  . 
\end{equation}
The properties of the resulting calculus are listed in 
\cite[Lemma 5.1]{Da-Dy} and we will refer to those results later on.

When $ \tilde h = 1 $ we use the notation 
$ S^{\comp}_{\frac12} ( T^* X ) $ for symbols and denote by 
$ \Psi^{\rm{comp}}_{\frac12} ( X ) $ 
the corresponding class of pseudodifferential 
operators.
 The symbol map 
\[  \sigma : \Psi^{\comp}_{\frac12} ( X ) \longrightarrow S_{\frac12}^{\comp}
( T^*X ) / h^{\frac12} S^{\comp}_{\frac12} ( T^*X ) , \]
is still well defined but the operators in this class do not enjoy a proper symbol
calculus in the sense that $ \sigma ( A B ) $ cannot be related to $
\sigma ( A ) \sigma ( B) $. However, when $ A \in
\Psi^{\rm{comp}}_{\frac12} ( X) $ and  $ B \in \Psi ( X )  $ then 
$ \sigma ( A B ) = \sigma ( A) \sigma ( B ) + {\mathcal O}( h^{\frac12} ) _{
  S_{\frac12} ( T^* X ) }  $ -- see \cite[Lemma
3.6]{SZ10} or \cite[Lemma 5.1]{Da-Dy}.

\subsection{Fourier integral operators}
\label{s:fio}

In this paper we will consider Fourier integral operators associated
to canonical transformations.  It will also be sufficient to consider
operators which are compactly microlocalized as we will always
work near $ p^{-1} ( [ - 2 \delta, 2 \delta ] ) \cap w^{-1} ( 0 ) $
which by assumption \eqref{eq:CAPe} is a compact subset of
$ T ^* X $. 

Suppose that $ Y_1, Y_2  $ are two compact smooth manifolds ($ Y_j = X $ or $
Y_j = \TT^n  $ in what follows) and that, $ U_j \subset T^* Y_j  $ are open subsets.
Let 
\[  \kappa : U_1 \to U_2 , \ \  \Gamma_\kappa'  \defeq \{ ( x, \xi , y, - \eta  ) :
(x,\xi ) = \kappa (  y , \eta ) ,  ( y , \eta )  \in U_1 \} \subset T^*Y_2 \times T^* Y_1 , \]
be a symplectic transformation, for instance $ \kappa = \varphi_t $, $
U_1 = U_2 = T^* X $. Here $ \Gamma_\kappa $ is the graph of $ \kappa $
and $ {}' $ denotes the twisting $ \eta \mapsto - \eta $. This follows
the standard convention \cite[Chapter 25]{Hor2}.

Following \cite[\S 5.2]{Da-Dy} we introduce the class of 
  {\em compactly microlocalized $h$-Fourier integral
  operator quantizing $ \kappa $},  $I^{\rm{comp}}_h ( Y_2 \times Y_1, \Gamma_\kappa'  ) $.
If $ T \in I^{\rm{comp}}_h ( Y_2 \times Y_1, \Gamma_\kappa'  )   $ then 
it has the following properties: $ T = {\mathcal O}( 1
) _{L^2 ( Y_1) \to
  L^2 ( Y_2 ) }  $;  there exist 
$  A_j \in \Psi^{\rm{comp}} ( Y_j ) $, $  \WFh (A_j ) \Subset U_j $ such
  that
\[  A_2 T = T + {\mathcal O}(h^\infty) _{ {\mathcal D}' ( Y_1 )  \to \CI ( Y_2 ) }, \ \ 
T A_1 = T + {\mathcal O}(h^\infty) _{ {\mathcal D}' ( Y_1 )  \to \CI ( Y_2 ) } ; \]
for any $ B_j  \in \Psi^m ( Y_j ) $, 
\be
\label{eq:Egor} \begin{split}  &  T B_1 = C_1 T + h T_1 , \ \  \sigma ( C_1 ) = \sigma (
B_1 )\circ \kappa^{-1} , \\  &  B_2 T = T C_2 + h T_2 , \ \  \sigma ( C_2 ) =
\sigma ( B_2 )\circ\kappa , \ \ T_j \in I^{\rm{comp}}_h ( Y_2 \times Y_1, \Gamma_\kappa'
). \end{split} \ee
The last statement is a form of Egorov theorem.

When $ B_j \in \widetilde \Psi_{\frac12}^{\comp} ( X ) $ then an
analogue of \eqref{eq:Egor} still holds in a modified form 
\be \label{eq:Egor12} \begin{split}  &  T B_1 = C_1 T +h^\frac12 \tilde
  h^{\frac12}  D_1 T_1 , \ \  \sigma ( C_1 ) = \sigma (
B_1 )\circ \kappa^{-1} , \\  &  B_2 T = T C_2 + h^\frac12 \tilde
  h^{\frac12}  T_2 D_2 , \ \  \sigma ( C_2 ) =
\sigma ( B_2 )\circ\kappa , 
\\
& \ \ \ \ T_j \in I^{\rm{comp}}_h ( Y_2 \times Y_1, \Gamma_\kappa' 
), \ \  C_j,  D_j \in \widetilde \Psi^\comp_{\frac12} ( X ) , 
\end{split} \ee
see Proposition \ref{p:EgorE} (applied with $ g \equiv 0 $).

An example is given by the operators
\begin{equation}
\label{eq:Aexp} A\, e^{ - i t P/h} ,\quad e^{ - i tP/h } A \in I^{\rm{comp}} ( X \times
X , \Gamma_{\varphi_t}' ), \ \text{ if $ A \in \Psi^{\rm{comp}} ( X )
  $. }
\end{equation}

In \S \ref{ats} we will also need a local representation of elements of 
$ I^\comp $ as oscillatory integrals -- see \cite{A},\cite[\S 3.2]{DyaGu} and
references given there. If $ T \in I^{\comp} ( \RR^n \times \RR^n , 
\Gamma_\kappa' ) $ is microlocalized to a sufficiently small 
neighbourhood $ \kappa ( U ) \times U \subset T^* \RR^n \times T^* \RR^n $
 (\cite[8.4.5]{e-z}) then 
\begin{equation}
\label{eq:osc-int} 
T u ( x ) = ( 2 \pi h )^{- \frac{k+n} 2}   \int_{\RR^{k}}
  \int_{\RR^n} e^{ \frac i h  \psi ( x , y,  \theta ) }  a ( x, y,
  \theta )  u ( y ) d y d \theta+{\mathcal O} (
  h^\infty) _{\mathcal S} \| u \|_{H^{-M} }  ,
\end{equation} 
for any $ M $. Here $ a \in \CIc ( \RR^{2n} \times \RR^{k } ) $, $ \psi \in \CI ( \RR^{2n} \times \RR^k ) $, and near $ \kappa
( U ) \times U $, the graph of $ \kappa $ is given by 
\begin{equation}
\label{eq:GammaK} 
\begin{gathered}   \Gamma_\kappa = \{ ( ( x, d_x \psi ( x, y ,
  \theta ) ) , ( y , 
- d_y \psi ( x, y , \theta ) ) \; : \;( x, y , \theta ) \in C_\psi
 \} , \\
C_\psi \defeq \{ ( x, y ,\theta ) \; : \; d_\theta \psi
( x, y , \theta ) = 0 ,\}  , \\ d_{ x, y,\theta } ( \partial_{\theta_j}
\psi ),  \ \ j = 1, \cdots, k , \ \text{ are linearly
  independent,} 
\end{gathered}
\end{equation}
For given symplectic coordinates $ ( x, \xi ) $ and $ ( y , \eta ) $
in neighbourhoods of $ \kappa ( U )$ and $ U $ respectively, 
such a representation exists with an extra variable of dimension $
k$, where $ 0 \leq k \leq n $, and $ n + k $
is equal to the rank of the projection \[  \Gamma_\kappa \ni ( ( x, \xi ) , ( y , \eta) ) \longmapsto ( x,
\eta ) , \]
assumed to be constant in the neighbourhood of $ \kappa ( U ) \times
U $ -- see for instance \cite[Theorem 2.14]{e-z}.  Since $
\Gamma_\kappa $ in \eqref{eq:GammaK} is the graph of a symplectomorphism
it follows that for some $ y' =(  y_{j_1} , \cdots, y_{j_{n-k} } ) \in
\RR^{n-k} $, 
\begin{equation}
\label{eq:nond}
D_\psi ( x ,y , \theta )  \defeq \det \Big( \frac{\partial^2 \psi}{\partial x_i \partial y'_{j'}} , \frac{\partial^2 \psi}{\partial x_i \partial \theta_j} \Big )   
\neq 0 . 
\end{equation}
For the use in \S \ref{ats} we record the following lemma, proved
using standard arguments (see for instance \cite{A}):

\begin{lem}
\label{l:fio}
Suppose that $ T $ is given by \eqref{eq:osc-int} and that $ B \in
\widetilde \Psi_{\frac12} ( \RR^n ) $. Then for any $ u\in L^2 $ with $ \| u
\|_{L^2} = 1$, 
\begin{equation*}
\begin{split}   B  T  u ( x ) & = ( 2 \pi h )^{- \frac{k+n} 2}   \int_{\RR^{k}}
  \int_{\RR^n} e^{ \frac i h  \psi ( x , y,  \theta ) }  a ( x, y, 
  \theta )\, b ( x, d_x \psi ( x, y , \theta ) )  u ( y ) d y
  d \theta + {\mathcal O} (
  h^{\frac12} \tilde h ^{\frac12}   ) _{ L^2}, \\
  T B u ( x ) & =  ( 2 \pi h )^{- \frac{n+k} 2}   \int_{\RR^{k}}
  \int_{\RR^n} e^{ \frac i h  \psi ( x , y,  \theta ) }  a ( x, y,
  \theta )\, b ( y, -d_y \psi ( x, y , \theta ) )  u ( y ) d y
  d \theta + {\mathcal O} (
  h^{\frac12} \tilde h ^{\frac12}  ) _{ L^2} , 
\end{split}
\end{equation*}
where $ b = \sigma ( B ) $. 
\end{lem}

\subsection{Fourier integral operators with operator valued symbols}
\label{fioo}

In \S \ref{ats} we will also use a class of Fourier integral operators
with {\em operator valued} symbols. We present what we need in an
abstract form in this section. Only local aspects of the theory will
be relevant to us and we opt for a direct presentation.

Suppose that $ {\mathcal H} $ is a separable Hilbert space and $ Q $ is an
(unbounded) self-adjoint operator with domain $ {\mathcal D} \subset
{\mathcal H} $.  We assume that $ Q :  {\mathcal D} \to {\mathcal H} $
is invertible and we put $ {\mathcal D}^\ell \defeq Q^{-\ell} {\mathcal H}
$, for $ \ell \geq 0 $. For $ \ell < 0 $, we define 
$ {\mathcal D}^{\ell } $ as the completion of $ \mathcal H $
with respect to the norm $ \| Q^\ell u \|_{\mathcal H} $.  

We define the following class of operator valued symbols:
\begin{equation}
\label{eq:ovs}
{\mathcal S}_{\delta} ( \RR^{2n} \times \RR^{k}, {\mathcal H},
{\mathcal D} ), 
\end{equation}
to consist of operator valued functions 
\[  \RR^{n} \times \RR^n \times \RR^{k} \ni ( x, y , \theta ) \longmapsto N
( x, y , \theta ) : {\mathcal D}^\infty \longrightarrow  {\mathcal H} \]
which satisfy the following estimates:
\begin{equation}
\label{eq:NO}   \partial^\alpha_{x, y , \theta } N ( x, y , \theta ) = {\mathcal
  O}_{\alpha , \ell } ( 1 ) :   
 {\mathcal D}^{ \ell + \delta | \alpha| }
\longrightarrow {\mathcal D}^\ell  , \end{equation}
for any multiindex $\alpha$ and $ \ell\in\ZZ $, uniformly in $ ( x, y , \theta ) $.  
We note that this class is closed under pointwise composition of the operators:
if $ N_j \in {\mathcal S}_\delta $ then $ N_j $ defines a family of
operators $ \mathcal D^\ell \to \mathcal D^\ell $,  hence
so does their product $N_1N_2$; the estimate \eqref{eq:NO} follows for the composition,
since for $ |\beta | + | \gamma | = | \alpha| $,
\[  \partial^\beta N_1 \partial^\gamma N_2 = 
{\mathcal O}( 1 )_{ \mathcal D^{ \ell + \delta | \alpha| } 
\to \mathcal D^{ \ell + \delta ( | \alpha | - |\beta | )} }
{\mathcal O}( 1 )_{ \mathcal D^{ \ell + \delta( | \alpha| 
- |\beta|)} 
\to \mathcal D^{ \ell } } = 
{\mathcal O} ( 1 )_{
 {\mathcal D}^{ \ell + \delta | \alpha| }
\longrightarrow {\mathcal D}^\ell}  . \]
Proposition \ref{p:meta} at the end of this section 
describes a class which will be used in 
\S \ref{ats}.

Suppose that $ \psi $ satisfies \eqref{eq:GammaK} and
\eqref{eq:nond}. We can 
assume that $ \psi $ is defined on $ \RR^{2n } \times \RR^{k} $.
For $ N \in {\mathcal S}_\delta $ and $ a \in \CIc ( \RR^{2n} \times
\RR^{k} ) $ we define the operator
\[   T : L^2 ( \RR^n ) \otimes {\mathcal H} \longrightarrow  L^2 ( \RR^n ) \otimes
{\mathcal H} ,  \ \ \ \ L^2 ( \RR^n ) \otimes \mathcal H 
\simeq L^2 ( \RR^n , \mathcal H) , \]
(the second identification is valid as $ \mathcal H $ is separable 
\cite[Theorem II.10]{resi} but it is convenient in definitions to 
use the tensor product notation) by 
\begin{equation}
\label{eq:defT}
T ( u \otimes v ) \defeq 
( 2 \pi h )^{- \frac{n+k} 2}   \int_{\RR^{k}}
  \int_{\RR^n} e^{ \frac i h  \psi ( x , y,  \theta ) }  a ( x, y,
  \theta )  \left( u ( y ) \otimes N ( x, y , \theta ) v \right) d y d
  \theta.
\end{equation}
This operator is well-defined since $ a $ is compactly supported, but to
obtain a norm estimate which is uniform in $ h$  we need to assume
that $ N \in \mathcal S_0 $:
\begin{lem}
\label{l:normf}
Suppose that $ N \in \mathcal S_0 ( \RR^{2n+k}, {\mathcal H}, \mathcal
D )
 $ and that $ T $ is given by \eqref{eq:defT}. Then 
\begin{equation}
\label{eq:normf} \| T \|_{ L^2 ( \RR^n ) \otimes \mathcal H \to L^2 ( \RR^n )
  \otimes \mathcal H } = \max_{ C_\psi} \frac{ \  \  \   | a | \| N
  \|_{\mathcal H \to \mathcal H}   } {
  \sqrt{| D_\psi |} } + {\mathcal O} ( h) , 
\end{equation} 
where $ C_\psi \defeq \{ ( x, y , \theta ) : \partial_\theta
\psi = 0 \} $, and $ D_\psi $ is given by \eqref{eq:nond}. 

If $ 
N \in \mathcal S_\delta ( \RR^{2n + k } , \mathcal H, \mathcal D) $
then 
\begin{equation}
\label{eq:normf1}
T = {\mathcal O} ( 1 ) : L^2 ( \RR^n ) \otimes {\mathcal D}^{ \delta m_n
  + \ell } \longrightarrow 
L^2 ( \RR^n ) \otimes {\mathcal D}^{ \ell} , \end{equation}
where $ m_n $ depends only on the dimension $ n $.
 \end{lem}
\begin{proof}
The  estimate \eqref{eq:normf} follows from 
a standard argument based  on considering $ T ^* T $ and from
\cite[Theorem 13.13]{e-z}. The estimates \eqref{eq:NO} with $ 
\delta = 0 $ and $ \ell = 0 $ show that the operators can be treated 
just as scalar symbols. 

To obtain \eqref{eq:normf1} we note that
\[  \partial^\alpha_{ x, y , \theta } \left( Q^{-L} N ( x, y , \theta )
\right) = {\mathcal O} ( 1 ) : \mathcal H \to \mathcal H , \ \ 
\text{ for $ | \alpha | \delta  \leq L $.} \]
To obtain the norm estimate \eqref{eq:normf} we only need
a finite number of derivatives, $M$, depending only on 
the dimension. Taking $ m_n\delta \geq L$, we can then apply \eqref{eq:normf} to the operator $ Q^{-L} T $, which gives the bound \eqref{eq:normf1}
for $ T $.
\end{proof}

A special case of is given by $ \kappa = id $. In that
case we deal with pseudodifferential operators with 
operator valued symbols. 
The following lemma summarizes their basic properties:
\begin{lem}
\label{l:prodop}
Suppose that $ N_j \in \mathcal S_{\delta_j } ( \RR^{2n} ) $, $j=1,2$.
For $ u \in {\mathscr S} 
 ( \RR^n )$ and $ v \in {\mathcal D}^\infty $ we define
\[ \Op ( N_j ) ( u \otimes v ) \defeq 
\frac{1}{( 2 \pi h)^n } \int e^{ \frac i h \langle x - y , \xi 
\rangle } \left[ N_j ( \textstyle{ \frac{x+y} 2 } , \xi ) v \right] 
u ( y ) dy d \xi . \]
These operators extend to 
\begin{equation}
\label{eq:normN}
 \Op ( N_j ) = {\mathcal O} ( 1 ) : L^2 ( \RR^n ) \otimes {\mathcal D}^{ \ell 
+ m_n \delta_j } \to L^2 ( \RR^n ) \otimes {\mathcal D}^\ell , 
\end{equation}
and satisfy the following product formula:
\begin{equation}
\label{eq:propop}
\Op (N_1 ) \Op ( N_2 ) = \Op ( N_1 N_2 ) + h R , \ \ \ 
R = {\mathcal O} ( 1 ) : 
L^2 ( \RR^n ) \otimes {\mathcal D}^{ \ell 
+ m_n ( \delta_1 + \delta_2 )  } \to L^2 ( \RR^n ) \otimes {\mathcal D}^\ell .
\end{equation}
Here and in \eqref{eq:normN}, 
$ \ell$ is arbitray and  $ m_n $ depends only on the dimension $ n $.
\end{lem}
\begin{proof}
When $ \delta_1 = \delta_2 = 0 $ the proof is an immediate
vector valued adaptation of the standard arguments presented
in \cite[\S\S4.4,4.5]{e-z} where we note that only a finite
number (depending on the dimension) of seminorms of symbols 
is needed.
In general, \eqref{eq:NO} gives
\begin{equation} 
\label{eq:QN1N2}  \partial_{x,\xi}^\alpha  Q^{-L } N_j Q^{-M } = {\mathcal O} ( 1 ) 
: \mathcal D^\ell \to \mathcal D^\ell , \ \  | \alpha| \delta_j \leq L + M , \end{equation}
and the norm estimates \eqref{eq:normN} follows.  
To obtain the product formula we note that, using \eqref{eq:QN1N2},
 it applies to $ Q^{-M} N_1 $ and $ N_2 Q^{-M} $ for $ M $
 sufficiently large depending on $ n $. Hence
\[ \begin{split} \Op ( N_1 ) \Op ( N_2 ) & = Q^{M} \Op ( Q^{-M} N_1 ) \Op ( N_2 Q^{-M}
) Q^M \\
& =  Q^M \Op ( Q^{-M} N_1 N_2 Q^{-M} ) Q^M + Q^M {\mathcal O}
(h)_{L^2 \otimes \mathcal D^p \to L^2 \otimes \mathcal D^p}  Q^M\\
&  =
\Op ( N_1 N_2 ) + \mathcal O ( h)_{ L^2 \otimes \mathcal D^{p+M} \to
L^2 \otimes  \mathcal D^{p-M}  } ,
\end{split}
\]
which gives \eqref{eq:propop} provided $ m_n(\delta_1+\delta_2) \geq 2M $.
\end{proof}

We can also factorize the operator $ T $ using the pseudodifferential
operators described in Lemma \ref{l:prodop}, the proof being
an adaptation of the standard argument.
 When $ S : L^2 ( \RR^n ) \to
L^2 ( \RR^n ) $ we also write $ S $ for  $ S \otimes \Id_\cH :
L^2 ( \RR^n ) \otimes \cH \to L^2 ( \RR^n ) \otimes \cH $.
\begin{lem}
\label{l:fact}
Suppose that $ T $ is given by \eqref{eq:defT} with $ N \in \mathcal
S_\delta $. Then 
\[  T =  T^\parallel \Op (  N_1 ) + h R_1 = \Op ( N_2 ) T^\parallel +
h R_2 , \]
where 
\begin{equation}
\label{eq:fact}
\begin{gathered}
T^\parallel\in I^\comp ( \RR^n \times \RR^n , \Gamma_\kappa'), \ \ 
T^\parallel u ( x )   = ( 2 \pi h )^{- \frac{k+n} 2}   \int_{\RR^{k}}
  \int_{\RR^n} e^{ \frac i h  \psi ( x , y,  \theta ) }  a ( x, y,
  \theta )  u ( y ) d y d \theta, \\
N_2 ( x , d_x \psi ( x, y, \theta )) = 
N_1 ( y , - d_y \psi ( x, y, \theta )) = N( x, y , \theta ) , \quad  ( x, y, \theta )\in C_\psi \\
R_j =  {\mathcal O} ( 1 ) : L^2 ( \RR^n ) \otimes {\mathcal D}^{ \delta m_n
  + \ell } \longrightarrow 
L^2 ( \RR^n ) \otimes {\mathcal D}^{ \ell} .
\end{gathered}
\end{equation}
Here, $ N_j \in \mathcal S_\delta ( \RR^n \times \RR^n , {\mathcal H},
\mathcal D) $,  and 
\[ \Op ( N_j ) = {\mathcal O} ( 1 ) :  L^2 ( \RR^n ) \otimes {\mathcal
  D}^{\delta m_n + \ell}   \longrightarrow L^2
( \RR^n ) \otimes {\mathcal D}^\ell  . \]
\end{lem}

In our applications we will have 
\begin{equation}
\label{eq:hos}   {\mathcal H} = L^2 ( \RR^{d_{\perp}} , d \tilde y )  , \ \ \
 Q = - \th^2 \Delta_{\tilde y} + \tilde y^2 + 1 , \end{equation}
so that  $ {\mathcal D}^\ell $ are analogous to Sobolev spaces (see 
\cite[\S 8.3]{e-z}).
In the rest of this section (as well in section~\ref{ats}), we will use the shorthand notations $\rho_\parallel = (x,y,\theta)$ in order to shorten the expressions, and to differentiate between 
these variables and the ``transversal variables" $(\ty,\teta)$.

We consider a specific class of metaplectic operators:
\be
\label{eq:simple}  N ( \rho_\parallel ) u ( \tilde y ) = ( 2 \pi  \tilde h)^{-
    d_{\perp}} \int_{\RR^{d_{\perp}} } \int_{\RR^{d_{\perp}}}  
(\det \partial^2_{\ty, \teta}q_{\rho_\parallel} ) ^{\frac12}  e^{
  \frac i \th ( q_{\rho_\parallel} ( \tilde y , \tilde \eta )  - \langle \tilde
  \eta , \tilde y' \rangle )} u ( \tilde y' ) d \tilde y' ,  
\end{equation}
where $ q_{\rho_\parallel}(\ty,\teta) $ is a real quadratic form in the variables $\ty,\teta$, with coefficients depending on $\rho_\parallel$, being in the class
$ S ( \RR^{2n+k} ) $, and the matrix of coefficients $ \partial^2_{\ty, \teta}  q_{\rho_\parallel}$ is assumed to be uniformly non-degenerate for all $\rho_\parallel$.  The definition involves a {\em choice} of the branch of 
the square root -- see Remark \ref{rem:sign} for further discussion of that.
For any fixed $\rho_\parallel$ these operators are unitary on $ \mathcal H$  (see for instance \cite[Theorem 11.10]{e-z}).  

The next proposition shows that this class fits nicely into
our framework:
\begin{prop}
\label{p:meta}
The operators $ N (\rho_\parallel) $ given by 
\eqref{eq:simple} satisfy
\begin{equation}
\label{eq:meta1}  \partial^\alpha_{\rho_\parallel} N (\rho_\parallel) = {\mathcal
 O}_{\alpha, \ell}  ( \th^{ - | \alpha|  } ) :    {\mathcal D}^{ | \alpha| +
 \ell  }
\longrightarrow {\mathcal D}^\ell  , \end{equation}
for all $ \ell$, That means that  \eqref{eq:NO} holds with $ \delta = 1 $ (the loss in $ \tilde
h $ is considered as dependence on $ \alpha $).  

If $ \tilde \chi \in
{\mathscr  S} ( \RR^{2d_{\perp} } ) $ is  fixed, $ \Lambda >1 $, and 
$ \tilde \chi_\Lambda ( \bullet ) \defeq \tilde \chi ( \Lambda^{-1}
\bullet ) $, 
then for any $\ell$ and $k\geq 0$,
\begin{equation}
\label{eq:meta2}
\widetilde \chi_\Lambda^w ( \tilde y , \tilde h D_{\tilde y } ) = 
\mathcal O_\ell ( \Lambda^{2k} ) : \mathcal D^{ \ell  } \to  \mathcal
D^{ \ell + k } , 
\end{equation}
so that 
\begin{equation}
\label{eq:meta3} \begin{split} 
&   \partial^\alpha_{\rho_\parallel} \left( \tilde  \chi^w_ \Lambda( \tilde y , 
\th  D_{\tilde y } )N ( \rho_\parallel ) \right) = {\mathcal
 O}_{\alpha, \ell}  ( \Lambda^{ 2  |\alpha| } \th^{ - | \alpha|  } ) :    {\mathcal D}^{ 
 \ell  } \longrightarrow {\mathcal D}^\ell \,, \\
&  \partial^\alpha_{\rho_\parallel } \left( N ( \rho_\parallel ) \tilde  \chi^w_ \Lambda( \tilde y , 
\th  D_{\tilde y } ) \right) = {\mathcal
 O}_{\alpha, \ell}  ( \Lambda^{  2 |\alpha| } \th^{ - | \alpha|  } ) :    {\mathcal D}^{ 
 \ell  } \longrightarrow {\mathcal D}^\ell .
\end{split}
\end{equation}
\end{prop}
\begin{proof}
We see that
$ \partial^\alpha_{\rho_\parallel} N ( \rho_\parallel ) $ is an operator of
the same form as \eqref{eq:simple} but with the amplitude 
multiplied by 
\[  \sum_{ | \beta| \leq 2 | \alpha | }
\th^{ - m_\beta  } 
\ty^{\beta_1} (\ty' )^{\beta_2 } \teta^{\beta_3}
q_\beta ( \rho_\parallel ) , \ \ q_\beta \in S ( \RR^{2n + k } ) , 
\ \ \beta = ( \beta_1, \beta_2 , \beta_3 )\in \NN^{3d_{\perp} },  \ \
\beta_j \in \NN^{d_{\perp}} \,,
 \]
where $ m_\beta \leq | \alpha | $.
Hence to obtain \eqref{eq:meta1}, it is enough to prove that 
\[  Q^{  \ell }   \ty^{\beta_1} N ( \rho_\parallel ) \left( (\ty' )^{\beta_2 }
( \th D_{\ty' })^{\beta_3} Q^{-\ell - | \alpha| } v ( \ty' ) \right) =
{\mathcal O} ( \| v \|_{ \mathcal H }  )_{\mathcal H} . \]
Using the exact Egorov's theorem for metaplectic 
operators (see for instance \cite[Theorem 11.9]{e-z}) we
see that the left hand side is equal to 
\[  N ( \rho_\parallel ) \left ( p^w_\beta ( \ty' , \th D_{\ty' } )
\left( K_{q }^* Q \right) ^\ell Q^{-\ell - |\alpha | } v ( \ty' )
\right),  \ \  K_q : ( \partial_{\teta } q, \teta )
\mapsto  ( \ty , \partial_{\ty } q ) , \]
$  q = q_{\rho_\parallel}$ and where $ p_\beta  $ is a polynomial of
degree less than or equal to $ | \beta|$.   Since $ |\beta | \leq
2 |\alpha| $, the operator $ p_\beta^w (K^*_q Q)^\ell Q^{-\ell - | \alpha |
}$
is bounded on $ {\mathcal H}$ (see for instance \cite[Theorem
8.10]{e-z}) so the unitarity of $ N $ gives the boundedness in 
$\mathcal H $.

To obtain \eqref{eq:meta2} we first note that $ 
\tchi_\Lambda \in S ( \RR^{2 d_{\perp} } ) $ uniformly 
in $ \Lambda > 1 $. Hence $ Q^{-\ell } \tchi^w_\Lambda Q^\ell
= \mathcal O ( 1 )_{ \mathcal H \to \mathcal H}  $, 
uniformly in $ \Lambda $ (again, see \cite[Theorem 8.10]{e-z}).
This gives \eqref{eq:meta2} for $ k  = 0 $. For the
general case we put $ Q_\Lambda =  1 + \Lambda^{-2}( ( \tilde h D_{\ty})^2  +  \ty^2 )$, and note that for any $ M$,
$ Q_\Lambda^M \tchi^w_\Lambda = \tchi^w_{\Lambda, M }  $, where
$ \tchi_{\Lambda, M }  \in S ( \RR^{2d_{\perp}})  $ uniformly in $ \Lambda $. Hence
it is bounded on $ L^2 ( \RR^{d_{\perp}} ) $ uniformly in $ \Lambda $ and $ \th $. 
We then 
write
\[ \begin{split}  Q^{k } \chi^w_\Lambda  ( \ty , \th D_{ \ty} ) & = 
 Q^{ k } Q_\Lambda^{ - k } Q_\Lambda^{k }
\chi^w_\Lambda  ( \ty , \th D_{ \ty} ) \\
& =  
 ( 1 + ( \th D_{\ty })^2 +  \ty^2 )^{k } 
 ( 1 + \Lambda^{-2}  ( \th D_{\ty })^2 + 
\Lambda^{-2}  \ty^2 )^{-k  } \chi_{\Lambda, k  }^w (
\ty , \th D_{\ty } ) \\
& = 
 \Lambda^{ 2 k } 
 ( 1 + ( \th D_{\ty })^2 +  \ty^2 )^{k} 
 ( \Lambda^2 +   ( \th D_{\ty })^2 + 
  \ty^2 )^{-k  }  \chi_{\Lambda, k  }^w (
\ty , \th D_{\ty } )  \\
& = {\mathcal O} ( \Lambda^{2 k } 
 )_{ L^2 ( \RR^{d_{\perp}} ) 
\to  L^2 ( \RR^{d_{\perp}} )  } ,
\end{split} \]
completing the proof of \eqref{eq:meta2}.
\end{proof}

\section{Classical dynamics}
\label{cld}

In this section we will describe the consequences of the normal
hyperbolicity assumption \eqref{eq:NH0},\eqref{eq:NH} needed in the 
proof of Theorem \ref{t:1}.

\subsection{Stable and unstable distributions}
\label{sud} 

Let $ K^\delta $ be the trapped set \eqref{eq:Kdel} and $ E_\rho^\pm
\subset T_\rho X $, $ \rho \in K^\delta $, the distributions 
in \eqref{eq:NH}. We recall our notation
$\varphi_t \defeq \exp t H_p$ for the Hamiltonian flow generated by
the function $p(x,\xi)$.

We start with a simple 
\begin{lem}
\label {l:basicd}
If $ \omega $ is the canonical symplectic form on $ T^* X$ then 
\begin{equation}
\label{eq:NH1}
\omega_\rho \restriction_{ E_\rho^\pm } = 0  ,
\end{equation}
that is $ E_\rho^\pm $ are isotropic.

Without loss of generality we can assume that the distributions
$E^\pm_\rho$ satisfy
\begin{equation}
\label{eq:NH2}    E_\rho^+ \oplus E_\rho^- = ( T_\rho K^\delta)^\perp , 
\end{equation}
where $ V^\perp $ denotes the symplectic orthogonal of $ V $.
\end{lem} 
\begin{proof}
The property \eqref{eq:NH1} 
follows from the fact that $\varphi_t$ preserves the symplectic
structure ($ \varphi_t^* \omega = \omega $).
For $ X , Y \in E_\rho^\pm $, 
\begin{gather*}   \omega_\rho ( X , Y ) = \omega_{ \varphi_{\mp t }  (
    \rho ) } ( ( 
d\varphi_{\mp t})(\rho)  X  ,  d\varphi_{\mp t}(\rho)  Y  )
\rightarrow 0 \,, \ \ t \rightarrow + \infty  \,. 
\end{gather*}

To see that we can assume \eqref{eq:NH2} we note that 
the distribution $\{(T_\rho K^\delta )^\perp,\,\rho\in K^\delta\}$ is invariant by
the flow: $ d \varphi_t ( \rho ) : T_\rho K \to T_{\varphi_t ( \rho )
} K $, and $ d \varphi_t (\rho ) $ is a symplectic transformation.
If 
$ \pi_\rho : T_\rho ( T^* X ) \to ( T_\rho K^\delta)^\perp $ is the symplectic
projection, then $ \pi_{ d\varphi_t ( \rho) }  \circ d \varphi_t (\rho)= 
d \varphi_t(\rho) \circ \pi_\rho $. This means that we may safely replace
$ E_\rho^\pm $  with $ \pi_\rho (  E_\rho^\pm  )$, without altering
the properties \eqref{eq:NH}. 
\end{proof}

\subsection{Construction of the escape function}
\label{cef}

To construct the escape function near the trapped set 
we need a lemma concerning invariant cones near $ K^\delta$.
To define them we introduce a Riemannian metric on $ T^*X $ and use
the tubular neighbourhood theorem (see for instance 
\cite[Appendix C.5]{Hor2}) to make the identifications
\begin{equation}
\label{eq:neK}  \begin{split}
 \neigh ( K^\delta ) & \simeq N^* K^\delta \cap \{  ( \rho, \zeta ) \in
T^* (T^*X)  :  \| \zeta \|_\rho  \leq 
\epsilon_1 \} \\
& \simeq ( T K^\delta )^\perp
\cap \{  ( \rho, z  ) \in
T (T^*X)  :  \| z \|_\rho  \leq 
\epsilon_1 \}  \\
&\simeq \{(m,z) : m\in K^\delta,\ z\in\IR^{2d_{\perp}},\ \| z \|_\rho  \leq 
\epsilon_1 \}\,.
\end{split}
\end{equation}
Here $( T K^\delta )^\perp$ denotes the symplectic orthogonal of $ T K^\delta
\subset T_{ K^\delta} ( T^* X )  \subset T (T^*X )$. Since $ K^\delta$
is symplectic, the symplectic form identifies $ (T K^\delta)^\perp $
with the conormal bundle $ N^* K^\delta $.  The norm $ \| \bullet
\|_\rho $ is a smoothly varying norm on $ T_\rho ( T^* X )$. 
We write $ d_\rho  ( z, z') = \| z - z'\|_\rho $ and introduce a
distance function $ d : \neigh ( K^\delta ) \times \neigh ( K^\delta )
\to [0 , \infty ) $ obtained by choosing a Riemannian metric on 
$ \neigh ( K^\delta ) $. We  have 
$ d ( ( m , z ) , ( m , z' ) ) \sim d_{m} ( z, z') $ and the
notation $ a \sim b $, 
here and below, means that there exists a constant $ C\geq 1 $ (independendent
of other parameters)  such that $ b/C \leq a \leq C b $.

Assuming that $ E_\rho^\pm $ are chosen so that \eqref{eq:NH2} 
holds we can define (closed) cone fields by putting
\begin{gather}
\label{eq:Cpmd}  
\begin{gathered} C_\rho^\pm \defeq \{ z \in (T_\rho K^\delta )^\perp :  d_\rho ( z ,
E_\pm^\rho ) \leq \epsilon_2 \| z \|_\rho ,\ \| z \|_\rho  \leq \epsilon_1 \}, \\ 
 C^\pm \defeq \bigcup_{\rho \in
  K^\delta } C_\rho^\pm  \subset \neigh (K^\delta )  ,
\end{gathered}
\end{gather}
where we used the identification \eqref{eq:neK}. Since the 
maps 
$ \rho \mapsto E_\rho^\pm $ are continuous, $ C^\pm $ are 
closed. 

The basic properties $ C^\pm $ are given in the following 
\begin{lem}
\label{l:Cpm}
There exists $ t_0 > 0 $ and $\eps_2^0>0$ such that, for every $ t > t_0 $ there
exists  $ \epsilon_1^0 $ such that if one chooses
$ \epsilon_j < \epsilon_j^0 $, $ j=1,2 $ in the 
definition of $ \neigh (K^\delta ) $ and $ C^\pm $, then 

\begin{equation}
\label{eq:Cpm1}
\rho \in C_\pm  , \  \varphi_{\pm t } ( \rho ) \in \neigh( K^\delta )
\ \Longrightarrow \ \varphi_{\pm t} ( \rho ) \in C_\pm .  
\end{equation}
In fact a stronger statement is true: for some constant 
$ \lambda_1 > 0 $ and any $t\geq t_0$,
\begin{equation}
\label{eq:Cpm2}
\rho  ,   \varphi_{\pm t } ( \rho ) \in \neigh( K^\delta )
\ \Longrightarrow \
d ( \varphi_{\pm t} ( \rho ), C^\pm ) \leq e^{ - \lambda_1  t} d (
 \rho  , C^\pm ) . 
\end{equation}
Finally,
\begin{equation}
\label{eq:Cpm3}
d ( \rho, C^+ )^2 + d ( \rho, C^- )^2 \sim d ( \rho, K^\delta)^2 .
\end{equation}
\end{lem}
The conclusions \eqref{eq:Cpm2} and \eqref{eq:Cpm3}  are 
similar to \cite[Lemma 4.3]{NSZ2} and \cite[Lemma 5.2]{SjDuke} but
the proof does not use foliations by stable and unstable manifolds
which seem different under our assumptions.
\begin{proof}
For $ \rho \in \neigh( K^\delta ) $ let $ ( m, z ) $, $ m \in K^\delta
$ and $ z \in \RR^{2d_{\perp} } \simeq ( T_m K^\delta)^\perp$ be local coordinates near $ \rho
$. Similarly let $ ( \tilde m , \tilde z ) $ be local coordinates 
near $ \varphi_t ( \rho ) \in \neigh ( K^\delta ) $ (by assumption in 
\eqref{eq:Cpm1}). Then if for each $m$ we put $d_\perp \varphi_t ( m )\defeq d
\varphi_t ( m ) \rest_{ ( T_m K^\delta)^\perp }$, the map $\varphi_t$
can be written as,
\begin{equation}
\label{eq:varphi} \begin{split}  \varphi_t ( m , z ) & = \big( \varphi_t ( m ) +{ \mathcal O}_t ( \| z \|^2
) ,  d_\perp \varphi_t ( m )  z  + 
{ \mathcal O}_t ( \| z \| ^2 ) \big)\\
& = \big(  \varphi_t( m_1 )  ,  d_\perp \varphi_t (m_1   ) z + 
{ \mathcal O}_t ( \| z \| ^2 ) \big),  
\ \ \ m_1 =  m  + {\mathcal O}_t  ( \| z \|^2) \,. 
\end{split}
 \end{equation}
(Here we identify $ (T_m K^\delta)^\perp $ with $ \RR^{2 d_{\perp} } $ and
and consider $ d_\perp\varphi_t ( m ) : \RR^{2d_{\perp} } \to \RR^{2d_{\perp} } $, with
similar identification near $ \varphi_t ( \rho ) $. The norm $ \|
\bullet \| $ is now fixed in that neighbourhood.)

Let $ z = z_+ + z_- $ be the decomposition of $ z $ corresponding to 
$ ( T_{m_1} K^\delta)^\perp = E^+_{m_1} \oplus E^-_{m_1} $ (we assumed
without  loss of generality that \eqref{eq:NH2} holds). The continuity
of $ \rho \mapsto E^\pm_\rho $ and the definition of $ C^+_{\rho} $ show
that if $ \epsilon_1 $ is small enough depending on $ t $ (so that $ d
( m, m_1 ) = {\mathcal O}_t ( \|z\| ^2
) $ is small)
\begin{equation}
\label{eq:Cz+}  z \in C^+_{m} \ \Longrightarrow \ \| z_- \|_{m_1}  \leq 2 \epsilon_2 \|
z_+ \|_{m_1}  . \end{equation}
Since
\[
 d_\perp \varphi_t (m_1   ) z
= \sum_{\pm} d_\perp \varphi_t (m_1   ) z_\pm , 
 \quad  d_\perp \varphi_t (m_1   ) z_{\pm}  \in E_{\varphi_t ( m_1) } ^\pm , 
\]
normal hyperbolicity implies that for some $C>0$ and $\lambda_1>0$ 
\begin{equation}
\label{eq:gath}
\begin{split}
\| d_\perp \varphi_t  (m_1   ) z_+ \| &\geq \frac {1}{C} e^{ 2 \lambda_1 t } \| z_+\|, \\  
\| d_\perp \varphi_t  (m_1   ) z_-\| & \leq {C} e^{ - 2 \lambda_1 t } \| z_-\| ,
\end{split}
\end{equation}
for all positive times $ t$.

If $z\in C^+_{m}$, then this and \eqref{eq:Cz+}  show
$$
\|d_\perp \varphi_t (m_1   ) z_-\| \leq 2C^2\,e^{-4\lambda_1
  t}\,\eps_2 \|d_\perp \varphi_t (m_1   ) z_+\|\,.
$$
Let us take $t_0$ such that $2C^2\,e^{-4\lambda_1 t_0}<1/2$. For $
t\geq t_0 $ and $ \epsilon_1 $ small enough depending on $ t$ 
this shows that 
\begin{equation}
\label{eq:dC+}  z\in C^+_{m }\ \text{and}\ \|z\|  \leq \eps_1, 
\ \| d_\perp \varphi_t (m_1   ) z \| \leq \epsilon_1  \Longrightarrow 
 d_\perp \varphi_t (m_1   ) z 
+ {\mathcal O}_t ( \| z \|^2 ) \in C^+_{ \varphi_t ( m_1 ) } ,
\end{equation}
which in view of \eqref{eq:varphi}  proves \eqref{eq:Cpm1} in the $ + $ case with the $ - $ case being
essentially the same. 

To obtain \eqref{eq:Cpm2} we note that for $ ( m , z ) \in
\neigh( \rho, K^\delta ) $, 
\[  d ( ( m , z ) , C^{+} ) \sim 
    d_{m}  ( z , C^+_{m}  ) \sim \|z_-\| ( 1 -  \bbbone_{ C^+_m} ( z
) ),  \ \ z = z_+ + z_- , \ \ z_\pm \in E_m^\pm ,
\]
where $ \bbbone_A $ is the characteristic function of a set $ A$. 
(To see the first $ \sim $ we need to show that $ d ( ( m, z ) ,
C^+) \leq  c_0 d_m ( z , C_m^+ ) $ for some $ c_0 $, which follows
from an argument by 
condradiction using pre-compactness of $ K^\delta $.)

We also observe that if $ d \varphi_t ( m_1 ) z \in \neigh( K^\delta )
$ then \eqref{eq:dC+} gives, for $ \epsilon_1 $ small enough depending
on $ t $, 
\[   1 - \bbbone_{ C_{\varphi_t ( m_1)}  ^+} \!\!\! \left( d_\perp \varphi_t ( m_1 ) z
+ {\mathcal O}_t ( \| z \|^2 ) \right) 
\leq 1 - \bbbone_{ C_m^+ } ( z ) . \]
Hence, using \eqref{eq:varphi} and \eqref{eq:gath}, writing $ z =
z_- + z_+ $ as before, and taking $ \epsilon_1 $ sufficiently small 
depending on $ t \geq t_0 $,
\[ \begin{split} 
d ( \varphi_t ( m , z ) , C^+) & \sim d_{\varphi_t ( m_1) } \big( d_\perp
\varphi_t (m_1) z + 
{\mathcal O}_t ( \|z\|^2)  , 
C_{\varphi_t ( m_1 ) }^+ \big)
 \\
& \sim   \| d_\perp \varphi_t ( m_1) z_-  \| \left( 1 + {\mathcal O}_t (
  \| z_- \|) \right) 
\big( 1 - \bbbone_{C^+_{\varphi_t ( m_1 ) }} ( d_\perp \varphi_t (m_1) z +
{\mathcal O}_t ( \|z\|)^2 ) \big) \\
& \leq C  e^{ - 2 \lambda_1 t } \| z_- \| \left( 1 + {\mathcal O}_t (
  \| z_- \| ) \right) ( 1 - \bbbone_{C^+_{ m_1} } ( z ) ) \\
& \leq  C' e^{- 2\lambda_1 t } d_{m_1}  ( z , C_{m_1}^+ ) \sim
C' e^{-2 \lambda_1 t } d_m (  z , C_{m}^+ ) \\
& \leq  e^{ - \lambda_1 t } d ( ( m , z ) , C^+ ) .
\end{split}
\]
Here in the second line we used the fact that $ \| z \| \leq C \|z_-
\| $ if the distance is non zero (with $C$ depending on $\eps_2$). 
In the fourth line we used the
continuity of the cone field, $ m \mapsto C^+_m $.

This proves \eqref{eq:Cpm2}.
The last claim \eqref{eq:Cpm3} is immediate from the construction of
$ C^\pm $ and the fact that $ E_\rho^+ \cap E^-_\rho = \{ 0 \} $.
\end{proof}

We now regularize $
d ( \rho , C^\pm )^2 $ uniformly with respect to a parameter $
\epsilon $. It will eventually be taken to be $  h/\tilde h $, where 
$\tilde h $ is a small constant independent of $ h$.
Lemma \ref{l:Cpm}  and the arguments of 
\cite[\S 4]{NSZ2} and \cite[\S 7]{SZ10} immediately give
\begin{lem}
\label{l:regu}
There exists $ t_0 > 0 $ such that for any $ t > t_0 $, there exists
a neighbourhood $ \VV_t  $ of $ K^{2 \delta} $ and a constant  $C_0>0$
such that the following holds.

For any
small $ \epsilon > 0 $ there
exist functions $ \hph_\pm \in \CI ( \VV_t \cup \varphi_t ( \VV_t ) ) $
such that
for $ \rho \in \VV_t \cap p^{ -1} ([-\delta, \delta]) $, 
\begin{equation}
\label{eq:es1'}
\begin{split}
& \hph_\pm  ( \rho )  \sim d( \rho  , C_\pm )^2 + \epsilon , \ \
\gamma_\pm ( \rho ) \geq \epsilon 
 \,, \\
& \pm (\hph_\pm ( \rho )  - \hph_\pm ( \varphi_t ( \rho ) )  + C_0 \epsilon
\sim  \hph_\pm ( \rho )
\,,\\
& \partial^\alpha
 \hph_\pm (\rho ) =\Oo ( \hph_\pm ( \rho )^{ 1
- |\alpha|/2} )\,,  \\
&  \hph_+ ( \rho ) + \hph_- ( \rho )
\sim d( \rho ,  K^\delta )^2 +  \epsilon \,.
\end{split}
\end{equation}
\end{lem}

Following \cite[\S 4]{NSZ2} and \cite[\S 7]{SZ10} again this gives
us an escape function for a small neighbourhood of the trapped set.
We record this in 
\begin{prop}
\label{p:esc}
Let $ \hph_\pm $ be the functions given in Lemma \ref{l:regu}.
For $ L \gg 1 $ independent of 
$ \epsilon $, define 
\begin{equation}
\label{eq:es2}
\hG \defeq
\log ( L \epsilon + \hph_- ) - \log (  L \epsilon + \hph_+ )
\end{equation}
on a  neighbourhood $ \VV $ of the trapped set $ K^{2 \delta } $.

For any $ t_0 $ large enough, and $ L$ depending on $
t_0 $,  we can find a neighbourhood of 
$ U_1 \Subset \VV $ of  $ K^{2 \delta } $  and $ c_1, c_2 , C_1 , C_2, 
 > 0 $, independent of $ L$, 
such that 
\[ 
\hG=\Oo(\log(1/\eps)),\quad \partial_\rho^\alpha \hG =
 \Oo (   \min (
\gamma_+ , \hph_- )^{-\frac{|\alpha|}2
} )  = \Oo (  \eps^{ -\frac{|\alpha|}2
} ) \,, \quad  |\alpha | \geq 1 \,, \]
and such that for $ \rho \in   U_1 \cap p^{-1} ( [ - \delta, \delta ]
)   $, 
\begin{gather}
\label{eq:es4}
\begin{gathered}
\partial_\rho^\alpha ( \hG (\tF ( \rho ) ) - \hG ( \rho )) =
 \Oo(   \min (\hph_+ , \hph_- )^{-\frac{|\alpha|}2
} )  = \Oo(  \eps^{ -\frac{|\alpha|}2
} ) \,, \quad |\alpha | \geq 0 \,,  \\
  d ( \rho,
K^\delta )^2 \geq C_1 \eps\ \ 
\Longrightarrow \  \hG (\tF ( \rho ) ) - \hG ( \rho ) 
\geq c_1 / L   \,,  \\
  d ( \rho,
K^\delta )^2 \leq  c_2 L  \eps\ \ 
\Longrightarrow \  | G ( \rho ) | \leq C_2 . 
\end{gathered}
\end{gather}
\end{prop}

\medskip

\begin{rem}\label{rem:constants}
For the reader's convenience we make some comments on the
constants in Proposition \ref{p:esc} referring to the 
proof of \cite[Lemma 4.4]{NSZ2} for details. The constant 
$ L$ has to be large enough depending on the implicit 
constants in \eqref{eq:es1'}. The constants $ C_1 , C_2 $ have
to be large enough, and constants $ c_1, c_2 $ small 
enough, depending on the implicit constants in \eqref{eq:es1'}.
In \S \ref{eat} it matters that we can 
take $ c_2 L > C_1 $ which is certainly possible. 
\end{rem}

In the intermediate region between $ U_1 $ and $ \{x: w(x) > 0 \} $ 
we need an escape function similar to the one constructed in 
\cite[\S 4]{DatVas10} and \cite[Appendix]{GeSj}.  We work here
under the general assumptions of \S \ref{aatr} and present a
slightly modified argument.
\begin{lem}
\label{l:cef3}
Suppose that  $ X $ is a compact smooth manifold, 
$ p \in S^m ( T^*X; \RR ) $, $ w \in S^k ( T^* X ;
[0, \infty )) $, $ k \leq m $, and that \eqref{eq:CAPe} holds. 
For any open neighbourhood  $ V_1 $ of $ K^{ 3 \delta } $, there exists  $ \epsilon_1 > 0 $ and a function 
$G_1 \in \CIc \big( p^{-1} ( ( - 2 \delta , 2 \delta ) ) \big)$ such that
\begin{gather}
\label{eq:lcef3} 
\begin{gathered}  G_1(\rho)=0\ \ \text{for $\rho$ in some neighbourhood of $K^{3\delta}$,}\\
H_p G_1 ( \rho ) \geq 0  \ \text{  for $ \rho
  \notin w^{ -1} ( ( \epsilon_1 , \infty ))  $}, \\
H_p G_1 ( \rho ) > 0 \ \text{ for $ \rho \in 
p^{ -1 } ( [ - \delta, \delta ] ) \setminus \big( V_1 \cup w^{-1} ( (
\epsilon_1 , \infty )  ) \big)$.}
\end{gathered}
\end{gather}
\end{lem}
\begin{proof} Call $ U_0 \defeq w^{-1} ( ( 0 , \infty ) ) $ and suppose $ \rho \in 
p^{ -1 } ( [ - 2\delta, 2\delta ] ) \setminus ( V_1 \cup U_0 ) $. 
We first claim that there exist $ T_\pm = T_\pm ( \rho ) $, 
$ T_- < 0 < T_+ $,  such 
that 
\begin{gather} 
\label{eq:T1}  \varphi_{T_+}  ( \rho ) \in U_0 \ \text{ or } \ 
\varphi_{ T_-} ( \rho )  \in U_0 \,, \\
\label{eq:T2}  \varphi_{ T_\pm }( \rho ) \in V_1 \cup U_0  \,. 
\end{gather}
(Here and below we use the notation $ \varphi_A ( \rho ) = \{ \varphi_t ( \rho ) : t
\in A \} $.)

To justify these claims we first note that  since $ \rho \notin K^{2 \delta } $,
$ \varphi_{ \RR} ( \rho ) \cap U_0 \neq \emptyset  $ which 
implies that 
\begin{equation} 
\label{eq:T11} \exists \, T_1,  \ \  \varphi_{ T_1 } ( \rho ) \in U_0. 
\end{equation} 
Assuming that $ T_1  < 0 $ we want to show that $ \varphi_{ T_2 } ( \rho )
\in V_1 \cup U_0 $ for some $ T_2  > 0 $. Suppose that this is not true, 
that is 
\begin{equation}
\label{eq:empty}  \varphi_{ (  0 , \infty  )}( \rho ) \cap ( V_1 \cup
U_0 ) =  \emptyset\, .\end{equation}
Then 
for any $ t_j \to \infty $, 
\[  \rho_j \defeq \varphi_{t_j} (
\rho ) \in p^{-1} ( [ -  2 \delta, 2 \delta ] ) \setminus 
( V_1 \cup U_0 ), \ \ \varphi_{ [ 0 , \infty )} ( \rho_j ) \cap ( V_1
\cup U_0 ) = \emptyset . \]
By \eqref{eq:CAPe} the set $ p^{-1} ( [ -  2 \delta, 2 \delta ] ) \setminus 
( V_1 \cup U_0 ) $ is compact and hence, by passing to a subsequence, 
we can assume that $ \rho_j \to \bar \rho \notin V_1 \cup U_0 $. 
We have $ \varphi_ t ( \rho_j ) \to \varphi_ t ( \bar
\rho ) $, as $ j \to \infty $, uniformly for $ |t | \leq T $, and 
it follows that $ \varphi_{  [ 0 , \infty )} ( \bar \rho)  \cap ( V_1 \cup U_0 )
=  \emptyset $.  For $ t \geq -t_j $, 
\[  \varphi_t ( \rho_j )  = \varphi_{ t + t_j}( \rho ) \subset
\varphi_{  [ 0 , \infty )}( \rho ) \subset p^{-1} ( [ -2  \delta , 2 \delta
] ) \setminus ( V_1 \cup U_0 ) , \]
which means that $ \varphi_ t ( \bar \rho ) \notin V_1 \cup U_0 $
for $ t > - t_j \to - \infty$. We conclude that
$$ 
\varphi_ \RR  ( \bar \rho ) \cap  V_1 \cup U_0 =\emptyset
\ \Longrightarrow \ \varphi_\RR ( \bar \rho ) \in K^{3 \delta } . 
$$
This contradicts the property $\bar{\rho}\not\in V_1$, and proves the existence of $T_2>0$ such that 
$\varphi_{T_2}(\rho)\in V_1\cup U_0$. We call $T_-(\rho)=T_1$, $T_+ (\rho)= T_2$.


In the case  $ T_1 $ in \eqref{eq:T11} is
positive, a similar argument shows the existence of $T_2<0$ such that
$ \varphi_{ T_2}( \rho ) \in ( V_1 \cup
U_0 ) \neq  \emptyset$. In this case we call $ T_- ( \rho ) =T_2$, $T_+(\rho)=T_1$.

For each $ \rho \in p^{-1} ( [ - 2\delta, 2\delta ] ) $ we can find 
an open hypersuface $ \Gamma_\rho $,  
transversal to $ H_p $ at $ \rho $,
 such that, if $ \varphi_{ T_\pm} ( \rho ) 
\in U_0 $,  then for $ \rho' \in \Gamma_\rho  $, 
\begin{gather*} 
 \varphi_{ T_\pm} ( \rho' ) \in U_0 , \ \ \ 
\varphi_{  T_\mp}  ( \rho' ) \in V_1 \cup U_0 \,.
\end{gather*}
Notice that the closure of the tube
$\Omega_\rho \defeq \varphi_{ ( T_-  , T_+ )} ( \Gamma_\rho )$ does not intersect $K^{3\delta}$.
Using this tube, we construct a local escape functions $ g_\rho \in  \CIc ( \Omega_\rho ) $, with the following properties: 
for some $ \epsilon_\rho > 0$, and an slightly smaller tube
$ \Omega'_\rho \subset \Omega_\rho  $ containing $ \varphi_{ ( T_- , T_+ ) }(\rho ) $, 
\begin{equation}
\label{eq:grho}
H_p g_\rho  ( \rho' ) \geq 0 , \ \ \rho' \notin w^{-1} ( (
\epsilon_\rho, \infty  ) ) ,  
\ \
H_p g_\rho ( \rho' ) > 0 , \ \ \ \rho' \in \Omega_\rho' \setminus 
( w^{-1}( ( \epsilon_\rho, \infty  ) ) \cup V_1 )   .
\end{equation}
Here $ \epsilon_\rho $ is chosen so that if $ \varphi_ {T_\pm} ( \rho )
\in U_0 $ then 
$ \varphi_{  T_\pm}  (\Gamma_\rho ) \subset w^{ -1} ( ( 2 \epsilon_\rho,
\infty )) $. 

To construct $ g_\rho $ we take $ ( t, m  ) \in ( T_-,
T_+) \times \Gamma_\rho   $ as local coordinates:
$ ( t, m ) \mapsto \varphi_t  ( m ) \in \Omega_\rho $.  
Suppose that $ \varphi_{  T_-} ( \rho ) \in U_0 $, and that
$ \varphi_ {( T_- , T_- + \gamma ) } ( \Gamma_\rho )   \subset w^{-1}( (\eps_\rho,\infty) ) $ and 
$ \varphi_{ ( T_+ - \gamma , T_+ )}(\Gamma_\rho) \subset V_1\cup U_0  $. 
Choose $ \chi_\rho \in \CIc ( ( T_-, T_+ ) ) $ which 
is strictly increasing on $ ( T_- + \gamma, T_+ - \gamma ) $ and non-decreasing on $ ( T_+ - \gamma, T_+ ) $.
Also, choose $ \psi_\rho \in \CIc ( \Gamma_\rho ) $ 
with $ \psi_\rho ( \rho ) =1 $. Then put 
$   g_\rho ( \varphi_ t  ( m )) \defeq  \chi_\rho ( t ) \psi_\rho ( m ) $.
Since     $ H_p g_\rho = \chi_\rho' ( t ) \psi_\rho ( m ) $,
\eqref{eq:grho} holds. A similar construction can be applied in the case where 
$\varphi_{  T_-} ( \rho ) \in V_1$, $\varphi_{  T_+} ( \rho ) \in U_0$.

From the open cover 
\[  p^{-1} ( [ - \delta, \delta ] ) \setminus ( V_1 \cup U_0 ) \subset
\bigcup \left\{ \Omega_\rho  : {\rho \in  p^{-1} ( [ - \delta, \delta ] ) \setminus ( V_1
  \cup U_0 ) }  \right\} 
, \] 
one may extract a finite subcover $\bigcup_{j=1}^L\Omega_{\rho_j}$. 
The closure of this cover does not intersect $K^{3\delta}$, so that
the function
$  G_1 ( \rho ) \defeq \sum_{ l=1}^L g_{\rho_L} ( \rho ) $ 
satisfies \eqref{eq:lcef3}, for $\eps_0=\min_j \eps_{\rho_j}$.
\end{proof}

We conclude this section with a global escape function which 
combines the ones in Proposition \ref{p:esc} and  Lemma \ref{l:cef3}.
The estimates will be needed to justify the quantization of the escape
function in \S \ref{micr}. The proof is an immediate adaptation 
of the proof of \cite[Proposition 4.6]{NSZ2} and is omitted.

\begin{prop}
\label{p:esc2}
Let $\VV $, $ U_1 $, $\hG$ and $ t_0 $ be as in Proposition ~\ref{p:esc}, and 
let $ W_1 $ be a neighbourhood of $ K^{2 \delta }$ such that
$ W_1 \Subset U_1$,  $W_1 \cup \varphi_{t_0} (W_1) \Subset \VV$.
 
Take $ \chi\in \CIc ( \VV ) $ equal to $ 1 $
in $W_1 \cup \varphi_{t_0} (W_1)$, and let $ G_1 $ 
be the escape function constructed in Lemma~\ref{l:cef3}  for $ V_1 =
W_1 $. Then for any $ \Gamma > 1 $, $  G \in \CIc ( T^*X ; \RR ) $ defined by 
\begin{equation}
\label{eq:defG}
G \defeq \chi C_3  \Gamma \hG + C_4\log(1/\eps)\, G_1 
\end{equation}
where $ C_3 $ and $ C_4 $ are sufficiently large, 
satisfies the following estimates
\begin{gather}
\label{eq:es3}
\begin{gathered}
 |G(\rho)|  \leq C_6\,\log ( 1 / \epsilon ) \,, \quad
\partial^\alpha G  = 
\Oo( \eps^{-|\alpha|/2} ) \,, \quad |\alpha | \geq 1 \,, 
  \\
\rho\in W_1 \;\Longrightarrow \; \quad G ( \varphi_{t_0}  ( \rho)  ) - G ( \rho )  \geq - C_7\,,
\\
\rho \in W_1 \cap p^{-1} ( [ - \delta, \delta ] ) \,, \quad
  d ( \rho, K^\delta)^2 \geq C_1 \eps
\;  \Longrightarrow \;  G ( \varphi_{t_0}  ( \rho)  ) - G ( \rho )
\geq  2 \Gamma 
\,,\\
\rho\in p^{ -1 } ( [ - \delta, \delta ] ) \setminus \big( W_1 \cup w^{-1} ( (\epsilon_1 , \infty )  ) \big)
 \;\Longrightarrow \; G(\varphi_{t_0} (\rho)) - G (\rho)
\geq C_8\,\log(1/\eps)\,,
\end{gathered}
\end{gather}
with $ C_8 > 0 $.

In addition we have
\be\label{eq:orderf}
 \frac{\exp G ( \rho ) }{ \exp G( \rho' ) } \leq C_9 \left( 1 
+ \frac {  d( \rho , \rho') }{\sqrt \epsilon} \right)^{N_1} \,,
\ee
for some constants $ C_9 $ and $ N_1 $.
\end{prop}

\section{Analysis near the trapped set}
\label{ats}

In this section we will analyse the cut-off propagator
\be\label{e:cutoff-propag}
  \chi^w \exp ( - i t P / h ) \chi^w \,, 
\ee
where $ \chi^w = \Op ( \chi ) $, $ \chi \in \CIc \cap \widetilde
S_{\frac12} $ and $ \supp \chi \subset \{ \rho : d ( \rho, K^{\delta }
) < R ( h / \tilde h)^{\frac12} \}$ for some $R>1$ independent of
$\th,\,h$. We could take two different cut-offs on both sides, as long
as they share the above properties.

Our objective is to prove the following bound (announced in \eqref{eq:tricky}):
\begin{prop}\label{p:crucial-bound}
For any $ \epsilon_0 > 0 $ and $ M > 0 $, there exist $C_0 >0$, 
$ \th_0 $, and a function $ \th \mapsto h_0 ( \th ) >0 $, such that 
for $0<\th<\th_0$ and $0<h<h_0(\th)$, 
\begin{gather}
\label{eq:tricky2}
\begin{gathered}
\| \chi^w  e^{ - i t P / h } \chi^w  \|_{ L^2 \to L^2 } \leq C_0 \, \tilde h^{ -d_{\perp}/2} \,
\exp\left(- \textstyle{\frac12}  t {(\lambda_0 - \epsilon_0)}  \right) , 
\quad
 0 \leq t \leq M \log 1/ \th  \,, 
\end{gathered}
\end{gather}
where $ \lambda_0 $ is given by \eqref{eq:t11}. 
\end{prop}
Since $e^{ - i t P / h }$ is unitary, the above bound is nontrivial only for 
\[ 0  \leq  \frac{d_{\perp} }{\lambda_0}\log \frac 1\th \leq t \leq M \log \frac 1 \th\,.\]

\subsection{Darboux coordinate charts}
\label{dcc}

We start by setting up an adapted atlas of Darboux coordinate charts near $K^{\delta}$,
that is take a finite open cover 
\[
K^{\delta}\subset\bigcup_{j\in J}U_{j},
\]
and symplectomorphisms $\kappa_{j}:U_{j} \to V_{j}=\neigh(0,\IR^{2d})$.
The standard symplectic coordinates on $V_{j}$ then appear
as a local symplectic coordinate frame on $U_{j}$. We may choose the
coordinates such that they split into 
\[
X=(x,y),\ \ \ \Xi=(\xi,\eta),\quad y, \eta \in \RR^{d_{\perp}} ,\quad  x, \xi \in
\RR^{d-d_{\perp}} \,,
\]
such that the symplectic submanifold $K^{\delta}\cap U_{j}$ is
identified with $\cK\cap V_j\subset \IR^{2d}$, where 
$$\cK\defeq\left\{
  y=\eta=0\right\}\subset \IR^{2d} . $$
That is, $(x,\xi)$ is a local coordinate frame
on $K^\delta$, while $(y,\eta)$ describes the transversal
directions.

We also assume that for each $\rho\in K^{\delta}\cap U_{j}$, identified
with some $(x,0,\xi,0)\in V_{j}$, the subspace
$\left\{ \left(x,y,\xi,0 \right),\: y\in\IR^{d_{\perp}}\right\} $
is $\eps$-close to the transversal unstable space $d\kappa_{j}(E_{\rho}^{+})$,
while the subspace $\left\{ \left(x_0,0,\xi_0,\eta \right),\,\eta\in\IR^{d_{\perp}}\right\} $
is $\eps$-close to the transversal stable space
$d\kappa_{j}(E_{\rho}^{-})$. 


We want to describe the flow in the vicinity of $K^{\delta}$, using
these local coordinates. We choose 
a (large) 
time $t_0>0$, and express the time-$t_0$ flow
$\varphi_{t_0}:U_{j_{0}}\to U_{j_{1}}$ in the local coordinate frames,
through the maps
\begin{equation}
\label{eq:DeAr}
 \kappa_{j_{1}j_{0}}\defeq\kappa_{j_{1}}\circ\varphi_{t_0}\circ\kappa_{j_{0}}^{-1}:
 D_{j_1 j_0}\to A_{j_1 j_0}\,,\end{equation}
where $D_{j_1 j_0 }\subset V_{j_0}$ is the {\em departure set}, while
$A_{j_1 j_0}\subset V_{j_1 }$ is the {\em arrival set}. This is
defined when
$\varphi_{t_0}(U_{j_0})\cap U_{j_1}\neq\emptyset$ and such a pair 
$ {j_1j_0 } $ for which this holds will be called {\em physical}.

Below we will also
consider the maps $\kappa^n_{j_{1}j_{0}}$ representing the time-$nt_0$
flow in the charts $V_{j_0}\to V_{j_1}$ -- see \S \ref{s:back-iterated}.

\subsection{Splitting $e^{-it_0P/h}$ into pieces}\label{s:pieces}

We want to use the fact that the propagator $e^{-it_0P/h}$ is a
Fourier integral operator on $M$ associated with $\varphi_{t_0}$. To make
this remark precise, we will use a smooth partition of unity
$\big(\pi_j\in \CIc(U_j,[0,1])\big)$ such that each cut-off $\pi_{j}$ is equal to unity near
some $\tU_{j}\Subset U_j$, and the quantized cut-offs $\Pi_i\defeq
\Op(\pi_i)$ satisfy the following quantum partition of unity:
\be\label{e:partition}
\Pi\defeq \sum_{j=1}^J\Pi_j\,\Pi_j^*\equiv \Id\ \quad\text{microlocally
  in a neighbourhood of }K^\delta\,.
\ee
We will then split  $e^{-it_0P/h}$
into the local propagators 
\begin{equation}
\label{eq:Tfl} T^\flat_{j_1j_0}\defeq\Pi_{j_1}^*
e^{-it_0 P/h} \Pi_{j_0}, 
\end{equation}
 which can be represented by
operators on $L^2(\IR^d)$ as follows.
We define Fourier integral operators $\cU_{j}:L^{2}(X) \to L^{2}(\IR^{d})$ quantizing
the coordinate changes $\kappa_{j}$, and microlocally unitary in
some subset of $V_j\times U_j$ containing $\kappa_j(\supp
\pi_j)\times\supp\pi_j$, so that 
\be\label{e:iden}
\forall j,\quad \Pi_j\Pi_{j}^* = \Pi_j \cU_j^*\cU_j\Pi_j^*+\cO(h^\infty)\,,
\ee
The local propagators $T^\flat_{j_1 j_0}$ are then represented by
\be\label{e:T_jj}
T_{j_{1}j_{0}}\defeq\cU_{j_{1}}\Pi_{j_1}^*\,e^{-it_0 P/h}\,\Pi_{j_0}\cU_{j_{0}}^{*}\,.
\ee
Notice that for an unphysical pair $j_1j_0$,
$T_{j_1j_0}=\cO(h^\infty)_{L_2\to L_2}$. 
For a physical pair $j_1j_0$, 
$T_{j_1j_0}$ is a Fourier integral operator associated with the local symplectomorphism $\kappa_{j_1j_0}$.

From the unitarity of $e^{-it_0 P/h}$ we draw the following property of
the operators $T_{j'j}$.
\begin{lem}\label{l:normT}
The operator-valued matrix $\bT\defeq (T_{ij})_{i,j=1,\ldots, J}$, acting on the space
$L^2(\IR^d)^J$ with the Hilbert norm $\| \bu\|^2 = \sum_{j=1}^J \|u_j\|_{L^2}^2$,
satisfies 
$$
\|\bT\|_{L^2(\IR^d)^J\to L^2(\IR^d)^J} = 1+\cO(h)\,.
$$ 
\end{lem}
\medskip

\noindent {\em Proof.} 
From \eqref{e:iden}, the action of $T_{j_1j_0}$ on
$L^2(\IR^d)$ is (up to an error $\cO (h^\infty) _{L^2\to L^2}$)
unitarily equivalent with the
action of $T^\flat_{j_1j_0}$ on $L^2(X)$. Hence, the action of $\bT$ on $L^2(\IR^d)^J$
is equivalent to the action of $\bT^\flat$ on
$L^2(X)^J$, where $ \bT^\flat $ is the matrix of operators \eqref{eq:Tfl}.

To obtain the norm estimate we follow \cite[Lemma 6.5]{AnNo07}, 
put  $\cH\defeq L^2(X)$, $U=e^{-it_0P/h}$, and define
the row vector of cut-off operators
$C=(\Pi_i)_{i=1,\ldots,J}$. The operator valued
matrix $\bT^\flat$ can be written as $\bT^\flat = C^* (U\otimes \Id_J)C$. Its operator
norm on $\cL(\cH^J)
$ satisfies
\begin{align*}
\| \bT^\flat \|_{\cL(\cH^J)}^2 &= \| (\bT^\flat)^* \bT^\flat
\|_{\cL(\cH^J)} 
= \| C^* (U\otimes \Id_J) CC^* (U^*\otimes \Id_J) C \|_{\cL(\cH^J)}
\\ &
 = \| C^* (U\Pi U^*\otimes \Id_J) C \|_{\cL(\cH^J)}\,.
\end{align*}
Egorov's theorem (see \eqref{eq:Egor}) and \cite[Theorem 13.13]{e-z}
imply that $\Pi^1\defeq U\Pi U^*$ is
a positive semidefinite operator of norm $1+\cO(h)$, with symbol equal to $1+\cO(h)$ near $K^\delta$, and 
its square
root $\sqrt{\Pi^1}$, as well as the product $\sqrt{\Pi^1}\,\Pi\,\sqrt{\Pi^1}$, have the same
properties. Hence, 
\begin{align*}
\| \bT^\flat \|_{\cL(\cH^J)}^2 & = 
\|  \left((\sqrt{\Pi^1}\otimes \Id_J)\,C \right)^* 
(\sqrt{\Pi^1}\otimes \Id_J)\,C  \|_{\cL(\cH^J)} \\ 
& = \|  (\sqrt{\Pi^1}\otimes \Id_J)\,C 
\left( (\sqrt{\Pi^1}\otimes \Id_J)\,C \right)^*  \|_{\cL(\cH)} \\
&=  \|\sqrt{\Pi^1}\,\Pi\,\sqrt{\Pi^1}\|_{\cL(\cH)} 
=  1 + \cO ( h ) \,. \xqedhere{2in}
\end{align*}

\subsection{Iterated propagator}
In this section we explain how to use the $T_{j'j}$ to study
our cut-off propagator \eqref{e:cutoff-propag}.

First of all, Egorov's theorem \eqref{eq:Egor12} applied to
$T=\cU_{j}\,\Pi_{j}^{*}$, $B_2=\chi^w$ allows us to write
\be\label{e:chi-chi_j}
\cU_{j}\,\Pi_{j}^*\,\chi^w = \chi_{j}^w \,\cU_{j}\,\Pi_{j}^* +
\cO( h^{\frac12} \tilde h^{\frac12 } )_{L^2\to L^2} \,,
\ \ \ j=1,\ldots,J,
\ee
where the symbol $\chi_{j} = \chi\circ\kappa_{j}^{-1}  \in
\tS_{\frac12}(T^*\IR^d)$.

We start from a arbitrary normalized state $u\in L^2(X)$, and 
represent the part of $u$ microlocalized near $K^\delta$ through the (column) vector of states 
$$
\bu \defeq (u_j)_{j=1,\ldots,J},\quad u_j\defeq
\cU_j \Pi_j^*\, u ,
\qquad 
\|\bu\|^2 \defeq \sum_j \|u_j\|^2 = \la u,\Pi u\ra + \cO(h^\infty)\| u
\|_{L^2} ^2 \,.
$$
The equations (\ref{e:T_jj}) and (\ref{e:chi-chi_j}) show that 
\[ \begin{split}
\Pi 
e^{-it_0 P/h}\chi^w \,u & = \sum_{ j } \Pi_j \Pi_j^* e^{-it_0
  P/h}\chi^w \,u = \sum_{j_1, j_0} 
 \Pi_{j_1} \Pi_{j_1} ^* e^{-it_0   P/h} \Pi_{j_0 }\Pi_{j_0}^*
\chi^w  \,u \\
& = \sum_{j_1, j_0}  
\Pi_j \mathcal U_{j_1}^* \mathcal U_{j_1} \Pi_{j_1} ^* e^{-it_0  P/h} 
\Pi_{j_0}  \mathcal U_{j_0}^* \mathcal U_{j_0}\Pi_{j_0}^*
\chi^w \, u  + {\mathcal O} ( h^\infty )_{ L^2 \to L^2 }  
\\
& = \sum_{j_1,j_0} \Pi_{j_1}\,\cU_{j_1}^*\,
T_{j_1j_0} \chi_{j_0}^w u_{j_0}+ \cO( h^{\frac12} \tilde h^{\frac12 } )_{L^2\to L^2} \,.
\end{split}
\]
Similarly, for $n\geq 2$ the propagator $e^{-int_0 P/h}$ can be represented by
iteratively applying the operator valued matrix $\bT$ to the vector
$\bu$. By inserting the identities (\ref{e:partition}),(\ref{e:iden}) $n$
times in the expression $\Pi e^{-int_0 P/h} \chi^w u$, we get the following
\begin{lem}\label{l:iterate}
For any $n\in\IN$ (independent of $h$), we have
\be \label{e:iterate} 
\begin{split}
\Pi e^{-int_0 P/h} \, \chi^w u 
&=\sum_{j_{n},\ldots,j_0} 
\Pi_{j_n}\cU_{j_n}^*\,
T_{j_nj_n-1}\cdots T_{j_1j_0} \,\chi_{j_0}^w\,u_{j_0}+ \cO_{n}
( h^{\frac12} \tilde h^{\frac12 } )_{L^2\to L^2} \\
&=\sum_{j_n,j_0} \Pi_{j_n}\cU_{j_n}^* 
\,[(\bT)^n]_{j_nj_0}\,\chi_{j_0}^w\,u_{j_0}  +
\cO_{n}
 ( h^{\frac12} \tilde h^{\frac12 } )_{L^2\to L^2} \,, 
\end{split} \ee
where the matrix of operators, $ {\bT} $, was defined in Lemma \ref{l:normT}.
\end{lem}

\subsubsection{Inserting nested cut-offs}\label{s:inserting-chi}

In this section we modify the Fourier integral operators $T_{j'j}$, taking into
account that in the above expression their products are multiplied by narrow cut-offs $\chi_j^{w}$.

By construction of $\chi_{j}$, there exists $R_0>0$ (independent of
$h,\ \th$) such that for any index $j$ the
cut-off $\chi_{j}\in\tS_{\frac12}$ is supported inside the microscopic cylinder
\be\label{e:cylinder}
B_{R_0(h/\th)^{1/2}}\defeq \{(x,y,\xi,\eta)\,:\,|y|,|\eta|\leq R_0(h/\th)^{1/2}\}\subset T^*\IR^d\,.
\ee 
Fix some $R_1\geq 2R_0$, and choose a function $\tchi^0\in C^\infty_0(\IR^{2d_{\perp}},[0,1])$ equal to unity
in the ball $\{|\ty|,|\teta|\leq R_1\}$, and supported
in $\{|\ty|,|\teta|\leq 2R_1\}$. 
Normal hyperbolicity implies that there exists $\Lambda>2$ such that the cylinders
$B_{\bullet}$ (see \eqref{e:cylinder}) satisfy
\be\label{e:nesting}
\kappa_{j'j}(B_{2R (h/\th)^{1/2} })\Subset B_{R \Lambda (h/\th)^{1/2} }\,, 
\ee
for all $ 0 < R < 1  $ and  any physical pair $j'j$. 

We then define the families of nested\footnote{ Below we use the notation $\chi^0\succ \chi$
for nested cut-offs, meaning that $ \chi^0 \equiv 1 $ near $\supp(\chi)$.} cut-offs
$\{ \chi^n\} _{n\in \IN}$, $\{ \tchi^n\} _{n\in \IN}$ as follows:
\begin{align}\label{e:cutoffs}
\forall n\in\IN,\quad 
\tchi^n(y,\eta) &\defeq \tchi^0(y \Lambda^{-n},\eta \Lambda^{-n}),\\
\chi^n(x,y,\xi,\eta)&\defeq
\tchi^n\big(y(\th/h)^{1/2},\eta (\th/h)^{1/2}\big)  
\in \tS_{\frac12}(T^*\IR^d)\,.
\,,\label{e:rescaled-cutoffs}
\end{align}
We stress that the $\tS_{\frac12}(T^*\IR^d)$ seminorms of $\chi^n$ hold uniformly in $ n $:  the smoothness  of $ \chi^n $ actually improves when $n$ grows.
From the assumption
$R_1>R_0$ we draw the nesting $\chi^0\succ \chi_j$ for any
$j=1,\ldots,J$.
Furthermore, the property \eqref{e:nesting} implies that 
\be\label{e:nesting0}
\text{for any physical pair }j'j,\qquad
\chi^{n+1}\succ\chi^n\circ\kappa_{j'j}\,.
\ee 
From these nesting properties and from Egorov's property
\eqref{eq:Egor12} we easily obtain
the following 
\begin{lem}\label{l:cutoffs}
For any $j=1,\ldots,J$ we have
\be\label{e:chi0chi}
(\chi^0)^w\chi_{j}^w=\chi_{j}^w + \cO
(\th^\infty)_{L^2 \to L^2}  \,,
\quad
\chi_{j}^w\,(\chi^0)^w = \chi_{j}^w + \cO(\th^\infty)_{L^2 \to L^2}\,.
\ee
In addition, we have the estimate 
\be\label{e:insert-cutoff}
T_{j' j}(\chi^n)^w = (\chi^{n+1})^w\,T_{j' j}(\chi^n)^w + \cO
(\th^{\infty})_{L^2 \to L^2}\,,
\ee
uniformly for all $j,j'=1,\ldots,J$ and for all $n$ independent of $h$.
\end{lem}
We will actually only use
$n$ smaller than $ M \log 1/\th$ for some $ M  > 0$ independent of
$\th,\,h$, so  our cut-offs $\chi^n$ will all be localized in
microscopic neighbourhoods 
of $\cK$ 
when $h\to 0$. Furthermore, for
such a logarithmic time the number of terms in the sum in the 
middle expression 
in \eqref{e:iterate} is
bounded above by $J^{n+1} \leq \th^{-N}$ for some $N>0$. As a result,
taking into account the above cut-off insertions, this sum can be rewritten as
\begin{multline} \label{e:sum2} 
\Pi e^{-int_0 P/h} \chi^w u \\ 
= \sum_{j_{n},\ldots,j_0} \Pi_{j_n}\cU_{j_n}^*\,
T_{j_nj_{n-1}}(\chi^{n-1})^w\cdots
T_{j_2j_1}(\chi^1)^wT_{j_1j_0}(\chi^0)^w \chi_{j_0} u_{j_0} +
\cO (\th^\infty) _{L^2 \to L^2 }
\,.
\end{multline}
In the next section we will carefully analyze the kernels of the
operators $T_{j'j}\,(\chi^{k})^w$.

\subsection{Structure of the local phase function}\label{s:phase}

To analyze the Fourier integral operators we will examine the structure of the generating
function for the symplectomorphism
$\kappa_{j_1j_0}$.

We start by studying the transverse
linearization $d_\perp \kappa(\rho)$ of the map
$\kappa=\kappa_{j_1j_0}$, for a point $\rho\in\cK\cap D_{j_1j_0}$. In
our symplectic coordinate frames, this transverse map is represented
by the symplectic matrix $S_{j_1j_0}(\rho)=S(\rho) \in {\rm{Sp}} (2d_{\perp}, \RR )$
given by 
\be\label{e:S(rho)}
S(\rho) \defeq
\frac{\partial(y^1,\eta^1)}{\partial(y^0 ,\eta^0)} 
(\rho)\,,\qquad \rho\in\cK\,.
 \ee
The linear symplectomorphism
$S(\rho)$ admits a quadratic generating function $Q_\rho(y^1,y^0,\theta')$,
where $\theta'\in\IR^{d_{\perp}}$ is an auxiliary variable: the graph of the map
${}^T\! (y^0,\eta^0)\mapsto{}^T\! (y^1,\eta^1)= S(\rho){}^T\!(y^0,\eta^0)$
can be obtained by identifying the critical set 
$$C_{Q_\rho}=\{(y^1,y^0,\theta')\,:\,
\partial_{\theta'} Q(y^1,y^0,\theta') = 0\}\subset \IR^{3d_{\perp}}. $$
This
critical set is in bijection with the graph of $S(\rho)$ through the
rules 
$$\eta^1=\partial_{y^1}Q_\rho (y^1,y^0,\theta'),\quad \eta^0 =
-\partial_{y^0}Q_\rho (y^1,y^0,\theta'),\quad (y^1,y^0,\theta')\in
C_{Q_\rho}\,.
$$
More structure comes from 
taking the normal hyperbolicity into account. Recall that our
coordinates are 
chosen so that that $E^+$
and $E^-$ are $\eps$-close to $ \{ \eta =  0 \} $ and $ \{ y = 0 \}$,
respectively. (Here we identified $ E^\pm $ with their images under $
d\kappa_j$ -- see \S \ref{dcc}.)
This implies the existence of a continous familty of 
symplectic transformation 
$$ \mathcal K \cap D_{j_1 j_0 } \ni \rho \longmapsto R(\rho)\in {\rm Sp}(2d_{\perp}, \RR ), 
$$ 
such that 
\begin{equation}
\label{eq:Rrh}   R ( \rho ) ( \{ \eta = 0 \} ) = E^+_\rho, \ \ \ 
R ( \rho ) ( \{ y = 0 \} ) = E^-_\rho , \ \ \
 R ( \rho ) = I + {\mathcal O} ( \epsilon ) . \end{equation}
Since $d_\perp\kappa(\rho)\equiv S(\rho)$ maps
$E^{\pm}_{\rho}$ to $E^{\pm}_{\kappa(\rho)}$, the matrix 
\be\label{e:tildeS}
\tS(\rho)\defeq
R(\kappa(\rho))^{-1}\,S(\rho)\,R(\rho),\qquad \rho\in \cK\cap D_{j_1j_0}\,,
\ee 
is block-diagonal:
\begin{equation}
\label{eq:Srh}
\tS(\rho) = \begin{pmatrix}\Lambda(\rho)&0\\0&
  {}^T\!\Lambda(\rho)^{-1} \end{pmatrix}\,.
\end{equation}
The normal hyperbolicity \eqref{eq:NH} implies that,
provided $t_0$ has been chosen large enough\footnote{Recall that
  $\kappa$ represents $\varphi_{t_0}$.}, the matrix $\Lambda(\rho)$
is expanding, uniformly with respect to  $\rho$:
\be\label{e:Lambda-expand}
\exists \ \nu>0,\quad\forall\rho\in \cK,\qquad \|\Lambda^{-1}(\rho)\|\leq e^{-\nu}<1\,.
\ee 
More precisely, for any small $\vareps>0$, if $t_0$ is chosen large enough the coefficient $\nu$ can be taken of the
form $\nu=t_0(\lambda_{\min}-\eps_0)$, where $\lambda_{\min}>0$ is the
smallest positive transverse Lyapunov exponent of $\varphi_t$ near $K^\delta$.

Combining \eqref{eq:Rrh},  \eqref{e:tildeS} and \eqref{eq:Srh} gives
\be
\label{eq:pkap}
S(\rho) 
 =    \begin{pmatrix}
  \Lambda(\rho)+\cO(\eps\Lambda)& \cO(\eps\Lambda)\\
\cO(\eps\Lambda) & \cO(\eps^2\Lambda + {}^T\!\Lambda(\rho)^{-1})\end{pmatrix}\,,\quad
\quad \rho\in \cK\,.
\ee

This explicit form, more precisely  the fact that the upper left block
is invertible, allows to use a special type of quadratic generating function:
\begin{lem}\label{l:Hmatrix}
If $t_0$ is chosen large enough, for each $\rho$ the generating
function $Q_\rho(y^1,y^0,\theta')$ can be chosen in the following form:
\be\label{e:Q_rho}
Q_\rho(y^1,y^0,\theta') = \tQ_\rho(y^1,\theta')-\la y^0,\theta'\ra
\,. 
\ee
For any point $(y^1,y^0,\theta')$ on the critical set $C_{Q_\rho}$, the auxiliary variable $\theta'$ is
identified with $\eta^0$ of the corresponding phase space point.  
\end{lem} 

The specific form of the generating function 
corresponds to the 
geometric fact that the graph of $ S ( \rho ) $ 
admits $(y^1,\eta^0)$ as coordinates (that is,
the graph of $ S ( \rho ) $ projects bijectively onto the $(y^1,\eta^0)$-plane).

The function $\tQ_\rho(y^1,\eta^0)$ can be written in terms of a
symmetric matrix $H(\rho)$:
\be\label{e:tQ_rho}
\tQ_\rho(y^1,\eta^0) = \frac12 \la
(y^1,\eta^0),{}^T H(\rho)(y^1,\eta^0)\ra, \ \ 
H(\rho)=\begin{pmatrix}H_{11}&H_{12}\\H_{21}&H_{22}\end{pmatrix}, \ \ 
H_{12}\ \text{invertible}\,.
\ee
The matrix $S(\rho)$ is related to $H(\rho)$  in the following way:
\be\label{e:kappa-H}
 S(\rho) = \begin{pmatrix} H_{21}^{-1}
   & - H_{21}^{-1}  H_{22} \\ H_{11} H_{21}^{-1} &  H_{12} - H_{11} H_{21}^{-1} H_{22} \end{pmatrix}\,.
\ee
Comparing with \eqref{eq:pkap} we see that
\be\label{e:H-Lambda}
H_{12}^T = H_{21} = \Lambda(\rho)^{-1} +
\cO(\eps\Lambda(\rho)^{-1}),\quad H_{11}=\cO(\eps),\quad H_{22}=\cO(\eps)\,,
\ee
uniformly with respect to $\rho$.
The quadratic phase function $Q_\rho$ will be relevant when
we consider the metaplectic operator $M(\rho)$ quantizing
$S(\rho)$ in \S\ref{s:transverse-phase} (see also \eqref{eq:simple}).

\medskip

From the study of the linearized flow in the transverse direction,
we now consider the dynamics of 
\begin{equation}
\label{eq:tka}  \tkappa=\tkappa_{j_1j_0} : D_{j_1 j_0} \cap {\mathcal K}
\longrightarrow A_{j_1 j_0} \cap {\mathcal K}  . \end{equation}
along the
trapped set -- see Fig.~\ref{f:da} in \S \ref{s:back-iterated}. When no confusion is likely to
arise we use the notation $ D_{\bullet} $ and $ A_{\bullet } $ for the 
corresponding subsets of $ \mathcal K$.
There we have {\em no} assumptions on the flow,  except
for it being symplectic. 

Possibly after refining the covers $U_j$, each map $\tkappa$ can be generated by a
nondegenerate phase function $\psi=\psi_{j_1j_0}(x^{1},x^{0},\theta)$ defined in a
neighbourhood of the origin in $\IR^{d-d_{\perp}}\times \IR^{d-d_{\perp}}\times
\IR^k $, where $ 0 \leq k \leq n $ -- see \S \ref{s:fio}. 

Since the $U_j$ have been chosen small, the map $C_\psi\to
\Gamma_\tkappa$ can be assumed to be injective. Notice that the values
of $\psi$  away from $C_\psi$ are irrelevant. 

We now want to extend $\psi$ into a generating function of the map
$\kappa$, at least in a small neighbourhood of
$\cK$. The intuitive idea is to ``glue together'' the generating function $\psi$
for $\tkappa$, with the quadratic generating functions $Q_\rho$ for the
transverse dynamics $d_\perp\kappa(\rho)$.

Let us consider the following Ansatz for a generating function $\Psi$ for $\kappa$:
\begin{equation}
\label{eq:defdP}
\Psi(x^1,x^0,\theta;y^1,y^0,\theta')=\psi(x^1,x^0,\theta) +
\delta\Psi(x^1,x^0,\theta;y^1,y^0,\theta'),
\end{equation}
with an additional auxiliary variable $\theta'\in \IR^{d_{\perp}}$. To simplify
notation we split the variables into 
 longitudinal and transversal ones:
\begin{equation}
\label{eq:parper}  \rho_\parallel=(x^1,x^0,\theta) , \ \ \
\rho_\perp=(y^1,y^0,\theta'). \end{equation}
\begin{lem}
Near any point $ \rho \in \mathcal K $, $ \kappa $ is generated 
by $ \Psi $ of the form \eqref{eq:defdP} with 
the transversal correction, $\delta\Psi(\rho_\parallel,\rho_\perp)$,
satisfying
$$
\delta\Psi(\rho_\parallel,\rho_\perp) =
Q_{\rho_\parallel}(\rho_\perp) + \cO((y^1,\theta')^3)\,,
$$
where $Q_{\rho_\parallel}(\bullet)$ is a quadratic form of the same
type as (\ref{e:Q_rho},\ref{e:tQ_rho}), which depends smoothly on
$\rho_\parallel$. If $\rho_\parallel\in C_\psi$ corresponds to the point $(\rho^1;\rho^0)\in \Gamma_{\tkappa}$,
then $Q_{\rho_\parallel}=Q_{\rho^0}$.
\end{lem}
In other words, the quadratic forms $Q_{\rho_\parallel}$ extend the forms
$Q_{\rho}$ to a neighbourhood of $C_\psi$.
\begin{proof}
Since $\cK$ is preserved by
$\kappa$ and carries the map $\tkappa$, we may assume that for any
$\rho_\parallel$, the function $\delta\Psi(\rho_\parallel,\bullet)$ has no linear part in the
variables $\rho_\perp$. 
At each point $\rho_\parallel\in C_\psi$ (identified with some  $\rho^0\in \cK$), the
quadratic part $Q_{\rho_\parallel}(\rho_\perp)$ generates the linear transverse
deviation from $\tkappa$ near the point $\rho^0$, namely $d_\perp
\kappa(\rho^0)$. This means that $Q_{\rho_\parallel}=Q_{\rho^0}$, which
has the form  (\ref{e:Q_rho}). This form corresponds to the 
geometric fact that the graph of
$d_\perp\kappa(\rho^0)$ admits $(y^1,\eta^0)$ as coordinates.

This projection property locally extends to the graph of
$\kappa$: in some neighbourhood of $\cK$, the points of
$\Gamma_\kappa$ can be represented by the coordinates
$(\rho^0=(x^0,\xi^0)\in\cK; y^1,\eta^0)$, where
$y^1,\eta^0\in\neigh(0)$. This property shows that $\delta\Psi$ can be
written in the form
\be\label{e:deltaPsi}
\delta\Psi(\rho_\parallel,\rho_\perp) =
\delta\tPsi(\rho_\parallel,y^1,\theta')  - \la y^0,\theta'\ra \,.
\ee
As explained above, the quadratic part $\tQ_{\rho_\parallel}(\bullet)$ of
$\delta\tPsi(\rho_\parallel;\bullet)$ must be equal, for
$\rho_\parallel\in C_\psi$, to the corresponding $\tQ_{\rho^0}$ generating
$S(\rho^0)$. 
The equations for $C_\Psi$ show that, if we fix small values
$(y^1,\theta'=\eta^0)$, then value $\rho_\parallel$
such that $(\rho_\parallel,y^1,y^0,\eta^0)\in C_\Psi$ is  
$\cO((y^1,\eta^0)^2)$-close to $C_\psi$.
\end{proof}

\subsection{Structure of the propagators $T_{j'j}$}

From the above informations about the phase function
$\Psi=\Psi_{j'j}$, we can write the integral kernel of  
$T=T_{j'j}$ defined in \eqref{e:T_jj} and quantizing the map
$\kappa_{j'j}$, as an oscillatory integral.
 The general theory of Fourier integral operators (see \S\ref{s:fio}) tells us that its kernel
takes the form
\be\label{eq:T(X,X)}
T(x^1,y^1;x^0,y^0)=\int_{\IR^{L+d_{\perp}}} 
\frac{d\theta\,d\theta'}{ (2\pi   h)^{(k+d_{\perp}+d)/2} }   
\, a(\rho_\parallel,\rho_\perp)\,
e^{\frac{i}{h}\Psi(\rho_\parallel,\rho_\perp)}  +
\cO (h^\infty)_{L^2 \to L^2}\,,
\ee
where we use the notation \eqref{eq:parper}.
Let us group the variables $(x,y)=X$, $(\xi,\eta)=\Xi$,
$(\theta,\theta')=\Theta$. 
We may assume
that the symbol $a(X^1,X^0,\Theta)$ is supported in a small
neighbourhood of the critical set $C_\Psi$. In particular, for small
values of the transversal variables $\rho_\perp$, $a(\bullet,\rho_\perp)$ is supported near $C_\psi$.
From \eqref{e:T_jj}, this Fourier integral operator is microlocally subunitary in
$V_{j'}\times V_{j}$.

\subsubsection{Using the cut-off near $\cK$}
We now take into account the cut-offs $(\chi^k)^w$, and study the truncated
propagator $T\,(\chi^{k})^w$  appearing in
\eqref{e:sum2}.
\begin{lem}
\label{l:nest}
For any $k\geq 0$ we have 
\be\label{e:chiTchi}
T\,(\chi^{k})^w = T^{\chi^k}+ \cO(\th^{\infty})_{L^2 \to L^2} \,,
\ee
where the Schwartz kernel of the operator $T^{\chi^k}$ is given by 
\be\label{e:T^sharp-chi}
T^{\chi^k}(x^1,y^1;x^0,y^0) \defeq \int\frac{d\theta\,d\eta^0}{(2\pi
  h)^{(k+d_{\perp}+d)/2} }\; a(\rho_\parallel,\rho_\perp)\,\chi^{\sharp(k+1)}(y^1)\,\chi^{k}(y^0,\eta^0)\,
e^{\frac{i}{h}\Psi(\rho_\parallel,\rho_\perp)}\,,
\ee
where  $ \chi^{\sharp   k} \defeq \chi^{k}|_{ \eta = 0 }  $, with $ \chi^k $ given in
\eqref{e:cutoffs}.  
\end{lem}
\begin{proof} 
As in
\eqref{e:insert-cutoff}, the nesting property
$\chi^{\sharp(k+1)}\succ\chi^k\circ\kappa_{j'j}$ and
the uniformity (in $ k $) of the symbol estimates on 
$ \chi^k $ imply that
\be\label{e:chiTchi0}
(\chi^{\sharp(k+1)})^w\, T\,(\chi^{k})^w = T\,(\chi^{k})^w
+ \cO(\th^\infty)\,,
\ee
uniformly for all $k\geq 0 $. (We recall that uniformity in $ k $ is
due to \eqref{e:cutoffs} and \eqref{e:rescaled-cutoffs} and 
the uniform error estimate comes from \eqref{eq:Egor12}.)
The Fourier integral operator calculus in the class $\tS_{\frac12}$ presented in
Lemma~\ref{l:fio} has the following consequence:
\[ (\chi^{\sharp(k+1)})^w\,T\,(\chi^{k})^w = T^{\chi^k}+
\cO(h^{\frac12}\th^{\frac12}), \]
which combined with \eqref{e:chiTchi0} gives \eqref{e:chiTchi}.
\end{proof}

\subsubsection{Rescaling the transversal coordinates}
Since we work at distances $\sim(h/\th)^{\frac12}$ from the trapped
set, it will be convenient to use the rescaled transversal variables
\be\label{e:rescaling}
\ty=(\th/h)^{\frac12}y,\qquad
\teta = (\th/h)^{\frac12} \eta.
\ee
Our cut-offs $\chi^{k}$, $\tchi^k$ defined in
(\ref{e:cutoffs},\ref{e:rescaled-cutoffs}) are then related by 
$\tchi^{\bullet}(\ty,\teta)=\chi^{\bullet}(y,\eta)$.
This change of variables induces the following unitary rescaling
$\cT 
\;  : \; 
L^2(dx  \,dy)\to L^2(dx\,d\ty)$:
\be\label{e:rescaling2}
\cT
u (x,\ty) \defeq 
(h/\th)^{{d_{\perp}}/2}\,u(x,(h/\th)^{\frac12}\ty) = (h/\th)^{
{d_{\perp}}/2}\,u(x,y)\,.
\ee
We recall (see for intance \cite[(4.7.16)]{e-z}) that
\[  \cT a^w (x,  y , hD_x,  h D_y ) \cT^* = \tilde a^w ( x, \tilde y ,
h D_x \tilde h
D_{\tilde y} ),  \ \ \ \tilde a ( x, \tilde y , \xi, \tilde \eta ) \defeq a (x,
y , \xi , \eta ) . \]

Through this rescaling, the operator
$T^{ \chi^k}$ is transformed into
$$\tT^{ \chi^k}\defeq\cT
T^{\chi^k}\cT^* 
\;  : \; L^2(dxd\ty) \longrightarrow  L^2(dxd\ty) , $$ 
with Schwartz kernel 
\be
 \begin{split} \label{e:Tchi-kernel}
 \tT^{ \chi^k}(x^0,\ty^0,x^1,\ty^1) = & \,
 \int_{\RR^k} \int_{\RR^{d_{\perp}} } \frac{d\theta}{(2\pi h)^{\frac{k+d-d_{\perp}}2}} \frac{d\teta^0}{(2\pi \th)^{d_{\perp}}}\
a(\rho_\parallel,(h/\th)^{\frac12}\trho_\perp)\, \\
& \ \ \times 
\tchi^{\sharp (k+1)}(\ty^1)\, \tchi^{k}(\ty^0,\teta^0)\,
e^{\frac{i}{h}\psi(\rho_\parallel) + \delta\Psi(\rho_\parallel;(h/\th)^{\frac12}\trho_\perp)}
\end{split} \ee

\subsubsection{Transversal linearization}

The factor $\tchi^{\sharp(k+1)}(\ty^1)\tchi^k(\ty^0,\teta^0)$ 
appearing in the integrand \eqref{e:Tchi-kernel}
allows us to simplify the above kernel. Indeed, it implies that
the variables $\trho_\perp=(\ty^1,\ty^0,\teta^0)$ are integrated over
a set of diameter $\sim R_1\Lambda^{k}$. One can then Taylor expand the amplitude and phase function
$\delta\Psi$ in \eqref{e:Tchi-kernel}:
\[ \begin{split}
& a(\rho_\parallel,(h/\th)^{\frac12}\trho_\perp)\,
e^{\frac{i}{h}\delta\Psi(\rho_\parallel;(h/\th)^{\frac12}\trho_\perp)}
\tchi^{\sharp(k+1)}(\ty^1) \tchi^{ k}(\ty^0,\teta^0)  
=\\
& \ \ \ \ \ \ \  \big(a(\rho_\parallel,0)+\cO_{\th,k}(h^{\frac12})_{S ( T^* \RR^d ) } \big)\,
 e^{\frac{i}{\th}Q_{\rho_\parallel}(\trho_\perp)}\,\tchi^{\sharp(k+1)}(y^1)
 \tchi^{k}(\ty^0,\teta^0)\,.
\end{split} \]
Since we will restrict ourselves to
values $k\leq M\log 1/ \th$, uniformly bounded with respect to $h$, we may omit
to indicate the $k$-dependence in the remainder.
As a result, up to a small error we may keep only the quadratic part of $\delta\Psi$, namely
consider the operator with the Schwartz kernel
\[
\int_{\RR^k} \int_{\RR^{d_{\perp}} }  \frac{d\theta}{(2\pi h)^{\frac{k+d-d_{\perp}}
      2} } 
\frac{d\teta^0}{(2\pi \th)^{d_{\perp}}}\
a(\rho_\parallel,0)\,\tchi^{\sharp(k+1)}(\ty^1)\tchi^{ k}(\ty^0,\teta^0)\,e^{\frac{i}{h}\psi(\rho_\parallel)}\,
e^{\frac{i}{\th}Q_{\rho_\parallel}(\trho_\perp)}\,.
\]
Combining the above pointwise estimates with the fact that
$a\in S(T^*\IR^{3d})$, and
with \eqref{e:chiTchi0} and Lemma \ref{l:nest}, gives
\be\label{e:Tchi-lin} 
\begin{gathered}
\tT^{ \chi^k}=\tT ( \tilde \chi^k ) ^w ( \tilde y , \tilde h
D_{\tilde y} ) +\cO (\th^{\infty})_{L^2\to L^2}\,, \\
\tT ( x^1, y^1; x^0 , y^0 ) = 
\int\frac{d\theta}{(2\pi h)^{(k+d-d_{\perp})/2}} \frac{d\teta^0}{(2\pi \th)^{d_{\perp}}}\
a(\rho_\parallel, 0)\,
e^{\frac{i}{h}\psi(\rho_\parallel)} e^{ \frac i \th Q_{\rho_\parallel}
  ( \tilde \rho^\perp ) } 
\end{gathered}
\ee
uniformly for $k\leq M\log 1/\th|$. 

\subsubsection{Factoring out the transversal contribution}\label{s:transverse-phase}

For each $\rho_\parallel\in\supp a(\bullet,0)$, the quadratic phase $Q_{\rho_\parallel}(\bullet)$ generates 
a symplectic transformation
$S(\rho_\parallel)$ (which, in the case $\rho_\parallel\in C_\psi$ corresponds
coincides with the transformation $S(\rho^0)$ of \eqref{e:S(rho)}). 
As already shown in \eqref{eq:simple}, this phase allows to represent the metaplectic operator
$M(\rho_\parallel):L^2(d\ty)\to L^2(d\ty)$ which $\th$-quantizes this symplectomorphism:
\be\label{e:metapl}
M(\rho_\parallel)(\ty^1,\ty^0)\defeq 
(2\pi   \th)^{-d_{\perp}}\int_{\RR^{d_{\perp}}}  \det(H_{12}(\rho_\parallel))^{1/2}\,e^{\frac{i}{\th}Q_{\rho_\parallel}(\trho_\perp)}
{d\teta^0}\,,
\ee
where $H_{12}(\rho_\parallel)$ is the block matrix appearing in
$Q_{\rho_\parallel}$, similarly as in (\ref{e:Q_rho},\ref{e:tQ_rho}). 

\begin{rem}\label{rem:sign}
In the expression \eqref{e:metapl} we implicitly chose a sign for the square root of
$\det(H_{12}(\rho_\parallel))$. Indeed, the metaplectic representation of the symplectic group
is 1-to-2, a given symplectic matrix $S$ being quantized into two
possible operators $\pm M$. The relations \eqref{e:H-Lambda} and the uniform
expansion property \eqref{e:Lambda-expand} show that
$\det(H_{12}(\rho_\parallel))$ does not vanish on the support of the amplitude
$a(\bullet,0)$ (which is a small neighbourhood of $C_\psi\times \{\ty^0=\ty^1=\teta^0=0\}$), so we
may fix the sign in each connected component of this support. This
remark will be relevant in \S\ref{s:back-iterated}.
\end{rem}

Defining the symbol 
$$
\ta(\rho_\parallel)\defeq \frac{a(\rho_\parallel,0)}{\det(H_{12}(\rho_\parallel))^{\frac12}}\,,
$$ 
we interpret the operator $\tT $ in 
\eqref{e:Tchi-lin} as a Fourier integral operator with an operator
valued symbol, $ M ( \rho_\parallel ) $, where $ M $ is given by 
\eqref{e:metapl}. That fits exactly in the framework presented
in Proposition \ref{p:meta}: 
\[   \widetilde T  ( u \otimes v ) ( x^1, \tilde y^1 )  = 
(2\pi h)^{-(k+d-d_{\perp})/2} \int_{\RR^k } \int_{\RR^{n} } 
\ta(\rho_\parallel)\, [ M
(\rho_\parallel)  v ] ( \tilde y^1 )
\,e^{\frac{i}{h}\psi(\rho_\parallel)}  u ( x^0 ) dx^0 d\theta . 
\]
We now apply Lemma \ref{l:fact} to see that 
\be\label{e:factor1}
\tT = \Op ( M ) T^{\parallel} +\cO_{\th}(h)_{ \mathcal D^{m + \ell}
  \to \mathcal D^{ \ell} } \,,
\ee
where $ m = m_{d - d_\perp } $ is defined in \eqref{eq:Egor12} and
where the Schwartz kernel of $ T^\parallel$ is given by 
\begin{equation}
\label{eq:Tpar}
T^{\parallel}(x^0,x^1) = (2\pi h)^{-k }\int_{\RR^k } 
\ta(\rho_\parallel)\,e^{\frac{i}{h}\psi(\rho_\parallel)} {d\theta} \,.
\end{equation}
The operator valued symbol $M(\rho^1)$ is the metaplectic operator $\th$-quantizing $S(\rho^0)$, where $\rho^1=\tkappa(\rho^0)$ and 
$S(\rho^0)$ is given in \eqref{e:S(rho)}.
We summarize these findings in the following 
\begin{prop}
Suppose that the Schwartz kernel of $ T $ is given by
\eqref{eq:T(X,X)}, $ \chi^k $,  $\widetilde \chi^k $ are given in
\eqref{e:cutoffs}, and $ \cT $ is the unitary rescaling defined in
\eqref{e:rescaling2}.  

Then for $ k \leq K ( \th )  $, where $ K ( \th ) $ may depend on 
$ \th $ but not on $ h $, 
\be\label{e:factorization}
\cT \left( T (\chi^k)^w \right) \cT^* = \Op ( M ) T^\parallel
(\tchi^k ) ^\tw  
+ \cO (\th^{\infty})_{ L^2  \to L^2 } + {\mathcal \cO}_{\th } ( h )_{ L^2  \to L^2 } \,,
\ee
where $ T^\parallel $ is given by \eqref{eq:Tpar} and 
$ M ( x^1 , \xi^1 ) $ given by \eqref{e:metapl} with 
$ \rho_\parallel \in  C_\psi $ determined by 
$ ( x^1,  \xi^1 ) = ( x^1, \partial_{x^1} \psi ( \rho_\parallel ) ) $. Here and below we use the abbreviation $(\tchi^k )^\tw \defeq (\tchi^k)^w(\ty,\th D_{\ty})$.
\end{prop}
\begin{proof}
Lemma~\ref{l:nest}, \eqref{e:Tchi-lin}, and \eqref{e:factor1} 
give \eqref{e:factorization} with the remainder 
\[  \mathcal O ( \th^\infty )_{ L^2 ( dx d \ty ) \to L^2 ( d
  x d \ty ) } + \mathcal O_\th ( h )_{
L^2 ( dx ) \otimes  \mathcal D^{m }  \to L^2 ( dx d \tilde y  )  } ( \tchi ^k )^\tw, \] 
where $ m = m_{d - d_\perp} $ is given in \eqref{eq:Egor12}. 
The definition of $ \tchi^k $ in \eqref{e:cutoffs} and
\eqref{eq:meta2} show that 
\[  ( \tchi^k )^\tw = \mathcal O ( \Lambda^{ 2 m k } ) :
L^2 ( d \ty )  \longrightarrow {\mathcal D}^{ m } \,, \]
and that gives the remainder in \eqref{e:factorization}.
\end{proof}

\subsection{Back to the iterated propagator}\label{s:back-iterated}

We can now come back to \eqref{e:iterate} and \eqref{e:sum2}, re-establishing the
subscripts $j_{k+1}j_k$ on the releavant objects. We rescale
all the operators by conjugating them through
$\cT
$. Fixing the limit indices
$j_0,j_n$, we want to study the sum of operators obtained by
conjugation of terms in \eqref{e:sum2} by $ \cT $:
\be \label{e:product1} \begin{split}
\cT [ \bT^{n} ]_{j_n j_0 } (\chi^0 )^w \cT^* & =  \cT
\left( \sum_{\bj}
 \prod_{k=n-1}^{0} T_{j_{k+1}j_k}\,(\chi^k)^w \right) \cT^* + {\mathcal O} (
\th^\infty )_{  L^2 \to   L^2  }
\\ & = 
\sum_{\bj} 
\tT_{j_nj_{n-1}} (\tchi^{n-1})^{\tw} \cdots
(\tchi^1)^{\tw}\tT_{j_1j_0} (\tchi^0)^{\tw} + {\mathcal O} (
\th^\infty )_{  L^2 \to   L^2  }
\end{split} \ee
where the sum runs over all possible sequences $\bj=j_{n-1}\ldots j_1$.
A sequence (which could be thought of geometrically as a path) $j_n\bj j_0$ will be relevant only
if it is {\em physical}, meaning that there exists points $\rho\in
K^\delta$ such that $\varphi_{kt_0}(\rho)\in U_{j_k}$ for all times
$k=0,\ldots,n$ (we say that the path $j_n\bj j_0$ contains the trajectory of
$\rho$). Any unphysical sequence leads to a term of order
$\cO(h^\infty)$.
On the other hand, for a given point $\rho\in
K^\delta$ there are usually many sequences $\bj$ containing its
trajectory, since the neighbourhoods $(U_j)$'s overlap, and so do the
cut-offs $(\pi_j)$.

For physical sequences $ j_n \bj j_0 $ we define the departure
set $ D_{j_n \bj j_0 } $ as the set of points $ \kappa_{j_0} (\rho) $, 
$ \rho \in U_{j_0} = \kappa_{j_0}^{-1} ( V_{j_0} ) $ such that
$ \varphi_{\ell t_0} ( \rho ) \in U_{j_\ell } $ for $ 0 \leq \ell \leq n $. We then put 
\begin{equation}
\label{eq:Dupn} 
D^n_{j_nj_0} = \bigcup_{\bj} D_{j_n \bj j_0} = \kappa_{j_0}
\left( \{\rho\in U_{j_0 }\cap K^\delta,\ \varphi_{nt_0}(\rho)
\in U_{j_n}\}\right)\,.
\end{equation}                  

We now simplify the expression 
\eqref{e:product1}, in the following way.
\begin{lem}\label{l:factor2}
In the notation of \eqref{e:iterate}  and\eqref{e:product1}, and for $
n \leq M \log 1/ \th $, 
\be\label{e:factor2}
\cT [ \bT^{n} ]_{j_n j_0 } (\chi^0 )^w \cT^*  
= \Op(M^{n}_{j_nj_0})  T^{n\parallel}_{j_n j_0} (\tchi^0)^\tw
\,+\cO(\th^\infty)_{L^2 \to L^2} \,.
\ee 
Here $T^{n\parallel}_{j_n j_0}$ is a Fourier integral operator on $L^2(dx)$ quantizing the
map $\tkappa^n:V_{j_0}\to V_{j_n}$,  defined on the 
departure set $D^n_{j_0j_n}$.
For each $\rho\in A^n_{j_n j_0} = \tilde \kappa^n (D^n_{j_nj_0}) $ (the arrival set) the operator
valued symbol $M^{n}_{j_nj_0}(\rho)$ is a metaplectic operator
quantizing the symplectic map 
$$S^n_{j_nj_0}( (\tilde \kappa^n )^{-1} (\rho)) =  d_\perp \kappa^n((\tilde \kappa^n )^{-1} (\rho)).$$
\end{lem}
\begin{proof}
If we insert the approximate factorizations \eqref{e:factorization} in
a term $\bj$ of the sum in the left hand side of \eqref{e:product1},
this term becomes
\be\label{e:sum2'}
\Op(M_{j_n j_{n-1}} )  T^{\parallel}_{j_{n}j_{n-1}} (\tchi^{n-1} )^\tw \,
\cdots \Op(M_{j_{1}j_{0}} ) 
T^{\parallel}_{j_{1}j_{0}} (
\tchi^{0} )^\tw + \cO ( \th^\infty )_{L^2 \to L^2} .
\ee
We now observe that just as we inserted the cut-offs $ \chi^k $  to obtain
\eqref{e:sum2} from \eqref{e:iterate} we can remove them 
so that each term becomes
\be
\label{eq:Twc}  \Op(M_{j_n j_{n-1}} )  T^{\parallel}_{j_{n}j_{n-1}}
\cdots \Op(M_{j_{1}j_{0}} ) 
T^{\parallel}_{j_{1}j_{0}}  ( \tchi^{0} )^\tw + \cO ( \th^\infty )_{L^2 \to L^2} .
\ee
We can now apply Lemmas \ref{l:prodop},\ref{l:fact} and Proposition \ref{p:meta} to
see that 
\be
\label{eq:TDD}  \begin{split} &  \Op(M_{j_n j_{n-1}} )  T^{\parallel}_{j_{n}j_{n-1}}
\cdots \Op(M_{j_{1}j_{0}} ) 
T^{\parallel}_{j_{1}j_{0}} = \\
& \ \ \ \ \ \ \ \ \  
\Op(M_{j_nj_{n-1}\cdots j_0} ) T^{\parallel}_{j_{n} \bj j_0} + 
 {\mathcal O}( \tilde h^{ -2 m_{d - d_{\perp}} n } h  )_{ L^2 ( dx ) \otimes
   {\mathcal D}^{ 2 n m_{ d-d_{\perp} } }  \to L^2 } , \end{split} \ee
where we use the shorthands
\begin{align*}
T^{\parallel}_{j_{n}j_{n-1}\cdots j_0}&\defeq T^{\parallel}_{j_{n}j_{n-1}}\,
T^{\parallel}_{j_{n-1}j_{n-2}}\cdots 
T^{\parallel}_{j_{1}j_{0}},\\ 
M_{j_nj_{n-1} \cdots j_0}   &\defeq 
( M_{j_nj_{n-1}} ) (M_{j_{n-1}j_{n-2}} \circ \tkappa_{j_{n-1} j_n } )   \cdots
( M_{j_{2}j_{1}}\circ \tkappa_{j_2 \cdots j_n} )\, 
(M_{j_{1}j_{0}} \circ \tkappa_{j_1 \cdots j_n} ) ,\\
\tkappa_{j_kj_{k-1}\cdots
  j_0}&\defeq\tkappa_{j_kj_{k-1}}\circ\tkappa_{j_{k-1}j_{k-2}}\cdots\circ\tkappa_{j_1j_0}.
\end{align*}
These expressions only make sense for physical sequences $j_n\bj
j_0$. The map $\tkappa_{j_n \bj j_0}$ is defined on
the departure set $D_{j_n \bj j_0}$.

\begin{figure}[ht]
\includegraphics[width=6.4in]{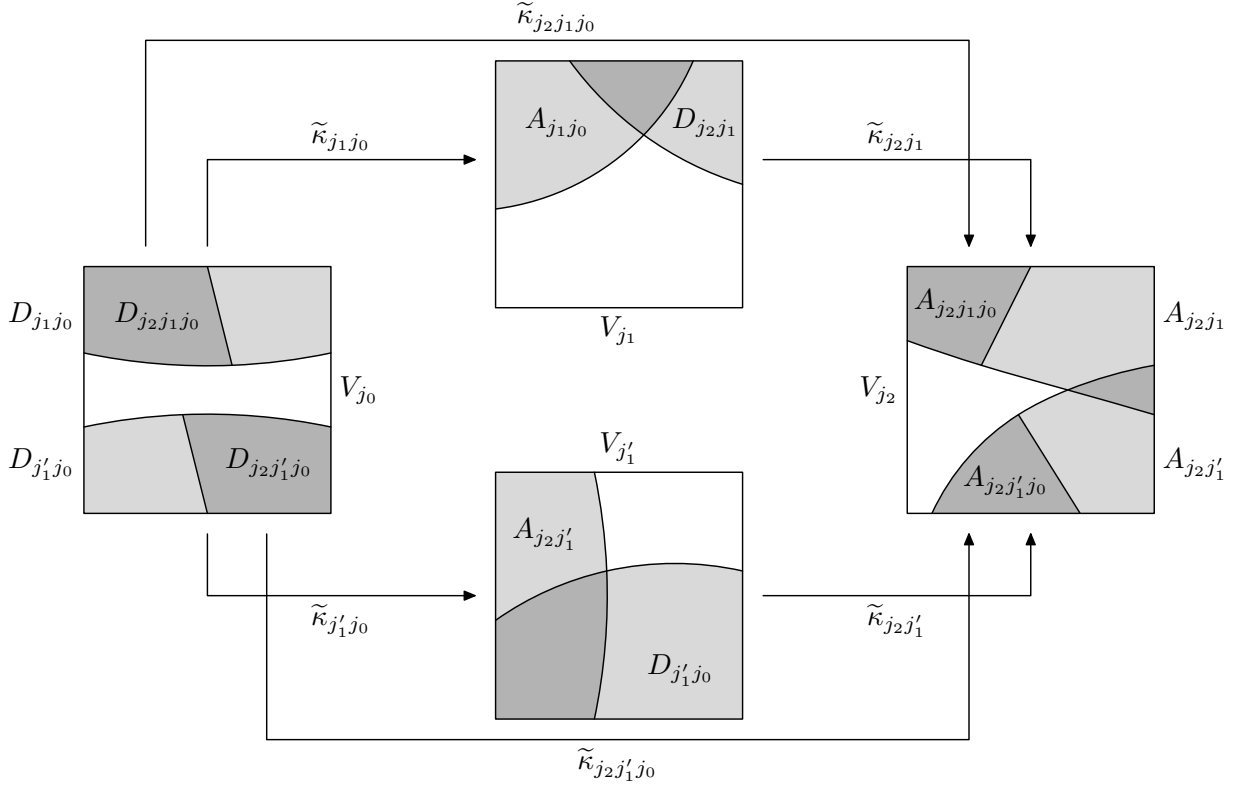}
\caption{Schematic representation of the departure 
and arrival sets for $ \bf j $ of lengths $ 1$ and $2$. We show two 
{\em physical} sequences $ j_2 j_1 j_0 $ and $ j_2 j_1' j_0 $ and
the corresponding maps \eqref{eq:tka}. As remarked there
we use the same notation for the departure and arrival sets
on $ \mathcal K$.}
\label{f:da}
\end{figure}

The metaplectic operator $M_{j_n \bj j_0}(\rho)$ quantizes the symplectomorphism
$S_{j_n \bj j_0}(\rho^0)$, with $\rho=\tkappa^n(\rho^0)\in A_{j_n\bj j_0}$. 
This symplectomorphism represents, in the
charts $V_{j_0} \to V_{j_n}$, the transverse linearization of the flow
$\varphi_{nt_0}$ at the point $\kappa_{j_0}^{-1}(\rho^0)$. As a
consequence, the symplectic matrix $S_{j_n \bj j_0}(\rho^0)$ is identical for all
sequences $j_n \bj j_0$ containing the trajectory of
$\rho^0$, and we call this matrix $S^n_{j_nj_0}(\rho^0)$. Hence, two metaplectic operators
$M_{j_n \bj j_0}(\rho)$, $M_{j_n \bj' j_0}(\rho)$ corresponding to two different allowed sequences
can at most differ by a global sign.

For all $\rho$ in the arrival set
$$ 
A^n_{j_nj_0} = \bigcup_{\bj} A_{j_n \bj j_0} = \kappa_{j_n}(\{\rho\in U_{j_n}\cap K^\delta,\ \varphi_{-nt_0}(\rho)\in
U_{j_0}\})\,, 
$$ 
we choose the sign of the metaplectic
operator $M^n_{j_nj_0}(\rho)$ quantizing $S^n_{j_nj_0}(\rho^0)$, such that $M^n_{j_nj_0}(\rho)$
depends smoothly on $\rho$ on each connected component of
$A^n_{j_nj_0}$ (there is no obstruction to this fact, due to the property mentioned in the
Remark~\ref{rem:sign}: the symplectomorphisms $S^n_{j_nj_0}(\rho)$
also have the form \eqref{eq:pkap}).
Hence, for each physical sequence $j_n\bj j_0$ we have 
\be\label{e:M-M^n}
M_{j_n \bj j_0}(\rho)=\vareps_{j_n \bj j_0}(\rho)
M^n_{j_n j_0}(\rho),\quad \rho\in D_{j_n \bj j_0}\,,
\ee
for some sign $\vareps_{j_n \bj j_0}(\rho)\in\{\pm\}$ constant on each
connected component of $A_{j_n \bj j_0}$. As before, the functions $\rho\mapsto
\vareps_{j_n \bj j_0}(\rho)$, $\rho\mapsto M_{j_n \bj j_0}(\rho)$ can be smoothly extended outside $A_{j_n
  \bj j_0}$, into compactly supported symbols. 
Lemma  \ref{l:prodop} and the identity \eqref{e:M-M^n} give
\be
\label{e:SNN}
\Op(M_{j_nj_{n-1}\cdots j_0} ) T^{\parallel}_{j_{n} \bj j_0} =
\Op(M^{n}_{j_n j_0}) \, (\vareps_{j_n \bj j_0})^w \,
T^{\parallel}_{j_{n} \bj j_0} +\cO_{\th}(h)_{ L^2 ( dx ) \otimes \mathcal
  D^{m_{d_{\perp}} } \to L^2 } \,.
\ee

When $ (\tchi^0)^\tw $ is inserted in \eqref{eq:TDD} and
\eqref{e:SNN} we apply
\eqref{eq:meta2} to see that 
\[   {\mathcal O}( \tilde h^{ - 2nm_{d - d_{\perp}} } h )_{ L^2 ( dx ) \otimes
   {\mathcal D}^{ 2n m_{ d-d_{\perp} } }  \to L^2 } ( \chi^0 )^\tw =
\mathcal O (  \tilde h^{ - 2nm_{d - d_{\perp}} } h  )_{ L^2 \to L^2 } =
\mathcal O_\th ( h )_{ L^2 \to L^2}, 
\]
and hence that error term can be absorbed into $ \mathcal O ( \th
^\infty ) $.

Returning to \eqref{e:sum2'} we see that 
the sum in the right hand side of \eqref{e:product1} 
can be
factorized in the following way:
\be\label{e:factorization2}
 \begin{split} & \sum_{\bj} \tT_{j_{n}j_{n-1}}(\tchi^{n-1})^{\tw}\cdots
\tT_{j_{1}j_{0}} (\tchi^0)^\tw 
 = \\
& \ \ \ \ \  \ \ \ \  \Op(M^{n }_{j_nj_0}) 
\big(\sum_{\bj}
T^{\parallel}_{j_{n}
{\bj}j_0}\,(\vareps_{j_n {\bj}j_0})^w\big)
 ( \tchi^0 )^\tw \
+ \cO(\th^\infty)_{L^2 \to L^2} \,,
\end{split}
\ee 
with a uniform remainder for $n\leq M\log 1/\th$.
Let us put $T^{n\parallel}_{j_n j_0}\defeq\sum_{\bj}
T^{\parallel}_{j_{n}
{\bj}j_0}(\vareps_{j_n \bj j_0})^w$, so that
the above identity reads exactly like in the statement of the Lemma.
The operator $T^{n\parallel}_{j_n j_0}$ is sum of Fourier integral
operators $T^{\parallel}_{j_{n}
{\bj}j_0}$
defined with different phase functions $\psi_{j_{n}
{\bj}j_0}$, yet
these phases generate (on different parts of phase space) the same map 
$\tkappa^n:D^n_{j_nj_0}\to A^n_{j_nj_0}$. Hence, $T^{n\parallel}_{j_n
  j_0}$ is a Fourier integral operator quantizing $\tkappa^n$.
This completes the proof of \eqref{e:factor2}. 
\end{proof}

The next lemma shows that the Fourier integral operator $T^{n\parallel}_{j_n j_0}$ is
essentially subunitary.
\begin{lem}\label{l:normTparallel}
Let $M>0$. For any small $\th>0$, there exists $h_0=h_0(\th)$ such
that, for any sequence $\bj$ of length $n\leq
M\log1/\th$ and any $h\leq h_0(\th)$, 
the operator $T^{n\parallel}_{j_{n}
j_0}$ satisfies the following norm estimate:
\begin{equation}
\label{eq:unif}
\|T^{n\parallel}_{j_{n}j_0}\|_{L^2(dx)\to L^2(dx)}\leq
1+\cO(\th )\,. 
\end{equation}
\end{lem}
\begin{proof}
We first note that we can bound the left hand side of \eqref{eq:unif}
by  $ \tilde h^{- C M }$, for some $ C $ -- that follows from a
trivial estimate of the terms $ T^\parallel_{ j_n \bj j_0 } $ in \eqref{e:factorization2}.

To prove \eqref{eq:unif} it is clearly enough to prove
the bound $ \|T^{n\parallel}_{j_{n}j_0} (\tchi^0)^\tw \|_{L^2(dx d
\tilde y )\to L^2(dx d \tilde y )}\leq 1+\cO(\th^\infty )$.
From Lemma \ref{l:normT} we know that $ \| {\mathbf T}^n \|_{ (L^2)^J \to (L^2)^J }
\leq  1 + {\mathcal O} ( h ) $, which implies that 
$ \| [ \mathbf T^n]_{j_0 j_n } \|_{ L^2 \to L^2 } \leq 1 + {\mathcal O} (
h ) $. 
Lemma \ref{l:factor2} then shows that
\begin{equation}
\label{eq:boundT} \| \Op(M^{n}_{j_nj_0})  T^{n\parallel}_{j_n j_0} (\tchi^0)^\tw\|_{
  L^2 \to L^2 } \leq 1 + \mathcal O ( \th^\infty ) . \end{equation}

The family of unitary metaplectic operators $ \rho \mapsto M_{j_n j_0}^n (\rho
) ^{-1} $ is well defined for $ \rho $ in the neighbourhood of the arrival
set $ A_{j_n j_0}^n $, and $
T^{n \parallel}_{j_n  j_0 } $  is microlocalized in any small
neighbourhood of $ A_{j_n j_0 }^n \times D^n_{j_n j_0 } \subset V_{j_n }
\times V_{j_0 } $.  Lemma \ref{l:prodop} and \eqref{eq:meta2} then show that
\[  \begin{split}  T^{n\parallel}_{j_n j_0} (\tchi^0)^\tw & = 
\Op ( (M^{n}_{j_n j_0 }) ^{-1})  \Op ( M^n _{ j_n j_0 } ) T^{n \parallel}_{    j_n j_0}  (\tchi^0) ^\tw  + \mathcal O_\th ( h \|
  T^{n \parallel}_{ j_n j_0 } \| )_{ L^2 ( dx ) \otimes \mathcal    D^{2m_{d-d_{\perp} } } \to L^2 } (\tchi^0)^w \\ 
& = \Op ( (M^{n}_{j_n j_0 }) ^{-1} ) \Op ( M^n _{ j_n j_0 }) T^{n \parallel}_{    j_n j_0} 
 (\tchi^0) ^\tw  + \mathcal O_\th ( h )_{ L^2 \to L^2 }
\,  . \end{split} \]
where we used the above a priori bound on $ \| T^{n \parallel}_{
  j_n j_0 } \|$.

Just as before we can insert the cut-off $ \tchi^n $ (see
\eqref{e:cutoffs})  with a $ \cO (
\th^\infty ) $ loss. We also introduce a cut-off $ \psi = \psi ( x,
\xi ) $ to a small neighourhood of $ A_{ j_n j_0 } $. (It was not
necessary before as $ T^{n \parallel} _{j_n j_0 } $ provided the
needed localization.) This and and \eqref{eq:boundT} give
the bound  
\[  \begin{split} 
 \|  T^{n\parallel}_{j_n j_0} (\tchi^0)^\tw \|  &   \leq \| \Op ( 
 (M^{n}_{j_n j_0 }) ^{-1} \psi ) (\tchi^n)^\tw  \| \|  \Op ( M^n _{ j_n
      j_0 } ) T^{n \parallel}_{    j_n j_0}  (\tchi^0) ^\tw \|  + \mathcal O (\th^\infty)  \\
& \leq  \| \Op ( (M^{n}_{j_n j_0 }) ^{-1} \psi ) (\tchi^n)^\tw \|
 ( 1 + \mathcal O (
    \th^\infty ) ) + \mathcal O ( \th^\infty ) . 
\end{split} \]
Since by Lemma \ref{l:prodop} and \eqref{eq:meta2}
\[ [ \Op ( (M^{n}_{j_n j_0 }) ^{-1} \psi )  (\tchi^n)^\tw ]^* 
\Op ( (M^{n}_{j_n j_0 }) ^{-1} \psi ) (\tchi^n)^\tw  =[ \psi^w ]^*
\psi^w  [(\tchi^n)^\tw ]^*
(\tchi^n)^\tw + {\mathcal O}_\th ( h )_{ L^2 \to L^2 } , \]
we have 
\[ \| \Op ( (M^{n}_{j_n j_0 }) ^{-1} \psi ) (\tchi^n)^\tw \| \leq \|
\psi^w \| \|
(\tchi^n)^\tw \| + \mathcal O_\th ( h ) \leq 1 + \mathcal O ( \th
), \]
and the bound \eqref{eq:unif} follows.
\end{proof}

\subsection{Inserting the final cut-off}
We now return to the operator $ \chi^w e^{ - i t n_0 P/h} \chi^w $.
From Lemma~\ref{l:iterate} we easily obtain
\be
\begin{split}
 \chi^w\, e^{-int_0 P/h} \, \chi^w u
&=\sum_{j_n,j_0} \Pi_{j_n}\cU_{j_n}^* \chi^w_{j_n}
\,[(\bT)^n]_{j_nj_0}\,\chi_{j_0}^w\,u_{j_0}  +
\cO 
 ( h^{\frac12} \tilde h^{\frac12 } )\\
&=\sum_{j_n,j_0} \Pi_{j_n}\cU_{j_n}^* \chi^w_{j_n}\,\,(\chi^0)^w
\,[(\bT)^n]_{j_nj_0}\,(\chi^0)^w\,\chi_{j_0}^w\,u_{j_0}  +
\cO ( \th^{\infty } )\,,
\end{split} 
\ee
where in the first line we used \eqref{e:chi-chi_j}, while in the second line we used \eqref{e:chi0chi}.
Hence our last step will consist in estimating the norm of the operator 
$(\chi^0)^{w} \,[\bT^n]_{j_n j_0}\, (\chi^0)^{w}$ (or its conjugate through $\cT$).
To this aim we will use Lemma \ref{l:fact}, Proposition \ref{p:meta}
and the factorization \eqref{e:factor2} to obtain
\be\label{e:factor21}
\begin{split} 
(\tchi^0)^{\tw} \cT \,[ \bT^n]_{j_n j_0} \cT^* \,(\tchi^0)^{\tw}  & =
(\tchi^0)^{\tw} \, \Op(M^n_{j_nj_0})
\,  T^{n\parallel}_{j_n   j_0}  (\tchi^0)^{\tw}
+  \cO(\th^\infty)_{L^2 \to L^2} \\
 & =  T^{n\parallel}_{j_n   j_0}  (\tchi^0)^{\tw} \, \Op(N^n_{j_nj_0})
\,(\tchi^0)^{\tw} +  \cO(\th^\infty)_{L^2 \to L^2} \,. 
\end{split} 
\ee
Here 
the operator valued symbol
$ N^{n}_{j_nj_0}(\rho)=M^n_{j_nj_0}((\tkappa^n)^{-1}(\rho))$, $\rho\in D^n_{j_nj_0}$, is a metaplectic operator quantizing the
symplectic map $S^n_{j_nj_0}(\rho) =  d_\perp \kappa^n(\rho)$. (Having
it  on the right now makes the notation slightly less cumbersome.)

In Lemma~\ref{l:normTparallel} we control the norm of
$T^{n\parallel}_{j_n j_0}$. There remains to control the norm of the factor
$(\tchi^0)^{\tw}\,\Op(N^n_{j_nj_0})\,(\tchi^0)^{\tw}$. For that it is 
enough to control the operator-valued symbol 
${\rm Op}^w_{\th,\ty}(\tchi^0)\,N^n_{j_nj_0}(\rho) \,{\rm Op}^w_{\th,\ty}(\tchi^0)$.

\subsubsection{Controlling the symbol}
In \eqref{eq:Rrh} we defined, for each point $\rho\in \cK\cap
D_{j_1j_0}$, 
a symplectic transformation $R(\rho)\in {\rm Sp}(2d_{\perp}, \RR )$ which maps the $y$-space to
$E^+_\rho$ and the $\teta$-space to $E^-_\rho$. This transformation
is $\eps$-close to the identity and in particular it is uniformly
bounded with respect to $\rho$. 

By iteration of this property, for any $\rho_0 \in D^n_{j_nj_0}$, 
the map $$\tS^n_{j_nj_0}(\rho_0)\defeq
R(\rho_n)^{-1}S^n_{j_nj_0}(\rho_0)R(\rho_0)$$ is
block-diagonal in the basis $(y,\eta)$:
\be\label{e:block-diagonal}
\tS^n_{j_nj_0}(\rho_0) = \begin{pmatrix}\Lambda^n(\rho_0)&0\\0&
  {}^T \!\Lambda^n(\rho_0)^{-1} \end{pmatrix}\,,
\ee
where $\Lambda^n(\rho_0)$ is expanding.
We may quantize $R(\rho)$ into metaplectic operators $A (\rho)$, and define
$$
\tN^n_{j_nj_0}(\rho_0) \defeq A(\rho_n)^{-1}\,N^n_{j_nj_0}(\rho_0)\,A (\rho_0)
$$
which quantizes $\tS^n_{j_nj_0}(\rho_0)$.

We can then rewrite
\begin{equation}
\label{eq:referee}
(\tchi^0)^\tw \,N^n_{j_nj_0}(\rho)\, (\tchi^0)^\tw = 
(\tchi^0)^\tw \,A(\rho_n)\,\tN^n_{j_nj_0}(\rho_0) \,A(\rho_0)^{-1}\,
(\tchi^0)^\tw \,.
\end{equation}
We are interested in the $L^2\to L^2$ norm of this operator. Since
metaplectic operators are unitary, and using the covariance of the Weyl
quantization with respect to metaplectic operators, this norm is equal to that of
$$
(\tchi^{0}_{\rho_n})^\tw \, (\tchi^{0}_{\rho_0} \circ
\tS^n_{j_nj_0}(\rho_0)^{-1})^\tw\,,\qquad
\tchi^{0}_{\rho_n} \defeq \tchi^0\circ R(\rho_n),\quad
\tchi^{0}_{\rho_0} \defeq \tchi^0\circ R(\rho_0).
$$
The block diagonal form of $\tS^n_{j_nj_0}(\rho_0)$ shows that
\be\label{e:product2}
[\tchi^{0}_{\rho_0}\circ (\tS^n_{j_nj_0}(\rho_0))^{-1}](\ty,\teta) =
\tchi^{0}_{\rho_0}(\Lambda^n(\rho_0)^{-1}\ty,{}^T\!\Lambda^n(\rho_0)\teta  )\,.
\ee

We may now invoke the following simple
\begin{lem}
\label{l:sympl}
Suppose that $ A $ is a $m\times m$ real invertible matrix and that $ \chi_1 ,
\chi_2 \in {\mathscr S } ( \RR^{2m} ) $. Then 
\begin{equation}
\label{eq:matrix}    \| \chi_1^w ( x, \tilde h D_x ) \chi_2^w ( A x,
{}^T\!A^{-1}  \tilde h D_x) \|_{ L^2 ( \RR^m ) \to L^2 ( \RR^m ) } \leq C  |\det A|^{\frac12}
\th^{- \frac m 2 } , \end{equation}
where $ C $ depends on certain seminorms of $ \chi_1 $ and $ \chi_2 $, but not on $ A $.
\end{lem}
We remark that the upper bound becomes nontrivial only if
$|\det A|\ll \th^{m/2}$. When that holds one cannot apply the
$\th$-symbol calculus any longer because the second
factor is not the quantization of a symbol in the class $ S ( \RR^{ 2 m} ) $,
uniformly in $ \th $ and $ A $.
When applicable, the symbol calculus would give the norm
equal to $ \max_{x,\xi}
|\chi_1(x,\xi)\,\chi_2(Ax,{}^T A^{-1}\xi)|+\cO(\th)$ -- 
see \cite[Theorem 13.13]{e-z}.

\begin{proof}
If we put $\hat\chi_j ( x ,Z ) \defeq \int_{\RR^m } \chi_j ( x , \xi ) e^{ i
  \langle Z , \xi \rangle} d \xi $, 
then the kernel of the operator in the lemma is given by 
\[ \begin{split}  K ( x, y ) & = 
\frac{ 1 }{ ( 2 \pi \tilde h )^{2m} } \int_{\RR^{3m} } 
 \chi_1 \left( \textstyle{\frac { x + z } 2} , \xi  \right) 
 \chi_2 \left ( \textstyle{\frac{ A z + A y } 2 }, {}^T\! A^{-1}  \eta
\right) e^{i \langle x - z , \xi \rangle/\tilde h  +  i  \langle
  z - y , \eta \rangle/ \tilde h  }  d \xi \, d \eta \, d z \\
& = 
\frac{ | \det A | } { ( 2 \pi \tilde h)^{2m} } \int_{\RR^m} 
\hat \chi_1 \left( \textstyle{\frac { x + z } 2} , \textstyle{\frac{
      x - z} {\tilde h}  }\right) 
\hat \chi_2 \left ( \textstyle{\frac{ A z + A y } 2} , \textstyle{
    \frac{ A z - A y } {\tilde h} }
\right) dz  .
\end{split}
\]
We will estimate the norm using Schur's Lemma and hence we need to
show that 
\begin{equation}
\label{eq:Schur}  \left(  \max_{x \in \RR^m } \int | K ( x, y ) | dy \right) 
\left( \max_{ y \in \RR^m } \int | K ( x, y ) | dx \right) \leq C^2 |
\det A |  \, \tilde h^{-m} . \end{equation}
Making a change of variables $ Z = ( x - z ) / \tilde h $ and $ X = ( x + z )
/ \tilde h $ we obtain
\[ \int | K ( x, y ) | d x \leq 
C_1 (\max_{ \RR^{2m} } | \hat\chi_2 |)\, | \det  A |  \th^{-m}\, \iint
  |\hat\chi_1 ( X , Z )| d Z d X \leq C { | \det  A |} {
\,     \tilde h^{- m }} .\]
To estimate the integral in $ y $ let 
\[   F ( Z ) =
\max_{\RR^m}  | \hat \chi_1 ( \bullet, Z ) | , \ \ 
G ( Y ) = \max_{ \RR^m } | \hat \chi_2 ( \bullet, Y ) |, 
\]
noting that our assumptions give $ F ( Z ) = {\mathcal O} ( \langle Z
\rangle^{-\infty } ) $, $ G ( Y ) = {\mathcal O} ( \langle Y
\rangle^{-\infty } ) $.
Changing variables to $ Z = ( x - z) / \tilde
h$ and $ Y = ( A z - A y ) / \tilde h $ we obtain,
\[  \int | K ( x , y ) | dy \leq C_3 \iint
F ( Z ) G ( Y ) d Z d Y \leq C . \]
This proves the upper bound \eqref{eq:matrix}.
\end{proof}

Applying Lemma \ref{l:sympl} to the product on the
right hand side of \eqref{eq:referee} we get the bound
$$
\|(\tchi^0)^\tw \,N ^n_{j_nj_0}(\rho_0)\, (\tchi^0)^\tw 
\|_{L^2(d\ty)\to L^2(d\ty)}\leq C(\tchi^{0}_{\rho_0},\tchi^{0}_{\rho_n})
|\det\Lambda^n(\rho_0)|^{- 1/2}\,\th^{-{d_{\perp}} /2}\,.
$$
Since the transformations $R(\rho)$ are uniformly bounded, the
prefactor $C(\tchi^{0}_{\rho_0},\tchi^{0}_{\rho_n})$ is uniformly bounded
with respect to  $\rho_0$. 
On the other hand, the determinant of $\Lambda^n(\rho_0)^{-1}$ can be
bounded as follows.
\begin{lem} 
Take $\eps_0>0$ arbitrary small. Then there exists
  $C_{\eps_0}>0$ such that, 
$$
\forall n\geq 1,\ \forall \rho_0\in D^n_{j_nj_0},\quad 
|\det \Lambda^n(\rho_0)^{-1}|\leq C_{\eps_0} e^{-(\lambda_0-\eps_0)nt_0}\,,
$$
where $ \lambda_0 $ was defined by \eqref{eq:t11}, and $t_0>0$ is chosen large enough, as explained in the comment following \eqref{e:Lambda-expand}.
\end{lem}
\begin{proof}
This follows from writing the definition of $\lambda_0$ using the local coordinate frames.
\end{proof}
We have thus obtained the following upper bound:
\be\label{e:bound2}
\|(\tchi^0)^\tw \,N^n_{j_nj_0}(\rho_0)\, (\tchi^0)^\tw 
\|_{L^2(d\ty)\to L^2(d\ty)} \leq C_\eps\ \th^{-d_{\perp}/2}\,e^{-(\lambda_0-\eps_0)nt_0}\,,
\ee
valid for any $n\geq 1$ and any $\rho_0\in D^n_{j_nj_0}$. In
particular, the time $n$ may arbitrarily depend on $\th$. 

When  $n\leq 
M \log 1/\th$, for $ M >0$ arbitrary large but independent of $\th$ or 
$h$, we combine this bound with \eqref{e:sum2}, Lemma
\ref{l:normTparallel} and Lemma \ref{l:factor2} to obtain 
the estimate \eqref{eq:tricky2}, which was the goal of this section.
\section{Microlocal weights and estimates away from the trapped set}
\label{micr}

In this section we will justify the estimates described as Step 2 of
the proof in \S \ref{out}. That will involve a quantization of 
the escape function $ G $ given in Proposition \ref{p:esc2}
with $ \epsilon = (h/\tilde h )^{\frac12} $. That means that we will
use the calculus described in \S \ref{12c}. 

\subsection{Exponential weights}
Suppose that $ g \in \CIc ( T^* X ; \RR ) $ satisfies
the following estimates:
\begin{equation}
\label{eq:assg} 
\frac{ \exp g ( \rho ) } { \exp g ( \rho' ) } \leq C \big( 1 + (
     \tilde h  / h )^{\frac12 }  d (\rho , \rho') \big)^N ,
\ \ \ \ 
\partial_\rho^\alpha g
= {\mathcal O} \big( ( h / \tilde h )^{-|\alpha|/2 } \big) ,  \ |
\alpha | >0  \,,
\end{equation}
for some $ N $ and $ C $, and for some distance function $d(\rho,\rho')$ on $ T^* X\times T^* X $
(since $ g $ is compactly supported, the estimate is independent of
the choice of $ d $ -- we can $ d $ to be the distance function 
given by a Riemannian metric). We note that $ G $ defined in Proposition \ref{p:esc2} with $ \epsilon
= ( h / \tilde h)^{\frac12} $ satisfies these assumptions.

We first recall a variant of the Bony-Chemin theorem
\cite[Th\'eor\`eme~6.4]{b-c},\cite[Theorem~8.6]{e-z}
in the form presented in \cite[Proposition 3.5, (3.21), (3.22)]{NSZ2} (as usual $g^w=\Op(g)$):
\begin{prop}
\label{p:bc}
Suppose that $ g\in \CIc ( T^* X )  $ satisfies \eqref{eq:assg}. Then
\begin{equation}
\label{eq:exOp}   \exp ( g^w) = b^w , \end{equation}
where the symbol $b(x,\xi)$ satisfies the bounds
\begin{equation}
\label{eq:weest}    | \partial^\alpha  b ( \rho ) |  \leq C_{ \alpha } \, e^{g ( \rho)}
\big( h / \tilde h \big) ^{ - |\alpha|/ 2 } 
\,, \end{equation}
in any local coordinates near the support of $ g $.

If $ \supp g \Subset U $, for an open $ U \Subset T^*X $, 
then 
\begin{equation}
\label{eq:locwe}  \partial_x^\alpha \partial_\xi^\beta  (  b (x, \xi
) - 1 )  =   {\mathcal O} ( h^\infty \langle \xi \rangle^{-\infty })  , \ \ ( x, \xi ) \in \complement U .\end{equation}

Also, if $ A \in \Psi^{\rm{comp} }  ( X ) $,  $ B \in \widetilde
\Psi_{\frac12}^{\comp} ( X ) $ and $ C \in \Psi_{\frac12}^{\comp} ( X)  $ then 
\begin{equation}
\label{eq:ABC}    \begin{split} 
& e^{  g^w  } A  e^{ - g^w } = A + i (  h \tilde h
)^{\frac12} A_1 , \ \  A_1 \in \widetilde \Psi_{\frac12}^{\comp} ( X ) ,  \ \ 
\WF_h ( A_1 ) \subset \WF_h ( A )  ,\\
& e^{ g^w} B  e^{ - g^w } = B + i \tilde h
B_1 , \ \  B_1 \in \widetilde \Psi_{\frac12}^{\comp} ( X ) , \ \ \WFh
( B_1 ) \subset \WFh ( B ) , \\
& e^{ g^w } C  e^{ - g^w } = C + i \tilde
h^{\frac12} 
C_1 , \ \  C_1 \in \Psi^{\comp}_{\frac12} ( X )  , \ \ \WFh ( C_1 )
\subset \WFh ( C ) .
\end{split}
\end{equation}
\end{prop}

The assumptions in \eqref{eq:assg} show that $ \exp g $ is an
order function for the $ \widetilde S_{\frac12} $ calculus -- see
\cite[\S 3.3, (3.17),(3.18)]{NSZ2}. Hence we can apply composition formulae.
In particular if $ g_j $, $ j=1,2$ satisfy
\eqref{eq:assg} then 
\begin{equation}
\label{eq:g1g2}
   \exp (g_1^w) \exp (g_2^w) = c^w , \ \ \ 
| \partial^\alpha c ( \rho ) | \leq  C_\alpha \exp ( g_1 + g_2 ) 
\big( h / { \tilde h} \big) ^{ - |\alpha|/ 2 }  . 
\end{equation}
Because of the compact supports of $ g_j $'s and because of
\eqref{eq:weest} derivatives can be taken in any local coordinates.

The  consequence of \eqref{eq:g1g2}  useful to us here is given in the following Lemma.
\begin{lem}
\label{l:L2}
Suppose that $   A \in \widetilde \Psi^{\rm{comp}}_{\frac12} (X )   $ and 
that 
\[  \widetilde \sigma  ( A ) = a + {\mathcal O}\big( ( h \tilde h)^{\frac12}  \big) _{ \widetilde
  S_{\frac12} }, \ \ \ 
 a \in \CIc ( T^* X )   \cap   \widetilde S_{\frac12} ( T^*X ) . \]
If $ U_{h , \tilde h } \defeq \{ \rho \in T^*X
  : d ( \rho, \supp a ) < ( h/\tilde h)^{\frac12} \} $, 
then 
\begin{equation}
\label{eq:normc}
\| A\,  e^{g_1^w} e^{ g_2^w} \|_{ L^2 \to L^2}  
= \sup_{ T^* X} (| a| e^{ g_1 + g_2 }
) + {\mathcal O} ( \tilde h 
\sup_{ U_{ h , \tilde h }   }  e^ { g_1 + g_2})   + {\mathcal O} ( 
h^{\frac12} \log ( 1/h) ) .
\end{equation}
\end{lem}
\begin{proof}
We first consider this statement in $ \RR^n $. 
We apply the standard rescaling \eqref{eq:stre} noting
that \eqref{eq:assg} imply that $ \tilde m_j = \exp \tilde g_j  $
are order functions.  If $ d $ is the Euclidean distance and if we put
\[  n_N  (\tilde \rho ) \defeq ( 1 + d ( \tilde \rho ,\widetilde U )
)^{-N} , \ \ \ 
\widetilde U  \defeq ( \tilde h / h
)^{\frac12} U_{ h ,\tilde h } , 
\]
then $ n_N $ is an order function for any $ N $, and $ \tilde a \in S (  n_N )
$ for all $ N$.
We have 
\[  A = \Op(a + ( h \tilde h )^{\frac12} a_1) , 
 \ \ \text{ for some 
$  a_1
\in \widetilde S_{ \frac12} $,} \] 
and hence, after rescaling,
\begin{gather*}   \tilde A \,  e^{ \Opt ( \tilde g_1 ) } e^{ \Opt ( \tilde g_2 ) } = 
\Opt ( \tilde b )  + ( h \tilde h)^{\frac12} \Opt (\tilde b_1 ), \\   
\tilde b \in S ( n_N \tilde m_1 \tilde m_2 ), \ \ \  \tilde b  - \tilde a e^{\tilde g_1 + \tilde
  g_2 } \in \tilde h S (  n_N \tilde m_1 \tilde m_2 ) , \ 
\ \ \tilde b_1 \in S ( \tilde m_1 \tilde m_2 ) . 
\end{gather*}
Put 
\[ \begin{split}   M = M ( h , \tilde h ) & \defeq \sup_{\IR^{2n}} n_N \tilde m_1
  \tilde m_2 \leq \sup_{\tilde \rho} \Big( \big( 1 + d ( \tilde \rho, \widetilde
  U )\big) ^{-N} e^{ \tilde g_1 ( \tilde \rho)  + \tilde g_2 ( \tilde \rho)
  } \Big) \\
& \leq \Big( \sup_{\widetilde U} e^{ \tilde g_1  + \tilde g_2  }
\Big) \left( 1 +
  \sup_{\tilde \rho}  ( 1 + C_1 C_2 d (  \tilde \rho ,
  \widetilde U )) ^{- N + N_1 + N_2 } \right) \\
&  \leq C \sup_{ U_{ h , \tilde h }   }  e^ { g_1 + g_2}  , 
\end{split} \]
where we took $ N \geq N_1 + N_2 $, with $N_j $, $ C_j $ appearing
in \eqref{eq:assg} for $ g_j $. 

We now apply  \cite[Theorem 13.13]{e-z} (with $ h $ replaced by $
\tilde h$) to $ \tilde b/ M \in S $. That gives
\[  \| \Op ( \tilde b ) \| = \sup |a| e^{ g_1 +  g_2 } 
+ {\mathcal O} ( \tilde h ) \sup_{ U_{ h , \tilde h }   }  e^ { g_1 +
  g_2}   .\]
Since $\tilde m_1 \tilde m_2 = {\mathcal O} ( \log ( 1/h ) ) $, 
applying the same argument to $ \tilde b_1 / \log ( 1/ h ) $
gives \eqref{eq:normc}. 

The calculus is invariant modulo $ {\mathcal O} ( ( h \tilde h
)^{\frac12} ) $ terms (see \eqref{eq:wisi} and \cite[\S
5.1]{Da-Dy},\cite[\S 3.2]{WuZ}), so
these local estimates on $\IR^n$ imply similar estimates on manifolds.
\end{proof}

The next result is a version of \eqref{eq:Egor} for exponentiated
weights $ g $.  It is a special case of \cite[Proposition 3.14]{NSZ2}
which follows from globalization of the local result 
 \cite[Proposition 3.11]{NSZ2}. 
We  state it using concepts recalled in \S\ref{s:fio}.
\begin{prop}
\label{p:EgorE}
Suppose that $ T \in I^{\rm{comp} } ( X \times X , \Gamma_{\kappa}' )
$ where $ \kappa : U_1 \to U_2 $, $ U_j \subset T^*X $, is a
symplectomorphism, 
that $g \in \CIc ( T^*X ) $ satisfies \eqref{eq:assg}, and that $ A
\in \widetilde \Psi^\comp_{\frac12} $.
Then 
\be
\label{eq:EgorE} 
 \begin{gathered}
 e^{ g^w } A  T =  T e^{ ( \kappa^* g)^w} B 
+  h^{\frac12}
  \tilde h^{\frac12} T_1 e^{  ( \kappa^* g)^w } C , \\
 T_1 \in I^{\rm{comp}}_h ( X\times X, \Gamma_\kappa' 
), \ \ B , C \in \widetilde \Psi_{\frac12} ( X ) ,  \ \ \sigma ( B )=
\kappa^* \sigma ( A ) . 
\end{gathered} \ee
\end{prop}

\subsection{Estimates away from the trapped set}
\label{eat}

We now provide precise versions of the estimates \eqref{eq:wei}
and \eqref{eq:wei0} described in the Step 2 of the proof in 
\S \ref{out}. 

For the escape function $ G $ constructed in Proposition \ref{p:esc2} we define the operator
\begin{equation}
\label{eq:defGw}  G^w \defeq \Op ( G ) \in  \log ( \tilde h/ h ) \widetilde
\Psi_{\frac12}^{\rm{comp} } ( X ) , \ \ \widetilde \sigma (G) = G + 
{\mathcal O} \big( ( h \tilde h)^{\frac12 -} \big) _{ \widetilde S_{\frac12} } . 
\end{equation}
Since $ G $ satisfies \eqref{eq:assg}, Proposition \ref{p:bc}
describes the exponentiated operator $ e^{G^w}=e^{\Op(G)} $.
We refer to Remark \ref{rem:constants} for the requirements
on the constants in the definition of $ G $. Intuitively, 
 $ G $ is bounded (independently of $ h $ and $ \th$)
in a $(h/\th)^{\frac12} $-neighbourhood of $ \mathcal K$, 
and satisfies the growth condition $ G( \varphi_{t_0} ( \rho ) )
- G ( \rho ) \geq 2 \Gamma $ outside of a {\em smaller} 
$ (h/\th)^{\frac12} $-neighbourhood of $ \mathcal K$.

The first lemma shows that the weights are bounded near the
trapped set:
\begin{lem}
\label{l:eat1}
Suppose that $ \chi \in \CIc ( T^* X ) \cap \widetilde S_{\frac12} (
T^*X ) $ has the property
\begin{equation}
\label{eq:eat0}  \supp \chi \subset  \{ \rho \in T^*X : d ( \rho, K^{2 \delta} ) < C_0 (
h/ \tilde h)^{\frac12}  \}, 
\end{equation}
for some constant $ C_0 $ satisfying $ 0 < (C_0+1)^2 < c_2 L $, 
in the notation of \eqref{eq:es4}.

Then for some constants $ h_0, \tilde h_0 , C_1 > 0 $ 
we have for $ 0 < h < h_0 $, $ 0 < \tilde h < \tilde h_0 $, 
\begin{equation}
\label{eq:eat1}    \| \chi^w e^{ G^w} \|  \leq C_1 , \ \ \| e^{ G^w}
\chi^w \| \leq C_1 .
\end{equation}
\end{lem}
\begin{proof}
Since   $ \widetilde \sigma ( \chi^w  ) = \chi +  {\mathcal O} (  { h^{\frac12} \tilde h^{\frac12} } ) _{
  \widetilde S_{\frac12}} $,
and 
$  | G ( \rho ) | \leq C_3 \Gamma C_2 $ for $ d ( \rho, K^{2 \delta }
  ) < ( C_0 + 1 )   ({ h}/{\tilde h})^{\frac12} $ 
(see \eqref{eq:es4} and \eqref{eq:defG}), 
the estimates in \eqref{eq:eat1} follow directly from 
Lemma \ref{l:L2}.  
\end{proof}

The main result of this section provides bounds for the
conjugated propagator. It relies heavily on the material about the
propagator for the complex absorbing potential (CAP) modified
Hamiltonian,  $ \exp ( - i t ( P - i W ) / h ) $, presented in the Appendix.
\begin{prop}
\label{p:eat}
Suppose that $ G^w $ is given by \eqref{eq:defGw} and
that $ A\in\Psi^{\rm{comp}}(X) $ satisfies
\begin{equation}
\label{eq:suppsio}  \WFh ( A ) \subset  p^{ -1} ( ( - \delta, \delta )) \cap w^{-1} ( [ 0 , \epsilon_1 ) ) , \end{equation}
for some $ \epsilon_1 > 0 $.

Then for some constants $ h_0, \tilde h_0 , C_1 > 0 $ 
we have for $ 0 < h < h_0 $, $ 0 < \tilde h < \tilde h_0 $,  
\begin{equation}
\label{eq:wei00} 
\|   e^{ - G^w } e^{ - i t_0 ( P - i W ) / h } e^{ G^w } A \| \leq e^{
2  C_1 } . \end{equation}
If $ \chi $ satisfies \eqref{eq:eat0} and in addition 
\begin{equation}
\label{eq:eat2} 
\chi ( \rho ) \equiv 1 \ \ \text{ for $ d ( \rho , K^{ 2 \delta } ) <
  \textstyle{\frac12} C_0 ( h/ \tilde h)^{\frac12} $, $ |p ( \rho )|
  \leq \delta $,}
\end{equation}
where $ C_0 $ is a large constant dependending on $ t_0 $,
then, if $ \|A \| \leq 1 $, 
\begin{equation}
\label{eq:wei1}  \| ( 1 - \chi^w  ) e^{ - G^w } e^{ - i t_0 ( P - i W
  ) / h } e^{ G^w }  A  \| <  e^{-
 \Gamma }  , 
\end{equation}
where $ \Gamma $ is the constant appearing in the definition
\eqref{eq:defG} of 
$ G $.
\end{prop}
\begin{proof}
Let $ A_{-G} \defeq e^{ G^w } A e^{ - G^w } $. Then \eqref{eq:ABC} in Proposition
\ref{p:bc} shows that 
\begin{equation}
\label{eq:Ati}  A_{-G } = A + {\mathcal 
O}_{ L^2 \to L^2 } ( h^{\frac12}  ) = {\mathcal O}( 1 )  _{L^2 \to L^2 } 
\ \text{ and } \
A_{ - G}  = \widetilde A A_{-G}  + {\mathcal O}( h^\infty ) , 
\end{equation}
where $ \widetilde A $ satisfies \eqref{eq:suppsi}.
To prove \eqref{eq:wei00} we use the notation of Proposition 
\ref{l:mod1}, 
and rewrite the operator on the right hand side as 
\begin{equation}
\label{eq:split1}  \begin{split} 
  e^{ - G^w } e^{ - i t_0 ( P - i W ) / h } e^{ G^w } A & = 
 e^{-G^w} e^{ - i t_0 P/h} e^{ G^w }   e^{ - G^w}
V
_{\widetilde A} 
 ( t_0 )   A_{-G} e^{G^w  } + {\mathcal O}
( h^{\frac12} ) _{L^2 \to L^2 }  \\
& = 
 e^{-G^w} e^{ - i t_0 P/h} e^{ G^w }  C (t_0) + {\mathcal O} (
 h^{\frac12} ) _{L^2 \to L^2}\, ,
\end{split} \end{equation}
where using \eqref{eq:ABC} and 
Proposition \ref{l:mod1}, 
\[  C ( t_0 ) \in \Psi_{\frac12} ^\comp( X) , \
\ \WFh ( C(t_0)  ) 
\subset \WFh ( A ) \cap w^{-1} (0) . \]

Since 
\[    e^{ \pm G^w }  = B e^{ \pm  G^w} + ( I - B ) + {\mathcal O} ( h^\infty ) _{L^2 \to L^2}
,  \ \text{ for some $ B \in
  \Psi^\comp ( X ) $,}  \]
Proposition \ref{p:EgorE} (applied with $ A \equiv I $) and
\eqref{eq:Aexp}  show that for
some $ B_0
\in \widetilde \Psi_{\frac12} ( X ) $, 
\[  e^{ - G^w } e^{ - i t_0 P /h} e^{ G^w } =
e^{ - i t_0 P/h } e^{ -( \varphi_{t_0}^* G )^w } \, e^{ G^w  } \left( I +
h^\frac12 \tilde h^{\frac12} B_0 \right) + {\mathcal O}
( h^\infty ) _{L^2 \to L^2 }. \]
From this and \eqref{eq:split1} we see that to prove \eqref{eq:wei00} it is enough to show that
\begin{gather}
\label{eq:comple}
\begin{gathered}     e^{ -( \varphi_{t_0}^* G )^w } \, e^{ G^w  } B_1  = 
{\mathcal O} ( 1 ) _{L^2 \to L^2 } , \ \ B_1 \in \Psi^\comp ( X) , \\ \WFh ( B_1
) \subset  p^{ -1} ( ( - \delta, \delta )) \cap w^{-1} ( [ 0 , \epsilon_1 ) ) .
\end{gathered} 
\end{gather} 
Lemma \ref{l:L2} applied with $ g_1 = - \varphi_{t_0}^* G $ and $ g_2
= G$, and the property $ G  - \varphi_{t_0}^* G \leq C_7 $ in
\eqref{eq:es3}  which holds in a neighbourhood of $ \WFh ( B_1 ) $, 
give \eqref{eq:comple} and hence \eqref{eq:wei00}.

To obtain \eqref{eq:wei1} we proceed similarly but applying the property
$ \varphi_{t_0}^* G - G \geq 2 \Gamma $ 
which is valid outside a $ (h/\tilde h)^{\frac12} $
neighbourhood of $ K^\delta $ -- see \eqref{eq:es3}.  In more detail,
Proposition \ref{p:EgorE} applied with $ A = 1 - \chi^w $
gives\footnote{Strictly speaking
$ 1- \chi^w \notin \widetilde \Psi_{\frac12}^{\comp } $ but the 
operator $ A \in \Psi^{\comp} $ provides the needed localization:
we can write $ A = A_0 A + \mathcal O ( h^\infty )_{ L^2 \to L^2 } $
where $ \WF_h ( \Id - A_0 ) \cap \WFh ( A ) = \emptyset $ and apply 
Proposition \ref{p:bc} to $ A_0 $.}
\[  \begin{split} 
& ( 1- \chi^w )   e^{ - G^w } e^{ - i t_0 ( P - i W ) / h } e^{ G^w } A  = \\
& \ \ \ ( 1  - \chi^w )  e^{-G^w} e^{ - i t_0 P/h} e^{ G^w }   e^{ - G^w}
V _{\widetilde A} 
 ( t_0 )   A_{-G} e^{G^w  }+ {\mathcal O}( h^{\frac12}
) _{L^2 \to L^2
}  = \\
& 
\ \ \ e^{ - i t_0 P/h}  e^{ -( \varphi_{t_0}^* G )^w}\, e^{  G^w   }
  ( 1 - ( \varphi_{t_0}^* \chi)^w ) 
e^{ - G^w} V _{\widetilde A}   ( t_0 )   A_{-G} e^{G^w  }
+ {\mathcal O} ( h^{\frac12} ) _{L^2 \to L^2 } \, ,
\end{split} \]
where we used the boundedness established in \eqref{eq:wei00} to
control the lower order terms. Defining $ \chi_1 \defeq \varphi_{t_0}^*
\chi $, we have, by the invariance of $ K^{\delta }$ under the flow,
\[   \chi_1 \equiv 1  \ \text{ for $ \ d ( \rho, K^\delta ) \leq C_1 (
  h/\tilde h )^{\frac12}$,  $ \ |p( \rho )| \leq \delta $. } \]
Let $ \psi \in \CIc ( T^*X ) $ be equal to $ 1 $ in the set $ W_1 $ of
Proposition ~\ref{p:esc2}, and
$ \supp \psi \subset ( w^{-1} ( 0 ) )^\circ $. 

Since \eqref{eq:ABC} and Proposition \ref{l:mod1} give
\[ \begin{split}  \| e^{ - G^w} V _{\widetilde A}   ( t_0 )   A_{-G}
  e^{G^w  }\| & \leq 
\| e^{ - G^w} V _{\widetilde A}   ( t_0 )   e^{G^w  }\| \| A\| 
\leq 
\| A \| \left( \| \tilde A \| + \mathcal O_{L^2 \to L^2 } ( \tilde
  h^{\frac12} ) \right)
\\
& 
\leq  1 + {\mathcal O} ( \th^{\frac12} ) , \end{split} \]
it is enough to show that
\begin{gather}
\label{eq:gat1}   \|    e^{ -( \varphi_{t_0}^* G )^w} e^{ G^w } ( 1 - \chi_1^w ) \psi^w
\|\leq e^{ -3\Gamma/2 }, \\  
\label{eq:gat2} \|    e^{ -( \varphi_{t_0}^* G )^w } e^{ G^w } ( 1 -
\chi_1^w )  ( 1 - \psi^w ) 
B \| \leq C h^\frac12 \log(1/h) 
\end{gather}
for $ B \in \Psi^\comp ( X ) $ with $ \WFh ( B ) \subset
w^{-1} ( [ 0
, \epsilon_1/ 2]) \cap p^{-1} ( [ - \delta, \delta ] ) $, is as in 
Proposition ~\ref{p:esc2}.
Both inequalities follow from Lemma \ref{l:L2} and properties of 
$ G $ in \eqref{eq:es3}. For \eqref{eq:gat1} 
we apply \eqref{eq:normc}. For  \eqref{eq:gat2} 
we note that
\[   \varphi_{t_0}^* G - G \geq C_8 \log ( \tilde h / h ) , 
 \ \text{ on }  \  \supp ( 1 - \psi ) \cap \WFh (B)  , \]
and \eqref{eq:normc} gives the estimate with the error dominating
the leading term. 
\end{proof}


\section{Proof of Theorem \ref{t:1}}
\label{pr}
We first prove \eqref{eq:main} which we rewrite as follows
\begin{equation}
\label{eq:main1}
\| U_G^n A \|_{L^2 ( X ) \to L^2 ( X ) } \leq C e^{ - n
  t_0 ( \lambda_0 - \epsilon_0) / 2 } , 
\ \  M_{\epsilon_0} \log \frac 1 {\tilde h } \leq n  \leq M \log \frac 1
{\tilde h} 
\end{equation}
where
\[ 
  U_G \defeq \exp ( - i t_0 \tP_G / h ) A  = e^{ - G^w } e^{ - i t_0 ( P - i W )
  / h } e^{ G^w } , \]
with $ t_0 $ chosen in previous sections, and 
\begin{equation}
\label{eq:WFA}
A \in \Psi^\comp ( X ) , \ \ \WFh ( A )
\subset p^{-1} ( ( - \delta, \delta ) )  .
\end{equation}

To apply the estimates of the last two sections we first
observe that Proposition \ref{l:mod2} implies that 
for any $ r $ there exist $ B_j \in \Psi^\comp $, $ j = 1, \cdots r 
$, each satisfying 
\eqref{eq:WFA}, such that 
\begin{equation}
\label{eq:onemore}
U_G ^r A = \prod_{ j=1}^{r } U_G B_j  + {\mathcal O} (
 h^\infty ) _{L^2 \to L^2} , \ \ B_r = A ,
\end{equation}
where the constants in the norm estimate $ {\mathcal O} (h^\infty ) $
depend on $ r $.
This means that, for $ r $ independent of $ h$ but depending on $
\tilde h $, 
 $ U_G^r A $ can be replaced by the product of
operators $ U_G B_j $, to which estimates of the previous 
section are applicable.

We now want to decompose $ U_G^n $ in such a way that the estimates
obtained in \S\S \ref{ats},\ref{micr} can be used.
For that we define 
\begin{equation}  U_{G}=U_{G,+}+U_{G,-},\qquad U_{G,+}\defeq U_{G}\chi^{w},\qquad
U_{G,-}\defeq U_{G}(1-\chi)^{w} .  \label{eq:Upm}
\end{equation}
We note that Proposition \ref{p:bc}  shows that
\begin{equation}
\label{eq:noWW} \begin{split} 
  \chi^w e^{- G^w }  e^{ - i t ( P - i W  ) /h } e^{ G^w }& =  
e^{-G^w }  e^{G^w}   \chi^w e^{- G^w }  e^{ - i t ( P - i W  ) /h }
e^{ G^w } \\
& = e^{-G^w} \chi^w e^{ - i t ( P - iW ) /h } e^{G^w} + {\mathcal O} (
\tilde h^{\frac12} h^{\frac12} )_{L^2  \to L^2}  \\
& = e^{-G^w} \chi^w e^{ - i t ( P - iW ) /h} e^{ i t P/h } e^{ - it P/h
  } e^{G^w} + {\mathcal O} (
\tilde h^{\frac12} h^{\frac12} )_{L^2  \to L^2}  \\
& = e^{-G^w} \chi_t^w e^{ - it P/h
  } e^{G^w} + {\mathcal O} (
\tilde h^{\frac12} h^{\frac12} )_{L^2  \to L^2}  , 
\end{split} \end{equation}
where  $\chi_t^w \defeq \chi^w e^{ - i t ( P - iW ) /h} e^{ i t P/h } $.
We now use  Proposition \ref{l:mod1} applied with 
with $ P $ replaced by $ - P $,
$ A \in \Psi^\comp $ satisfying $ \WFh ( I - A) \cap \WFh ( \chi^w ) =
\emptyset $. In the notation of \eqref{eq:VAt},  $ \chi_t^w = \chi^w V_A ( t )^* $, $ V_A ( t )^*
\in \Psi_\gamma^\comp ( X ) $. From \eqref{eq:VAt} 
\[  \sigma ( V_A ( t ) ) = \exp \left( - \frac 1 h \int_0^t \varphi_{-s}^*
W \right) \sigma ( A ) ,  \]
with a full expansion of the symbol in any coordinate chart given in
Lemma \ref{l:iw}. 
For $ \rho \in \supp \chi $, $ d ( \rho , K^\delta ) = {\mathcal O} (
h^{\frac12} )$, and as $ K^\delta $ is invariant under the flow
$ d ( \varphi_{-s}  ( \rho ) , K^\delta ) = {\mathcal O}_s (
h^{\frac12} ) $.  But that means that on the support $ \chi $, 
$ \varphi_{-s}^* W \equiv 0 $ for $ s \leq t $, where $ t $ is
independent of $ h $, as long as $ h $ is small enough.  This means
that $ \WFh ( I - V_A ( t ) ^* ) \cap \WFh ( \chi^w ) = \emptyset $ 
and hence, for all $ t $, 
$ \chi_t  = \chi + \Oo_t ( h^{\frac12} )_{S_{\frac12} } $.

Returning to \eqref{eq:noWW} this means that 
for $ t \leq C \log ( 1/\tilde h ) $ (in fact for any time bounded
independently of
$ h $), we have
\begin{equation}
\label{eq:noW}  \begin{split} 
&    \chi^w e^{- G^w }  e^{ - i t ( P - i W  ) /h } e^{ G^w } = 
\chi^w e^{- G^w }  e^{ - i t P  /h } e^{ G^w }  + \Oo_t (
 h ^{\frac12} )_{L^2 \to L^2 } , \\
& e^{- G^w }  e^{ - i t ( P - i W  ) /h } e^{ G^w }   \chi^w= 
e^{- G^w }  e^{ - i t P  /h } e^{ G^w }\chi^w  + \Oo _t(
 h ^{\frac12} )_{L^2 \to L^2 } . 
\end{split} \end{equation}

Using the notation \eqref{eq:Upm}
 \begin{equation}
\begin{split} 
U_{G}^{n} & =\sum_{\vareps_{i}=\pm}U_{G,\eps_{n}}\cdots
U_{G,\eps_{2}}U_{G,\eps_{1}}\\
& =\sum_{\bep\in\Sigma(n)}U_{\bep}, \qquad
U_{\bep} \defeq U_{G,\eps_{n}}\cdots
U_{G,\eps_{2}}U_{G,\eps_{1}} , 
\end{split}
\label{eq:decompo1}\end{equation}
where we used the symbolic words
$\bep=\eps_{1}\cdots\eps_{t}\in\Sigma(n)=(\pm)^{n}$. 
Now, for each word $\bep\neq--\cdots--$, call $n_{L}(\bep)$ (
$n_{R}(\bep)$, respectively), the number of consecutive $(-)$ starting from the
left (the right, respectively): 
\[ \bep=\underbrace{-\cdots-}_{n_{L} ( \bep )
}+**\cdots**+\underbrace{-\cdots-}_{n_{R} ( \bep ) }.\]
Given integers $n_{L},n_{R}$, call $\Sigma(n,n_{L},n_{R})$ the set
of words $\bep\in\Sigma(n)$ such that $n_{L}(\bep)=n_{L}$ and $n_{R}(\bep)=n_{L}$.
The decomposition (\ref{eq:decompo1}) can be split into
\[
U_{G}^{n}=U_{G,-}^{n}+\sum_{n_{L},  n_{R} }
\; \sum_{ {\bep\in
 \Sigma(n,n_{L},n_{R})}}U_{\bep}. \]
where the sum runs over $n_{L},n_{R}\geq0$ such that 
$ { n_{L}+n_{R}\leq n-1 } $. 

We make the following observations:
\[ \begin{split}
&  \Sigma( n, n_L, n_R ) = \left\{  (-)^{n_{L}}\!+\! (-)^{n_{R}}  \right\}  , \
\text{ if $ n_L + n_R = n- 1 $, }  \\
&  \Sigma( n , n_L, n_R ) = \left\{ 
(-)^{n_{L}}\!+\!\bep' \!+\! (-)^{n_{R}} : \bep'  \in\Sigma(n-n_{L}-n_{R}-2) \right\}, \ \text{ if $ n_L + n_R < n-1 $.}
\end{split} \]
Hence, the above sum can be recast into
\begin{equation}
\label{eq:recast}
\begin{split} 
U_{G}^{n} & =
U_{G,-}^{n}+\sum_{n_{L}=0}^{n-1}\left(U_{G,-}\right)^{n_{L}}U_{G,+}\left(U_{G,-}\right)^{n-n_{L}-1}
\\ 
& \ \ \ \ \
+\sum_{n_{L},n_{R}
}\left(U_{G,-}\right)^{n_{L}}U_{G,+}\left(U_{G}\right)^{n-n_{R}-n_{L}-2}U_{G,+}\left(U_{G,-}\right)^{n_{R}}
\end{split} 
\end{equation}
where the last sum runs over $n_{L},n_{R}\geq0$ such that $
n_{L}+n_{R}\leq n-2$.

The following lemma provides the estimate for terms in the last sum
on the 
right hand side of \eqref{eq:recast}:
\begin{lem}
\label{l:oneless}
For $\tilde{h}>h>0$ small enough, the following bound holds for $
r_0 \leq r \leq C_{0} \log(1/\tilde{h})$, $ r \in \NN $, 
\begin{equation} 
\left\Vert U_{G,+}\, U_{G}^{r}\, U_{G,+}\right\Vert_{L^2 \to
 L^2} 
\leq
C\,\th^{-d_{\perp}/2}\, 
\exp \left(     - \textstyle{\frac12}  t_0 r ({\lambda_0 -\eps}) \right)
, \label{eq:bound3-bis}\end{equation}
where the constant $C$ is uniform with respect to $h$, $\tilde{h}$ and $r$. \end{lem}
\begin{proof}
Lemma \ref{l:eat1} shows that $ e^{G^{w}}\chi^{w}, \ \chi^{w}e^{-G^{w}} = {\mathcal O} ( 1 ) _{L^2 \to
  L^2} $.
Also, Lemma \ref{l:L2} shows that for $ \chi_1 \in \tS_{\frac12}  $ with the same properties as $
\chi $ but equal to $ 1 $ on the support of $ \chi $,  
we have 
\[
\chi^{w}e^{-G^{w}}=\chi^{w}e^{-G^{w}}\chi_{1}^{w}+\cO (\tilde{h}^{\infty})_{L^2 \to L^2},
\ \ 
e^{G^{w}}\chi^{w} = \chi_{1}^{w}e^{G^{w}}\chi^{w} +\cO (\tilde{h}^{\infty})_{L^2 \to L^2} . \]
Using \eqref{eq:noW} 
the operator on the left hand side of \eqref{eq:bound3-bis} can be
rewritten as 
\[ \begin{split}
U_{G}\, \chi^{w}\left(U_{G}\right)^{r+1}\chi^{w} & =
U_G \,\chi^{w}\,e^{-G^{w}} \,e^{ - i (r+1) t_0 P /   h }\,
e^{G^{w}} \, \chi^{w} + {\mathcal O}( h^{\frac12} ) _{L^2 \to L^2} \\
 & = U_{G} \big( \chi^w  e^{-G^{w}}\chi_1^{w}+\cO(\tilde{h}^\infty 
)\big) e^{ - i (r+1) t_0 P /   h }\big(\chi_1^{w}e^{G^{w}} \chi^w
+\cO(\tilde{h}^\infty )\big)  \\ &  \ \ \ \ \ \ \ \ \ \ \ \ \ \ \ \ \ \
+ 
 {\mathcal O}( h^{\frac12}
) _{L^2 \to L^2} \end{split} 
\]
Hence,
\[
\begin{split}
\left\Vert U_{G,+}U_{G} ^{r}U_{G,+}\right\Vert_{L^2 \to
 L^2}  & \leq \left\Vert U_{G}\right\Vert \left\Vert \chi^{w}e^{-G^{w}}\right\Vert
\left\Vert \chi_{1}^{w}e^{ - i (r+1) t_0 P /   h }
  \chi_{1}^{w}\right\Vert \left\Vert e^{G^{w}}\chi^{w}\right\Vert
+\cO(\tilde{h}^{\infty}) \\ & \leq C\,\left\Vert \chi_{1}^{w}e^{ - i (r+1)
    t_0 P /   h }  \chi_{1}^{w}\right\Vert , 
\end{split} \]
where we used the fact that the operators $U_{G}$, $e^{G^{w}}\chi^{w}$
and $\chi^{w}e^{-G^{w}}$ are uniformly bounded on $L^{2}$. We can
now apply Proposition~\ref{p:crucial-bound}, 
replacing $t$ by
$(r+1)t_0 $ and $\chi$ by $\chi_{1}$. 
\end{proof}

Let us now take $n=C_{0}\log1/\tilde{h}$, with $C_{0}\gg1$. We recall
that $ \Gamma $ in \eqref{eq:wei} was assumed to satisfy $ \Gamma > t_0
{\lambda_0}/{2}$. We will use the bounds (\ref{eq:bound3-bis}),
and Proposition \ref{p:eat}: $\left\Vert U_{G,-} A \right\Vert < 
e^{-\Gamma }$.

Returning to the estimate for $ U_G^n A $ we 
first observe that \eqref{eq:onemore} and the estimates
\eqref{eq:wei1} in Proposition \ref{p:eat} give
\begin{equation}
\label{eq:wei11}  \| U_{ G,-}^m A \| \leq e^{-m \Gamma} + {\mathcal O}_r
( h^{\frac12} ) . 
\end{equation}

In \eqref{eq:recast}, for each $\ell=1,\ldots,n-2$, we group together terms with  $n_{L}+n_{R}=\ell$, and apply Lemma \ref{l:oneless} and \eqref{eq:wei11}:
\begin{align*}
\|  U_{G}^{n} A \|  & \lesssim e^{-n\Gamma}+n\,
e^{-(n-1)\Gamma} + \th^{-d_{\perp}/2} \sum_{\ell=1}^{n-2}( \ell + 1 )
\, e^{-\ell \Gamma} \, e^{ -t_0 (n-\ell )\frac{\lambda_0-\eps}{2}}  + \Oo (h^{\frac12} )\\
& \lesssim n\, e^{-n\Gamma} + \th^{-d_{\perp}/2} \,e^{-t_0 n\frac{\lambda_0-\eps}{2}}
 \sum_{\ell=1}^{n-2} ( \ell+ 1) \,
 e^{-\ell \big(\Gamma-t_0\frac{\lambda_0-\eps}{2}\big)} + \Oo (h^{\frac12} ) \\
 & \lesssim \th^{-d_{\perp}/2}\,e^{-t_0 n\frac{\lambda_0-\eps}{2}}.
 \end{align*}
 By taking $C_{0} =M_{\eps}  \gg 1/{\eps} $ we may absorb the prefactor $\th^{-d_{\perp}/2}$ and 
obtain, for $\tilde{h}>0$ small enough, 
\begin{equation} 
\label{eq:UGA1} 
\| U_{G}^{n} A \|  \leq C\,
\exp\big(-nt_0 \frac{\lambda_0-2\eps}{2}\big),\quad n\approx
M_{\eps}\log1/\tilde{h}\,.
\end{equation}

We can now complete the proof of \eqref{eq:t1} following the outline
in \S \ref{out}.  We first note that 
\eqref{eq:UGA1} gives \eqref{eq:main}, so that (see 
\eqref{eq:intt}) for 
\begin{equation}
\label{eq:zzz}   z \in [ - \delta/2  , \delta/2 ] - i h [ 0 , (
\lambda_0 - 3\epsilon_0 ) / 2 ]
\end{equation}  and $ A \in \Psi^\comp ( X ) $ satisfying
\eqref{eq:WFA}, 
\begin{gather*} ( \tP_G - z ) Q_A ( z ) =  A - R( z ) , \ \  R(z ) =  {\mathcal O} ( \tilde h ) _{L^2 \to L^2} ,
\\ Q_A ( z ) \defeq \frac i h \int_0^{T( \tilde h ) } e^{ - i t (    \tP_G - z )    / h } A dt = 
{\mathcal O}  \left( \frac{ T( \tilde h )} h \right) _{L^2 \to L^2 },
  \ \ T( \tilde h ) = M_{\eps_0} \log 1 /{\tilde h } . 
\end{gather*}
We now apply this estimate with  $ A \in \Psi^\comp ( X ) $ such that $ \sigma ( A ) \equiv 1 $ in $ p^{ -1} (
-3\delta/ 4, 3 \delta/4 ) \cap  w^{-1} ( [0, \epsilon_1 ) ) $. Then 
$ \tP_G - z \in \widetilde \Psi_{\frac 12} ^m ( X) $ is elliptic
outside of $ \WFh ( A ) $. Hence, using the $ \widetilde
\Psi_{\frac12} $ calculus of \S \ref{12c},  there exists $ \widetilde Q_A ( z
)\in  \widetilde \Psi_{\frac12} ^{-m} ( X ) $ such that 
\[ ( \tP_G - z ) \widetilde Q_{A} ( z ) = I - A + \widetilde R ( z ) , \
\ 
\widetilde R ( z ) = {\mathcal O} ( \tilde h ) _{L^2 \to
  L^2} . \]
The Fredholm operator 
$ \tP_G - z $ has index $ 0$ since $ \tP_G + i $ is invertible for small $
\tilde h$. 
It follows that for $ \tilde h $ small enough and $ z $ satisfying \eqref{eq:zzz}
\[  ( \tP_G - z )^{-1} =  ( Q_A ( z ) + \widetilde Q_A ( x ) ) ( I + R (
z) + \widetilde R_A ( z ) )^{-1} = {\mathcal O} ( 1/h)  \]

Since $ e^{ \pm G^w ( x, h D )} = {\mathcal O} ( h^{-M/2+1}) _{L^2 \to L^2}
$ for some $ M$, it follows that
\[ \begin{split} &  ( P - i W - z )^{-1} = {\mathcal O}  ( h^{-M} ) _{L^2 \to L^2} \, \ \ 
   z \in [ - \delta/2  , \delta/2 ] - i h[ 0 , (\lambda_0 - 3\epsilon_0 ) / 2 ]  , , \\ 
& ( P - i W - z ) ^{-1} = {\mathcal O} ( 1/ \Im z ) _{L^2 \to L^2}  , \ \ 
\Im z > 0 , \end{split} \]
where the second is immediate from non-negativity of $ W $ as an operator.

We now use a semiclassical maximum principle \cite[Lemma
4.7]{Bu},\cite[Lemma 2]{TZ} to obtain the bound for $ ( P - i W - z ) ^{-1} $
in \eqref{eq:t1} (after adjusting $ \delta $ and $ \epsilon_0 $). 

\medskip
\noindent
\begin{rem}  Strictly speaking we proved \eqref{eq:t1} for $ z \in [
- \delta/2 , \delta/2 ] - i h [ 0 , \lambda_0 /2 - \epsilon_1 ] $, for
any  $\epsilon_1 $, provided that $ h $ is small enough. 
\end{rem}


\section{The CAP reduction of scattering problems: Proof of
  Theorem \ref{t:2}} 
\label{recap}

In this section we will prove a generalization of Theorem \ref{t:2}
which applies to a variety of scattering problems.  Our approach of
reduction to estimates for the Hamiltonian  complex absorbing
potential (CAP) 
is based on the work Datchev--Vasy \cite{DatVas10_1}
(see also \cite[\S 4.1]{DyaGu}) but as 
the argument is simple and elegant we reproduce it in our slightly 
modified setting.

Let $ ( Y, g )  $ be a complete Riemannian manifold and
let 
\begin{equation}
\label{eq:P0}  P_g = -h^2 \Delta_g + V , \ \ \  V \in \CI ( Y ; \RR)
. 
\end{equation}

We make general assumption on $ ( Y, g ) $ which will allow
asymptotically Euclidean and asymptotically hyperbolic infinities. 

We assume that $ Y $ is the interior of a compact manifold $\overline Y$ with
a $ \CI $ boundary,   $ \partial Y  \neq \emptyset $.  We choose a defining function
of $ \partial Y $:
\begin{equation}
\label{eq:rho}   
\rho \in \CI (\overline Y;[0,\infty)), \ \ \  \{\rho
  = 0\} = \partial Y , \ \ \  d\rho |_{\partial  Y} \ne 0 .
\end{equation}

Let $p_g = | \xi |_g^2 + V ( x ) $ be the principal symbol of
$P_g$ and let 
\[  ( x ( t ) , \xi ( t ) ) = \exp t H_{p_g } ( x ( 0) ,
\xi ( 0 ) ) , \]
be the Hamiltonian flow (geodesic flow lifted to $ T^*Y $ when $ V
\equiv 0 $.  The first assumption on $ ( Y , g ) $ we make is a
non-trapping (convexity) assumption near infinity formulated using $
\rho $ with properties \eqref{eq:rho}:
\begin{equation}
\label{eq:conve}
 \rho ( x ( t ) )  \in  ( 0 , \epsilon_1) ,  \ \  \frac {d } { dt } \rho ( x
 (  t) )= 0  \ \Longrightarrow \ 
 \frac {d^2 } { dt^2 } \rho ( x  (  t) ) < 0 . \end{equation}

The trapped set at energy $ E \in [ - \delta, \delta ]  $ is defined as 
\[
( x, \xi )  \in K_E  \Longleftrightarrow p_g(x, \xi) = E \textrm{ and
} \exp(\RR H_{p_g}) ( x, \xi )  \textrm{ is compact in } T^*Y.
\]
We assume that the trapped set at energies $ | E| \leq \delta $, 
(see \eqref{eq:trapped1}), 
\begin{equation}
\label{eq:Kdell}  \text{ $ K^\delta$  is {\em normally}  {\em hyperbolic} in
  the sense of \eqref{eq:NH} .} 
\end{equation}

We now make analytic assumptions on $ P $. For that we first
assume that $ P_g $ can be modified inside a compact part of $Y$, to
obtain an operator 
\[   P_\infty = - h^2 \Delta_g + \widetilde V , \ \ \ \widetilde
V \in \CI ( Y ) ,   \ \ \ \widetilde V\rest_{ \rho < \epsilon_1 } = V\rest_{ \rho < \epsilon_1 }  , 
\]
with the 
following properties: for some $ s_0 > 0 $ and $ C_0 > 0 $, 
\begin{equation} 
\label{eq:prope1}
\|    \rho^{s_0} ( P_\infty - E - i 0 )^{-1} \rho^{s_0} \|_{
L^2   ( Y ) \to L^2  ( Y ) }  \leq \frac  {C_0}  h   , \  \ \ |E | \leq \delta , 
\end{equation}
and 
\begin{equation}
\label{eq:prope2}
\begin{split} & u = (P_\infty - E - i 0)^{-1} f , \ \ f \in \CIc ( Y )  \
  \ \Longrightarrow  \\
& \ \ \ \ \ \ \ \  \WFh ( u ) \setminus \WFh ( f ) \subset \exp ( [ 0 , \infty )H_{\Re p_\infty} ) \left(
  \WFh ( f ) \cap {p_\infty^{-1} ( E )} \right) , 
\end{split}
\end{equation}
where $ p_\infty \defeq | \xi|_g^2 + \widetilde V $.

We note that these assumptions do not require that the resolvent of
$P_\infty $ has a meromorhic continuation from $ \Im z > 0 $ to
the lower half-plane. A stronger conclusion will be
possible when we make that assumption: more precisely, for $ \chi \in \CIc ( Y ) $, we assume that
the resolvent $ ( P_\infty - z )^{-1} $ continues from $ \Im z >
0 $ {\em analytically} to  $ [ - \delta , \delta] - i h [  0 , C_0 ]
$, for
some $ C_0>0 $, and that for some $ N$, the following resolvent estimate holds:
\begin{equation}
\label{eq:prope3}
 \chi ( P_\infty  - z )^{-1} \chi = \Oo_{ L^2 \to L^2 } ( h^{-N
}) , \ \ z \in [ - \delta, \delta ] - i h [ 0, C_0 ] \,. 
\end{equation}
When $ P_\infty $ is chosen to be selfadjoint, interpolation
\cite[Lemma 4.7]{Bu},\cite[Lemma 2]{TZ} shows that \eqref{eq:prope3}
improves to 
\begin{equation}
\label{eq:prope4}
  \chi ( P_\infty  - z )^{-1} \chi = \Oo_{ L^2 \to L^2 } (
  h^{-1 + c_1 \Im z/ h }  \log ( 1/h) ) , \ \ z \in [ - \delta,
\delta ] - i h [ 0, C_0 ] . 
\end{equation}

We can now state a more general version of Theorem \ref{t:2}:
\begin{thm}
\label{t:2'}
Suppose that the Riemannian manifold $ ( Y , g ) $ and the potential $
V$ satisfy the assumptions \eqref{eq:conve}, \eqref{eq:Kdell},
\eqref{eq:prope1} and
\eqref{eq:prope2}.  In particular, the trapped set for the 
operator $ P = -h^2 \Delta_g + V $ is normally hyperbolic.

Then, for some constant $ C_1 $ (and $s_0 $ in \eqref{eq:prope1}), we have
 \begin{equation} 
\label{eq:t2'}
\|    \rho^{s_0} (  P_g - E - i 0 )^{-1} \rho^{s_0} \|_{
L^2   ( Y ) \to L^2  ( Y ) }  \leq C_1 \frac {\log (1/h) }  h   , \  \ \ |E | \leq \delta .
\end{equation} 

If in addition \eqref{eq:prope3} holds then, for any $ \epsilon_0 > 0 $,
$ \chi ( P_g - z )^{-1}
\chi $ can be continued analytically to $ [ - \delta/2 , \delta/2] - i h [
0 , \min (C_0, \lambda_0/2 - \epsilon_0 ) ] $, with $ \lambda_0 $ given
by \eqref{eq:t11}, 
and 
\begin{equation}
\label{eq:t2''}
 \chi ( P_g  - z )^{-1} \chi = \Oo_{ L^2 \to L^2 } ( h^{-N
}) , \ \ z \in [ - \delta/2, \delta/2 ] - i h [ 0,  \min ( C_0 ,
\lambda_0/2 - \epsilon_0 )  ] ,
\end{equation}
with the improved estimate \eqref{eq:t1} if \eqref{eq:prope4} holds.
\end{thm}

Before the proof we present two classes of manifolds which 
satisfy our assumptions. 
  We say $(Y,g)$ is \textit{asymptotically Euclidean}  if
\[
g = \rho^{-4}d\rho^2+ \rho^{-2}g_0 (\rho ) , \qquad \textrm{near } \partial  Y,
\]
where $g_0 ( \rho ) $ is a family of metrics on $\partial  Y$
depending smoothly on $\rho $ up to $ \rho =0$. 
We say $(Y,g)$ is  \textit{evenly asymptotically hyperbolic}  if
\[
g = \rho^{-2}d\rho^2 + \rho^{-2}g_0( \rho ), \qquad \textrm{near
} \partial   Y, \]
where $ \rho $ is as before but this time $g_0 ( \rho)  $ is a family
of metrics on $\partial  Y$ depending smoothly on 
$\rho^2$ (hence {\em even}) up to $\rho
=0$. 

In both cases the non-trapping assumption near infinity
\eqref{eq:conve} is valid: see \cite[Proof of Lemma 4.1]{DatVas10_1} for the
asymptotically hyperbolic case; the asymptotically Euclidean case
follows from the same proof, with the fourth displayed equation of the
proof replaced by \cite[(4.3)]{vz}.

For asymptotically Euclidean manifolds \eqref{eq:prope1} and
\eqref{eq:prope2} follow from the results of \cite{vz}. The
modification 
of $ V $ can be done in any way which produces a non-trapping 
classical flow: for instance we can choose $ \widetilde V = V + V_{\rm{int}}
$ where $ V_{\rm{int}} $ is a smooth, large non-negative potential (a barrier)
supported in $ \{ \rho > \epsilon_1 \} $.

To obtain
\eqref{eq:prope3} more care is needed but, under additional
assumptions
one can use an adaptation of the method of complex scaling
of Aguilar-Combes, Balslev-Combes and Simon -- see \cite{WZ} 
for the case of manifolds and for references.
The simplest example for which this is valid was considered in 
Theorem \ref{t:0}. 
For even asymptotically hyperbolic manifolds the properties
\eqref{eq:prope1}, \eqref{eq:prope2}, and \eqref{eq:prope3} all 
follow from the recent work of Vasy \cite{Va}.  

As long we are not interested in analytic 
continuation properties, the weaker assumptions  \eqref{eq:prope1} and \eqref{eq:prope2} may 
hold in the generality considered by Cardoso-Vodev \cite{CaVo}.

\begin{proof}[Proof of Theorem \ref{t:2'}]
To show how Theorem \ref{t:2'} follows from Theorem \ref{t:1}  we use
the parametrix construction of \cite[\S 3]{DatVas10_1}. For that we
first have to relate the situation in this section to the set-up in
Theorem \ref{t:1}.  It will be convenient to rescale $ \rho $ so that
in \eqref{eq:conve} we can take $ \epsilon_1 = 4 $. 

Let $X$ be any compact manifold without boundary such that
$\overline Y \subset X$ 
is a smooth embedding: for example, we may take $X$ to be the double
of $\overline Y$.  We then extend $ \rho $ to  $ \rho \in L^\infty ( X
) $ to be identically $0$ on $X \setminus Y$.  
Let $P \in \Psi^2 ( X ) $ be any selfadjoint semiclassical differential
operator satisfying
\[ P|_{\rho > 1}  = P_g|_{\rho > 1 } ,  \ \  P = -h^2 \Delta_{g_X} +
V_X , 
\]
where $ g_X $ is a Riemannian metric on $ X $ and $ V_X \in \CI ( X ;
\RR) $.

We then take $W \in \CI (X;[0,\infty))$ such that 
\[   W ( x ) = \left\{ \begin{array}{cc}   0 & \text{ for $ \rho ( x )
      > 1 $;} \\ 
1  & \text{ for $ \rho ( x ) < \frac12 $.} \end{array} \right.
\]
Let $ \widetilde V \in \CI ( Y) $ be a potential for which
\eqref{eq:prope1} and \eqref{eq:prope2} hold. We notice that one possibility to obtain the required properties for $P_\infty$
is to take a complex potential $ \widetilde V = V - i W_\infty $ where, $W_\infty \in \CI
(Y;[0,\infty))$ 
$$
W_\infty ( x ) =  \begin{cases}   0 & \text{ for $ \rho ( x )  < 4  $,} \\  
1 & \text{ for $ \rho ( x ) > 5   $,} \end{cases},\quad \text{see Fig.~\ref{f:1}}.
$$
Using the convexity property \eqref{eq:conve} it is easy to check that this operator satisfies \eqref{eq:prope1} and \eqref{eq:prope2}.
Then for $\Im z>0$, $ | \Re z | \leq \delta $,  define the following holomorphic families of operators
\[
R_X(z) = (P-iW-z)^{-1}, \qquad R_\infty(z) = (P_\infty - z)^{-1}.
\]
Due to the compactness of $X$, the family of operators $R_X(z) \colon L^2(X) \to L^2(X)$ is
meromorphic for $ z  \in \CC $. The resolvent $R_X(z)$ is estimated in Theorem~\ref{t:1}.
For the moment we only assume that $R_\infty(z) \colon L^2(Y) \to L^2(Y)$ is
holomorphic for $\Im z>0$ and satisfies \eqref{eq:prope1}, \eqref{eq:prope2}. 

Now take a cutoff function $\chi_X \in \CI ( \RR,[0,1])$ with
\[  \supp \chi_X \subset ( 2 , \infty ) , \ \ \ \supp (1-\chi_X)
\subset (-\infty , 3 ). \]  We put and 
$ \chi_\infty = 1 - \chi_X $.

\begin{figure}[ht]
\includegraphics[width=6in]{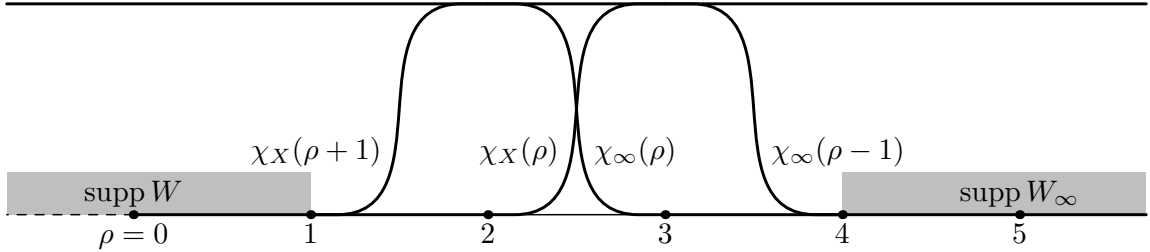}
\caption{Schematic representation of the cut-offs used in the proof of
  Theorem \ref{t:2'} as functions of $ \rho ( x ) $. The spatial
  infinity 
is represented by $ \rho ( x ) = 0$ and $ X \setminus Y $ corresponds
to $ \rho ( x ) \leq 0 $.}
\label{f:1}
\end{figure}

Our first Ansatz for the inverse of $(P_g-z)$ is the operator
\[
F(z) = \chi_X(\rho(\bullet)+1) R_X ( z ) \chi_X(\rho(\bullet)) +
\chi_\infty(\rho(\bullet)-1) R_\infty ( z) \chi_\infty(\rho(\bullet)).
\] 
Note that  $F (z) \colon L^2(Y) \to L^2(Y)$ for $\Im z >0$ since all
the cut-off functions are supported away from $ X \setminus Y $. 
Also, the support properties of $W $, $ W_\infty $ and $ \chi_X $ show
that 
\[ \begin{split} &  ( P_g - z ) \chi_X ( \rho( \bullet ) +  1 ) = \chi_X ( \rho ( \bullet ) + 1 )
( P - i  W - z ) +  [\chi_X (\rho(\bullet) + 1),h^2\Delta_g ] , \\
& ( P_g - z ) \chi_\infty  ( \rho( \bullet ) -  1 ) = \chi_\infty ( \rho ( \bullet ) - 1 )
( P_\infty  - z ) +  [\chi_\infty(\rho(\bullet)-1),h^2\Delta_g ] .
\end{split} \]
Hence
\[
(P_g - z) F( z)  = \Id + A_X( z )  + A_\infty (z ) , 
\]
where
\begin{gather*} 
A_X ( z ) = [\chi_X(\rho(\bullet)+1),h^2\Delta_g ] R_X ( z ) \chi_X(\rho(\bullet)),
\\ A_\infty ( z )  = [\chi_\infty(\rho(\bullet)-1),h^2\Delta_g ] R_\infty
( z ) \chi_\infty(\rho(\bullet)). 
\end{gather*} 
Note that $A_X ( z ) ^2 = A_\infty ( z ) ^2 = 0$, due to the support properties
\begin{equation}
\begin{split}\label{eq:support} 
\supp d \left( \chi_X(\rho(\bullet)+1) \right) \cap \supp \chi_X(\rho(\bullet)) &= \emptyset , \\
\supp d \left( \chi_\infty(\rho(\bullet)-1) \right) \cap \supp \chi_\infty(\rho(\bullet)) &=\emptyset . 
\end{split}
\end{equation}
 Moreover, thanks to assumptions \eqref{eq:conve} and
 \eqref{eq:prope2} (see \cite[Lemma 3.1]{DatVas10_1}), 
\begin{equation}
\label{eq:unifz}
\|A_\infty( z )  A_X ( z ) \|_{L^2(Y) \to L^2(Y)} = \Oo(h^\infty),
\text{ uniformly for $ \Im z > 0 $, $ | \Re z | \leq \delta$. }
\end{equation}
Consequently
\begin{gather}
(P_g - z) F ( z ) \big((\Id - A_X ( z ) - A_\infty ( z )  + A_X ( z )
A_\infty ( z )\big)  = \Id - E ( z )\,, \\
\text{where}\qquad
E ( z ) = A_\infty( z )  A_X( z )  - A_\infty ( z ) A_X ( z ) A_\infty
( z )\,.
\end{gather}
Using \eqref{eq:unifz} we see that $E(z)= \cO(h^\infty)_{L^2(Y)\to L^2(Y)}$, uniformly for $\Im z>0$, $  | \Re z | \leq \delta $.
This allows to write an explicit expression for $(P_g - z)^{-1}$:
\[
(P_g - z)^{-1} =  F ( z ) (\Id - A_X ( z ) - A_\infty ( z )  +
A_X ( z ) A_\infty ( z ) ) \sum_{n=0}^\infty E( z ) ^n\,.
\]
We now want to estimate $\| \rho^{s_0} ( P_g - z )^{-1} \rho^{s_0} \|_{L^2(Y)\to L^2(Y)} $. For this we expand the above identity using the 
expression of $F(z)$ (some terms vanish due to the support properties \eqref{eq:support}). 
Denoting $a_X=\| R_X ( z ) \| $, $a_\infty = \| \rho^{s_0} R_\infty ( z ) \rho^{s_0} \|$, we get the bound
\be\label{e:bound-a}
\| \rho^{s_0} ( P_g - z )^{-1} \rho^{s_0} \| \leq C\big( a_\infty + a_X + h a_\infty a_X + h^2 a_\infty^2 a_X \big) + \cO(h^\infty)\,.
\ee
Finally, we use the bounds  \eqref{eq:prope1} for $a_\infty$, the bound \eqref{eq:t1} for $a_X$ (with $\Im z\geq 0$), 
and obtain the desired estimate \eqref{eq:t2'}.

When the assumption \eqref{eq:prope3} holds,  the construction shows that for $ \chi \in
\CIc ( Y ) $ equal to $ 1 $ on a sufficiently large set,
\[
\chi (P_g - z)^{-1} \chi =  \chi F ( z ) \chi  \big(\Id - A_X ( z ) - A_\infty ( z )  +
A_XA_\infty ( z ) \big) \chi \sum_{n=0}^\infty ( E( z ) \chi) ^n, 
\]
continues analytically to the same region as {\em both} $ R_X ( z ) $ and 
$ \chi R_\infty ( z ) \chi $. The same expansion as above allows to bound from above $\|\chi (P_g - z)^{-1}\chi \|$
by the same expression as in \eqref{e:bound-a}, now using $a_X=\| \chi R_X ( z ) \chi \| $, $a_\infty = \| \chi R_\infty ( z ) \chi \|$. By using \eqref{eq:t1} for $a_X$, resp. \eqref{eq:prope3} for $a_\infty$  (with now $\Im z$ taking negative values), we obtain \eqref{eq:t2''}.
\end{proof}

For completeness we conclude this section with the proof of
Theorem \ref{t:0}. The conclusion is valid under more general
assumptions
of Theorem \ref{t:2'}.
\begin{proof}[Proof of Theorem \ref{t:0}]
In the notation of Theorem \ref{t:2'}, \eqref{eq:t0}  is equivalent to the estimate 
\begin{equation}
\label{eq:t0'} \|  \chi \psi ( P_g  ) e^{ - i t P_g / h } \chi \|_{ L^2 ( Y) \to L^2 (
Y )} \leq C \frac{\log 1/h}{h^{1+c_0\gamma}}  \, e^{ - \gamma t} + {\mathcal O} ( h^\infty ) , 
\qquad \gamma=\frac12(\lambda_0-\eps)\,,
\end{equation}
valid (with different constants) for any $ \chi \in \CIc ( Y ) $.
Let $ \tilde \psi \in \CI_{\rm{c}} ( \CC )  $ be an almost analytic
extension of $ \psi $, that is a function with the property that $
\tilde \psi \rest_{ \RR} = \psi $ and $ \bar \partial_z \tilde \psi (
z ) = {\mathcal O} ( | \Im z |^\infty ) $ (see for instance
\cite[Theorem 3.6]{e-z}). We can construct $ \tilde \psi $ so that 
$ \supp \tilde \psi \subset [- \delta/2, \delta/2 ] - i [ -
\delta, \delta ] ) $.
We start with Stone's formula
\[ \chi \psi ( P_g )   e^{ - i t P_g / h } \chi = 
\frac 1 { 2 \pi i  }
 \int_\RR \psi ( \lambda ) e^{ - i \lambda t } \chi  \left(
( P_g - \lambda - i0 )^{-1}  - 
( P_g - \lambda + i0 )^{-1}  \right) \chi \,d  \lambda . \]
We now write $ R_- ( z ) = ( P_g - z )^{-1} $, for the resolvent in 
$ \Im z  < 0 $ (that is for the analytic continuation of
$ ( P_g - ( z - i0 ) )^{-1} $ from
$ \Im z < 0 $) and $ R_+ ( z ) $ for the meromorphic continuation 
of the resolvent from $ \Im z > 0 $ to the lower half-plane.
We then apply Green's formula to obtain, for $ 0 \leq  \gamma < \lambda_0 / 2 $, 
\begin{equation}
\label{eq:Green} \begin{split}
\chi  \psi ( P_g)  e^{ - i t P_g  / h }  \chi
& = \frac{1}{2\pi i } \int_{ \Im z =  - \gamma h  }  e^{ - i t  z /h } \chi 
( R_+ ( z ) - R_ - ( z ) ) \chi \tilde
 \psi ( z )  d z 
\\ 
& \ \ \ \ \ \  + 
\frac{1}\pi \iint_{  - \gamma  h  \leq \Im  z \leq 0 }
 e^{ - i t  z /h  } \chi ( R_+ ( z ) - R_ - ( z ) )  \chi \bar
\partial_z \tilde \psi ( z )  dm ( z)  \,,
\end{split}
\end{equation}
where $ dm ( z ) $ is the  Lebesgue measure on $ \CC $.  
From \eqref{eq:t1} (see Theorem \ref{t:2'}) we get
\[ \|  \chi R_+ ( z ) \chi \|_{ L^2 \to L^2}  \leq  C h^{-(1  + 
c_0 \gamma) } \log ( 1/h)  , \qquad \| \chi R_- ( z ) \chi \|_{L^2
\to L^2 } \leq C / |\Im z | , \]
for $   - \gamma h \leq \Im z \leq 0 $. Inserting these bounds
in \eqref{eq:Green} gives \eqref{eq:t0'} and that 
proves (a generalized version of) Theorem \ref{t:0}.
\end{proof}

\section{Decay of correlations for contact Anosov flows: Proof of
  Theorem \ref{t:3}}
\label{decay}
Most of this section is devoted to the proof of Theorem \ref{t:3}. This proof will be obtained by adapting the proof of Theorem~\ref{t:1}, after reviewing the geometric point of view of Tsujii \cite{Ts} and 
Faure--Sj\"ostrand \cite{f-sj} (see also \cite{dadyz}). 
At the end of the section we deduce Corollary~\ref{t:4} on the decay of correlations.

\subsection{Geometric structure}
\label{gest}
Let $X $ be a smooth compact manifold of dimension $d=2k-1$, $ k \geq
2 $. We assume that $ X $ is equipped with a contact 1-form
$\alpha$, that is,  a form such that $ (d \alpha)^{\wedge  (k - 1) } \wedge
\alpha $ is non-degenerate. The Reeb vector field, $ \Xi$,  is 
defined as the unique vector field satisfying
\[   \Xi_x \in \ker d \alpha_x , \quad \alpha_x ( \Xi_x ) =1,\quad x \in X\,.
 \]
We assume that
\begin{equation}
\label{eq:Anos}
\gamma_t \defeq \exp t \Xi \ \text{ defines an Anosov flow on $ X $.} 
\end{equation}
That means that at each point $x\in X$, the tangent space has a
$ \gamma_t$-invariant decomposition
into neutral (one dimensional), 
stable and unstable subspaces (each $(k-1)$-dimensional): 
\begin{equation}
\label{eq:Ande}  T_x X = E_0 (x) \oplus E_s( x ) \oplus E_u ( x ) ,  \ \  E_0  ( x
) = \RR 
\Xi_x . \end{equation}
We note that 
$ E_u ( x ) \oplus E_s ( x ) $ span the kernel of $\alpha_x$. 

The dual decomposition is
obtained by taking $ E_0^* ( x ) $ to be the annihilator of
$ E_s ( x) \oplus E_u ( x ) $, $ E_u^* ( x) $ the annihillator
of $ E_u ( x ) \oplus E_0 ( x ) $, and similarly for $ E_s^*( x) $.
That makes $E_s^*( x ) $ dual to $E_u ( x ) $, $ E^*_u (x) $ dual to 
$ E_s ( x ) $, and $ E_0^* ( x ) $ dual to $ E_0 ( x ) $. The fiber
of the cotangent bundle then decomposes as 
\begin{equation}\label{e:cotdec}
T_x^*X = E_0^*(x)\oplus E_s^*(x)\oplus E_u^*(x) .
\end{equation}
The distributions $E_s^* (x) $ and $E_u^* ( x ) $ have only H\"older regularity, but
$E_0^* ( x ) $ and $E_s^*(x)\oplus E_u^* ( x ) $ are smooth, and 
$E_0^* ( x ) =\mathbb R \alpha_x \subset T_x^* X $.

The approach of \cite{f-sj} highlights the analogy between this
dynamical setting and the scattering theory for the Schr\"odinger equation.
The role of the Schr\"odinger operator is played by the (rescaled) generator of the flow
$ \gamma_t = \exp t \Xi $:
\begin{equation}
\label{eq:gaP}  \gamma_t^* u = e^{  i t P / h } u ,  \ \ \ u \in \CI ( X ) , \ \ \ P = - i h \Xi \,.
\end{equation}
The principal symbol of $P$ simply reads $p(x,\xi)= \xi ( \Xi_x )$.

The flow $\gamma_t$ can be lifted to a Hamiltonian flow $\varphi_t$ on $T^*X$:
$$
\varphi_t: (x,\xi) \longmapsto ( \gamma_t(x) , ^t\! d\gamma_t(x)^{-1} \xi ) ,
$$
which is generated by $p(x,\xi)$: 
$\varphi_t = \exp t H_p$. 

For each energy $E\in\IR$, the energy shell $p^{-1}(E)$ is a union of affine subspaces:
$$
p^{-1}(E) = \bigcup_{ x \in X } \{(x,\xi) : \alpha_x(\xi)=E\} = \bigcup_{ x \in X } (E\alpha_x + E^*_u(x)+E^*_s(x))\,.
$$
We note that each of these energy shells has infinite volume; as opposed to the scattering theory setting, infinity occurs here in the momentum direction (the fibers), while the spatial direction is compact.

The Anosov assumption implies that for $ t > 0 $, 
\begin{equation}
\label{e:hyperbolicity}
\begin{gathered}
| \varphi_t ( x , \xi )|\leq Ce^{-\lambda t}| \xi |,\quad \xi \in E_s^*
( x ) ,\ \ \
| \varphi_{-t}   (x , \xi )|\leq Ce^{-\lambda t}|\xi |,\quad \xi \in E_u^*( x
) ,
\end{gathered}
\end{equation}
where $ | \bullet | = | \bullet |_y  $ denotes a norm on $ T_y^* X $,
and we consider $ \varphi_t ( x, \xi ) \in T_{ \pi ( \varphi_t ( x,
  \xi))}^* X $. 
Since the action of $\varphi_t$ inside each fiber $T^*_xX$ is linear, we see that
the only trapped points in $T^*X$ must be on the line $E^*_0(x)$. More precisely, 
the trapped set at energy $ E\in\IR $ is given by 
\[      K_E = \bigcup_{ x \in X } \left( E_0^* ( x ) \cap p^{-1} ( E )
\right) = \bigcup_{ x \in X } E \alpha_x\, , 
 \]
that is $ K_E $ is the image of the section $ E \alpha $ in $T^*X$.

Stacking together energies $E \in (1-\delta,1+\delta ) $, $ 0 < \delta
< 1 $, we obtain the trapped set
$$
K^\delta = K = \bigcup_{| E - 1| <  \delta } K_E = \{ E\alpha_x,\,
x\in X,\, | E -1 | <  \delta  \}\subset T^*X\,.
$$
This trapped set is normally hyperbolic in the
sense of \eqref{eq:NH}. 

Indeed, we first note that $ K^\delta $ is  a symplectic
submanifold of $T^*X$ 
of dimension $ d+1 = 2k$.  Indeed, using $ ( x , E )$, $
x \in X $, as coordinates on $ K^\delta $, $ ( x , E ) \mapsto E
\alpha_x $, we have 
$$
\omega\rest_{K^\delta} = d ( E \alpha ) =  dE \wedge \alpha + E \, d
\alpha\, . 
$$
This form is nondegenerate for $ E \neq 0 $ since  
$\alpha$ is a contact form.

The tangent space to $ K^\delta $ is given by the image of
the differential of 
\[  X \times \RR \ni ( x , E ) \mapsto E \alpha_x   \defeq ( x , \xi=E
\beta ( x ) ) \,, \]
where we see $\beta(x)$ as the vector in $\IR^{d}$ representing $\alpha_x$.
Hence,
\begin{equation}
\label{eq:tanK}    \begin{split} 
T_{ E \alpha_x   } K^\delta  & =  E (d \alpha)_x  (T_x X,\bullet)
+ \RR \alpha_x
= \{ ( v , E\,d \beta ( x ) v + s \beta ( x
) ) : v \in T_x X , s \in \RR \} \\
& \subset T_x X \oplus T_x^*X  \equiv
T_{E\alpha_x } ( T^* X ) . 
\end{split} 
\end{equation}
Here $d\beta(x)$ can be interpreted as the Jacobian matrix ${\partial \beta}/{\partial x}$ on $\IR^{d}$.

For each $x\in X$, the symplectic orthogonal to $T_{ E \alpha_x   } K^\delta$, denoted $(T_{ E \alpha_x   } K^\delta)^\perp$, can be obtained by lifting the vectors in $\ker \alpha_x$ as follows:
$$
v\in \ker\alpha_x \mapsto L_E^\perp(v)\defeq (v,E\, ^t\!(d\beta(x))v)\in T_x X\oplus T^*_x X \equiv T_{E\alpha_x } ( T^* X )\,,
$$
where $^t\!(d\beta(x))$ denotes the transpose of $d\beta(x)$. This subspace $(T_{ E \alpha_x   } K^\delta)^\perp$ is symplectic and transverse to $K^\delta$:
$$
T_{ \rho  } K^\delta \oplus (T_{ \rho  } K^\delta)^\perp = T_\rho(T^*X),\quad \forall \rho=E\alpha_x \in K^\delta.
$$ 
Since $\ker\alpha_x= E_u(x)\oplus E_s(x)$, we can naturally split the orthogonal subspace into 
$$
(T_{ \rho  } K^\delta)^\perp = E^+_\rho\oplus E^-_\rho,\quad E^+_\rho = L_E^\perp(E_u(x)),\quad E^-_\rho = L_E^\perp(E_s(x)),\quad \rho=E\alpha_x\in K^\delta\,.
$$
The distributions $E^\pm_{E\alpha_x}$ are transverse to each other and flow-invariant. $E^+_{E\alpha_x}$ is a particular subspace of the global unstable subspace $E_u(x)\oplus E_u^*(x)\subset T_{E\alpha_x}(T^*X)$, and similarly for $E^-_{E\alpha_x}$. Hence, in the present setting, the subspaces $E^\pm_{\rho}$ exactly correspond to the subspaces described in Lemma~\ref{l:basicd}.

\subsection{Microlocally weighted spaces and the definition of
  resonances}
\label{micwes}

Following \cite{dadyz} we now review the construction
\cite{f-sj} of Hilbert spaces on which $ P - z $ (with $ P
$ given in \eqref{eq:gaP}) is a Fredholm operator for $ \Im z > -
\beta h $, for some arbitrary $ \beta >0$.

The key to the definition of these Hilbert spaces is a construction of
a weight function which we quote from \cite[Lemma 1.2]{f-sj} and
\cite[Lemma 3.1]{dadyz}. We use the notation $ E_\bullet^* =
\bigcup_{x \in X } E_\bullet^* ( x ) \subset T^* X $.
%
%
\begin{lem}
  \label{l:escape-f-sj}
Let $U_0,U_0'$  be conic neighbourhoods of $E_0^*$, with $U_0\Subset U'_0$
and $U'_0\cap (E_u^*\cup E_s^*)=\emptyset$.
There exist real-valued functions $m\in S^0(T^*X),f_0\in S^1(T^*X)$
such that
\begin{enumerate}
\item $m$ is
positively homogeneous of degree 0 for $|\xi|\geq 1/2$,
equal to $-1,0,1$ near the intersection of $\{|\xi|\geq 1/2\}$ with
$E_u^*,E_0^*,E_s^*$, respectively, and
\begin{equation}
  \label{e:hpm}
H_p m<0\text{ near }(U'_0\setminus U_0)\cap \{|\xi|>1/2\},\quad
H_p m\leq 0\text{ on }\{|\xi|>1/2\};
\end{equation}
\item $\langle\xi\rangle^{-1}f_0\geq c>0$ for some constant $c$;
\item the function $\cG\defeq m\log f_0$
satisfies
\begin{equation}
  \label{e:hpG}
H_p\cG\leq - 2 \text{ on }\{|\xi|\geq 1/2\}\setminus U_0,\quad
H_p\cG\leq 0\text{ on }\{|\xi|\geq 1/2\}.
\end{equation}
\end{enumerate}
\end{lem}

The function $\cG$ also satisfies derivative bounds
\begin{equation}
\label{eq:strG}
\cG  = {\mathcal O} ( \log \langle \xi \rangle ) , \ \ 
 \partial^\alpha_x\partial^\beta_\xi H^k_p \cG
= {\mathcal O} \left( \langle\xi\rangle^{-|\beta| + \epsilon} \right) 
, \quad| \alpha | + | \beta| + k \geq 1 \,, 
\end{equation}
for any $ \epsilon > 0 $. 

As in \cite[\S 3]{dadyz} we define
\begin{equation} 
\label{eq:defHG}
H_{t\cG} ( X ) \defeq e^{- t \cG^w}L^2 ( X , dx ) , 
\end{equation}
where $ t> 0 $ is a positive parameter.

The domain of $ P $ acting on $ H_{t\cG} $ is defined as
\begin{equation}
\label{eq:defDG}
{\mathcal D}_{t\cG} \defeq \{ u \in {\mathcal D}' ( X )  \; : \; u, P u \in H_{t\cG } \}.
\end{equation}
The action of $ P $ on $ H_{t \cG}$ is equivalent to the action of the
operator $P_{t\cG} $ on $ L^2 $:
\begin{equation}
\label{eq:PtG}
\begin{split}  P_{t\cG} & \defeq e^{ t \cG^w } P e^{ - t \cG^w } 
= \exp (t  \ad_{ \cG^w } ) P
\\ & = \sum_{k=0}^N \frac{ t^k } { k!} \ad_{\cG^w }^k P + R_N ( x, h D)
, \ \ \ 
R_N \in h^{N+1} S^{ - N  + \epsilon } , \ \ \forall \, \epsilon > 0 .
\end{split}
\end{equation}
The validity of \eqref{eq:PtG} follows from the fact 
that the operators $ e^{ \pm t \cG^w } $ are pseudodifferential
operators \cite[Theorem 8.6]{e-z}, hence the pseudodifferential
calculus applies directly \cite[Theorem 9.5, Theorem 14.1]{e-z}.
This expansion and the arguments in \cite[\S 3]{f-sj} 
give
\begin{prop}
\label{p:1}
For $ P_{t\cG} $ defined by \eqref{eq:PtG}, we have

\noindent
i) the operator $  P_{t\cG} - z : \mathcal D(P_{t\cG})\to L^2 $ is Fredholm
of index zero for $ \Im z > - th $. Here $\mathcal D(P_{t\cG})$
is the domain of $P_{t\cG}$.

\noindent
ii)
for $C>0$ large enough, $(P_{t\cG}-z)$ is invertible on $\{\Im z>Ch\}$.
\end{prop}

In~\cite{f-sj} the above construction was performed, replacing the $h$-quantization by the $h=1$ quantization. It lead
to the construction of  $H_{t\cG,1}(X)=e^{- t \cG^w(x,D)}L^2 ( X )$ equal, as vector space, to the above $h$-dependent space $ H_{t\cG} ( X )$.
The norms of these two spaces are equivalent with one another, but in an $h$-dependent way:
\begin{equation}
\label{eq:normtG}
h^N \| u \|_{ H_{t\cG,1} ( X ) } /C_0  \leq \| u \|_{ H_{t\cG}(X) } \leq C_0 h^{-N}\,\| u \|_{ H_{t\cG,1} ( X ) }\,,
\end{equation}
 --- see \cite[\S 3]{dadyz}.
As a consequence, if we call $P_1=-i\Xi$, Theorem~\ref{t:3} translates into the fact that $P_{t\cG,1}\defeq e^{ t \cG^w(x,D) } P_1 e^{ - t \cG^w(x,D) }$ is Fredholm in the strip $\{\Im\lambda>-t\}$, admits finitely many eigenvalues in the strip $\{\Im\lambda\geq -\lambda_0/2+\eps_0\}$, and satisfies the resolvent estimate 
\be\label{eq:t3-bis}
\| (P_{t\cG,1}-\lambda)^{-1}\|_{L^2\to L^2} =  \cO(\lambda^{N_0}),\quad \Im\lambda\geq -\lambda_0/2+\eps_0,\ \ |\Re\lambda|\geq C\,.
\ee

\subsection{Reduction to Theorem \ref{t:1}}

In order to prove Theorem~\ref{t:3}, we proceed as in the proof of Theorem \ref{t:2'} in \S\ref{recap}, by constructing
two operators which microlocally agree with $ P_{t \cG} $ (up to negligible 
error terms) on different subsets of $ T^*X $. 

Let $ W \in \Psi^1 ( X ) $ be as in Example 2 of \S \ref{aatr}. The 
trapped set defined in \S \ref{aatr} agrees with the trapped
set in \S \ref{gest} and, as shown in \eqref{e:hyperbolicity},  it satisfies, the assumptions
of normal hyperbolicity. Hence Theorem \ref{t:1} applies to 
$ \widetilde P = P - i W $. If $ A \in \Psi^\comp ( X ) $ satisfies
\begin{equation}
\label{WFA1}   \WFh ( A ) \Subset (\cG^{ -1} ( 0 ))^\circ ,  \ \ \ 
\WFh ( A ) \cap \{ ( x, \xi ) : {| \xi|_{g_x} \geq M }\}  = \emptyset,
\end{equation}
(where $ M$ is the one appearing in the definition of $ W $ -- see
\eqref{eq:fCAP})
then 
\begin{equation}   \label{eq:AP} A \widetilde P = A P_{t\cG} + {\mathcal O} ( h^\infty ) _{
    {\mathcal D'} \to \CI } .
\end{equation}

We now introduce an operator which has better global properties 
and agrees with $ P_{t\cG} $ near infinity. For that we proceed as in the proof of Theorem~\ref{t:2'}, and 
take $ W_\infty \in
\Psi^\comp ( X ) $ 
such that 
\begin{equation}
\label{WFW1}   \WFh ( W_\infty ) \Subset (\cG^{ -1} ( 0 ))^\circ ,  \ \ \ 
\ \ \WFh ( I - W_\infty  ) \subset \complement  K^\delta ,  \ \ \ 
\WFh ( W ) \cap \WFh ( W_\infty  ) = \emptyset.
\end{equation}
We then put
\[  P_\infty = P_{t\cG}  - i W_\infty .\]
Then for any $ B $ with $ \WFh ( I - B ) \subset
\complement \WFh ( W_\infty ) $, 
\begin{equation}
\label{eq:BP}    ( I - B ) P_\infty = ( I - B ) P_{t\cG} + {\mathcal O}( h^\infty ) _{{\mathcal
    D}' \to \CI } . 
\end{equation}

Properties of the operator $ P_\infty $ are listed in the following 
\begin{lem}
\label{l:lemA}
Fix $\beta>0$ and let $ t $ be large enough so that 
$P_{t\cG}- z $ is a Fredholm opeartor for $\Im z>-\beta h$.  
Then, there exists $ N_0 $ and $ h_0 $ such that, for $ 0 <h < h_0 $, 
$$ 
\| ( P_\infty - z )^{-1} \|_{L^2 \to L^2 } \leq  h^{-N_0}   \,,\quad z \in [ 1 - \delta/2 , 1 + \delta/2 ] - i h[ 0 , \beta ] \,.
$$
In addition the analogue of \eqref{eq:prope2} holds for $ P_\infty $: 
in the same range of $ z $, 
\begin{equation}
\label{eq:prope21}
\begin{split} & u = (P_\infty  - z)^{-1} f , \ \ f \in \CI ( X )  \
  \ \Longrightarrow  \\
& \ \ \ \ \ \ \ \  \WFh ( u ) \setminus \WFh ( f ) \subset \exp ( [ 0 , \infty ) H_p ) \left(
  \WFh ( f ) \cap p^{-1} ( \Re z  ) \right)  .
\end{split}
\end{equation}
\end{lem}
\begin{proof}
The first part follows from the now standard non-trapping
estimates (see \cite[\S 4]{SZ10}). In the setting of 
Anosov flows the details are presented in the proof
of \cite[Lemma 5.1]{dadyz} (only the escape function 
constructed in Lemma \ref{l:cef3} above is needed). 

The propagation 
result is a real principal type propagation result 
\cite[Theorem 12.5]{e-z}  which 
holds when the imaginary part of the symbol is 
non-positive -- see Lemma \ref{l:propa} below for
a dynamical version.
\end{proof}

\begin{figure}[ht]
\includegraphics[width=6in]{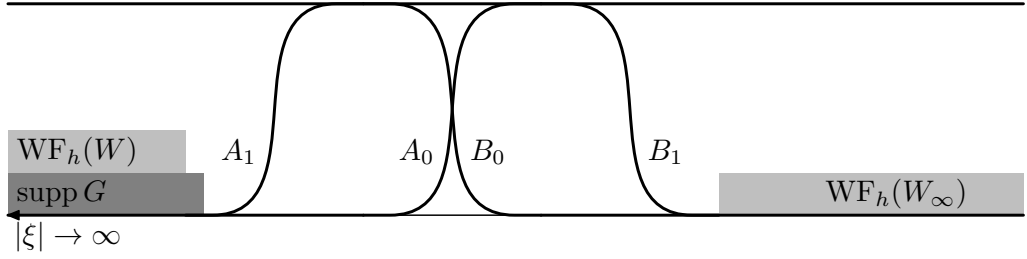}
\caption{Schematic representation of pseudodifferential 
cut-offs used in the proof of
  Theorem \ref{t:3}. The horizontal axis corresponds to $|\xi| $, the
  cotangent variable. Infinity in $ |\xi| $ plays the role of $ \rho
  = 0 $ in Fig. \ref{f:1}. The asymmetry is intentional, to 
stress that there is no need for an auxiliary manifold, as opposed to the proof of
Theorem \ref{t:2'} illustrated in Fig. \ref{f:1}.}
\label{f:2}
\end{figure}

\begin{proof}[Proof of Theorem \ref{t:3}]
The proof is a repetition of the proof of Theorem \ref{t:2}
with $ R_X $ replaced by $ ( \widetilde P -  z )^{-1} $  and $ 
R_\infty $ by $ ( P_\infty -  z)^{-1} $.  The spatial cut-off
functions are replaced by pseudifferential operators:
$  \chi_X ( \rho ( x ) )  $ is replaced by $   A_0 \in \Psi^\comp ( X ) $, satisfying 
\[  \WFh ( A_0 )  \cap \{ ( x, \xi) : { | \xi|_{g_x} \geq M } \} = \emptyset ,  \ \  \WFh ( I - A_0 ) \cap
\WFh ( W_\infty) = \emptyset , \]
where $ M$ is given in the definition of $ W $, see \eqref{eq:fCAP}.
The function 
$ \chi_X ( \rho ( x ) + 1  ) $ is replaced by $  A_1 \in \Psi^\comp
( X ) $,  
where
\[ \WFh ( I -  A_1 ) \cap \WFh ( A_0 ) = \emptyset , \ \ 
\WFh ( A_1  ) \cap \{ ( x, \xi) : {| \xi|_{g_x} \geq M } \}= \emptyset
,  \]
$  \chi_\infty ( \rho ( x ) ) $ is replaced by $  B_0 \defeq I - A_0 \in
\Psi^0 ( X) $, and finally
$ \chi_\infty ( \rho ( x ) -1 )$ by 
$ B_1 \in \Psi^0 ( X)$ , where 
\[  \WFh ( W_\infty  ) \cap \WFh ( B_1 ) =
\emptyset ,  \ \  \WFh ( I - B_1 ) \cap \WFh ( B_0 ) = \emptyset   .  \]
We also require that
\[  \WFh ( A_1) , \WFh ( I - B_1 )  \subset ( G^{-1} ( 0 ) )^\circ . \]

The parametrix is now obtained by putting 
\[  F ( z ) = A_1 ( P - i W - z)^{-1} A_0 + B_1 ( P_{t\cG} - i W_\infty -
z)^{-1} B_0 . \]
Using \eqref{eq:AP}, \eqref{eq:BP} and Lemma \ref{l:lemA} 
we obtain the theorem by proceeding as in the proof of 
Theorem \ref{t:2} in \S \ref{recap}.
\end{proof}

\begin{proof}[Proof of Corollary~\ref{t:4}] 
We will use the nonsemiclassical operator 
$ P_1= -i\Xi $. It is selfadjoint on $L^2(X)$ -- see \cite[Appendix A]{f-sj} -- 
hence, by Stone's formula, we get for any $f,g\in \CI(X)$
\[ \begin{split}  \int_X \gamma_{-t}^* f \, g \, dx & = \la e^{-itP_1} f,\bar{g}\ra \\
&= \frac 1 { 2 \pi i }  \int_\RR e^{ - i \lambda t } 
 \left( \la ( P_1 - \lambda - i0 )^{-1} f , \bar g \ra - 
\la ( P_1 - \lambda + i0 )^{-1} f , \bar g \ra  \right)  d \lambda \\
& = \frac 1 { 2 \pi i  }\sum_{\pm} \mp  \int_\RR e^{ - i \lambda t}  
( \lambda + i ) ^{-k} 
\la  ( P_1 - \lambda \pm i0 )^{-1} ( P_1 +i )^k f , \bar g \ra  d \lambda\, .
\end{split} \]
Here the brackets $\la\bullet,\bullet\ra$ represent $L^2(X)$ scalar products.

For $ t > 0 $ we can deform the contour in the integral corresponding
to $ + i 0 $ ($\lambda $ approaching the real axis from below), where $\|(P_1-\lambda)^{-1}\|\leq |\Im \lambda|^{-1}$,  so that for $k>1$ the integral is bounded as 
\[ - \frac 1 { 2 \pi i  }
 \int_{\RR - i A } e^{ - i \lambda t }  
( \lambda + i ) ^{-k} 
\la ( P_1 - \lambda  )^{-1} ( P_1 +i )^k f , \bar g \ra d \lambda  = {\mathcal O} ( e^{ - t A } \| f \|_{ H^k } \| g \|_{
  L^2 } ) \,. \] 
Thus, 
\[ \int_X \gamma_t^* f\, g\, dx  = \frac 1 { 2 \pi i  }
\int_\RR e^{ - i \lambda t }  
 ( \lambda + i)^{-k} \la ( P_1 - \lambda - i0 )^{-1}  ( P_1 + i )^k f , \bar g \ra 
+ {\mathcal O}_{f,g}  ( e^{- t A }  )  , \]
for any $ A $, with the bounds depending on seminorms of
$ f$ and $ g $ in $ \CI $.   
We now use the nonsemiclassical weights $\cG^w(x,D)$ constructed in \S\ref{micwes} to 
conjugate $P_1$, and write
\[  \begin{split} & \int_X \gamma_t^* f \, g \, dx  = \\
& \ \ \ \ \ \ \frac 1 { 2 \pi i  }
\int_\RR e^{ - i \lambda t }  
  (\lambda + i)^{-k} \la ( P_{t\cG,1}  - \lambda - i0 )^{-1}  ( P_{t\cG,1} + i )^k e^{t\cG^w(x,D)}
f ,e^{-t\cG^w(x,D)}  \bar  g \ra
+ {\mathcal O}_{f,g}  ( e^{- t A }  )  . \end{split} 
\]
The nonsemiclassical analogue \eqref{eq:t3-bis} of Theorem~\ref{t:3} shows that, by taking $ k > N_0+1 $, we may deform the contour of
integration down to $\Im \lambda = -\lambda_0/2+\eps$, collecting finitely many poles $\mu_j$, to finally obtain the expansion \eqref{eq:t4}.
\end{proof}

\vspace{0.5cm}
\begin{center}
\noindent
{\sc  Appendix: Evolution for the CAP-modified Hamiltonian.}
\end{center}
\vspace{0.4cm}
\renewcommand{\theequation}{A.\arabic{equation}}
\refstepcounter{section}
\renewcommand{\thesection}{A}
\setcounter{equation}{0}

In the appendix we show some properties of the CAP-modified
Hamiltonian, that is the Hamiltonian modified by adding a complex
absorbing potential. At first we work under the general assumptions 
\eqref{eq:CAPe}.

The semigroup $ \exp ( - i t ( P - i W )/h ) : L^2 ( X ) \to L^2 ( X )
$ is defined using the Hille-Yosida theorem: for $ h $ small
$ P - i W - i $ is invertible as its symbol is elliptic in the
semiclassical sense (see \eqref{eq:CAP} and \cite[Theorem 4.29]{e-z}).
Ellipticity assumption for large values of $ \xi $ also shows 
that $ P - i W $ is a Fredholm operator, and the comment
about invertibility shows that it has index $ 0 $. The estimate
\[ \| ( P - i W - z ) u \|\| u \| \geq - \Im  \langle ( P - i W - z
) u, u \rangle \geq  \Im z \| u \|^2 , \ \ \ u \in H^m_h ( X ) , \]
then shows invertibility for $ \Im z >  0 $, with the bound
\[ \| ( P - i W - z )^{-1} \|_{L^2 \to L^2 } \leq \frac 1 { \Im z } , 
\ \ \Im z > 0 . \]
Since the domain of $ P - i W $ is given by $ H^m ( X )$ which 
is dense in $ L^2 $, the hypotheses of the Hille-Yosida theorem are
satisfied, and 
\begin{gather}
\label{eq:HL} 
\begin{gathered}
\| e^{ - i t ( P - i W ) / h }\|_{ L^2 \to L^2 } \leq 1
, \ \ \ t \geq 0 , \\ 
e^{ - i t ( P - i W ) / h } e^{ - i s ( P - i W ) / h } = e^{ - i (t +
  s ) ( P - i W ) / h }, \ \ t, s \geq 0 . 
\end{gathered}
\end{gather}
Alternatively we can show the existence of the semigroup 
$ \exp ( - it ( P- i W ) /h ) $ using energy estimates,
just as is done in the proof of \cite[Theorem 10.3]{e-z}. We get that for any $ T>0 $, 
\begin{equation}
\label{eq:Sob}  e^{- i t ( P - i W ) / h } \in C\big( [ 0 , T ] ; {\mathcal L} (  H_h^s
( X ) , H_h^s ( X ) ) \big) \cap C^1 \big( [ 0 , T ]; { \mathcal L} ( H_h^{s} , 
H_h ^{s-m} ) \big) .
\end{equation}
Our final estimates will all be given for 
$ L^2 $ only and that is sufficient for our purposes.

The first result we state concerns propagation of semiclassical 
wave front sets. We recall the notation $\varphi_t=\exp(tH_p)$ for the Hamiltonian flow generated by $p(x,\xi)$.

\begin{lem}
\label{l:propa}
Suppose that $ A \in \Psi^\comp ( X ) $. Then for any $ T $
independent of $ h $ there exists a smooth family of operators
\begin{equation}
\label{eq:AQ}   [ 0, T ] \ni t \longmapsto Q ( t ) \in \Psi^\comp ( X ) , \ \
\WFh ( I - Q ( t ) ) \cap \varphi_t ( \WF_h ( A ) ) = \emptyset
, \end{equation}
such that 
\begin{equation} 
\label{eq:AQA}   ( I - Q ( t ) )\, e^{ - i t ( P - i W ) / h } A = {\mathcal O}( h^\infty ) _{L^2
  \to L^2 }  .
\end{equation}
In addition if $ \WFh ( A ) \subset w^{ -1} ( [ \epsilon_1, \infty ) ) $, $\epsilon_1 > 0 $, then for any fixed $t>0$,
\begin{equation}
\label{eq:AQW}
 e^{ - it ( P - i W )/h } A =  {\mathcal
  O}( h^\infty )_{ L^2\to L^2}  , \ \ \ \ A \,e^{ - it ( P - i W )/h }  = {\mathcal
  O}( h^\infty )_{ L^2\to L^2} . \end{equation}
\end{lem}
\begin{proof}
We first construct $ Q ( t) $ using a semiclassical adaptation 
of a standard microlocal procedure  -- see \cite[\S 23.1]{Hor2}. 
For that, let $ Q ( 0 ) \in \Psi^{\comp} ( X) $ 
be an operator satisfying $ \WFh ( I - Q ( 0 ) ) \cap \WFh ( A ) = 
\emptyset $, and with the principal symbol, $ q_0 ( 0 )$, independent
of $ h $.
Using the fact that the flow $ \varphi_t  $ is defined for all $ t $ 
we put $ q_0 ( t ) \defeq \varphi_{-t}^* q_0 (0) $. In terms of 
the Poisson bracket on the extended phase space
$ T^*( \RR_t \times X )\ni ( t, x, \tau, \xi )$, this means
that the function $q_0(t)$ satisfies the identity
$ \{ \tau +  p , q_0 ( t ) \} = 0  $. Consequently, at the quantum level we have
\begin{gather*}   [ h D_t + P , \Op ( q_0 ( t ) ] = h R_1 ( t ) , \ \  R_1 ( t ) \in
\Psi^\comp ( X) , \\ \Op ( q_0 ( 0 ) ) - Q ( 0 ) = h E_1 , \ \ E_1 \in
\Psi^\comp ( X ) , \end{gather*}  
and the principal symbols of $ R_1 $, $ E_1 $, $ r_1 , e_1 \in \CIc (
T^*X ) $, are independent of $ h$. If $ p_1 = \sigma ( ( P - \Op ( p )
)/ h $, we then solve (in the unknown $q_1(t)$) the equation
\[  \{ \tau + p , q_1 ( t ) \} = r_1 - \{ p_1 , q_0 ( t ) \} ,  \ \
q_1 ( 0 ) = e_1 . 
\]
By iteration of this procedure we obtain $ q_\ell \in \CI ( T^*X ) $ such that
 \begin{gather*}   [ h D_t + P ,  \sum_{ \ell =0}^{N-1} h^j \Op ( q_\ell (
  t )) ] = h^N R_N ( t ) , \ \  R_N ( t ) \in
\Psi^\comp ( X) , \\  \sum_{ \ell =1}^{N-1} h^\ell \Op ( q_\ell ( 0 ) ) - Q ( 0 ) = h^N E_N , \ \ E_N \in
\Psi^\comp ( X ) .  \end{gather*} 
By a standard Borel resummation we may construct $Q(t)\in \Psi^{\rm{comp}}(X)$ such that $ Q ( t ) \sim \sum_{\ell\geq 0}h^j \Op (q_\ell ( t ) ) $.

For any $N>0$ we can iteratively construct a sequence of auxiliary operators 
$ Q_j (  t )= Q_j ( t )^*  \in \Psi^\comp ( X ) $, $ 0\leq j\leq N$,  satisfying
\begin{equation} 
\label{eq:Qtj}
\begin{split}
 \WFh ( I - Q_{ j+1} ( t ) ) \cap \WFh ( Q_{j}  ( t ) )
& = \WF_h ( I - Q_{j}  ( t ) ) \cap \varphi_t ( \WF_h ( A ) ) \\
& = \WFh ( I - Q ( t ) ) \cap \WFh ( Q_j ( t ) )   = \emptyset ,  \\
& \hspace{-3.5cm} [  Q_j ( t ) , h D_t + P ] \in \CI \big( [ 0 , T ];  h^\infty
\Psi^{\comp } ( X ) \big) . 
\end{split}
\end{equation} 
(These assumptions imply that $ \varphi_t( \WFh ( A ) )
\subset \WF ( Q_j ( t ) ) \subset \WF ( Q_{j+1} ( t ) )  \subset \WF ( Q ( t ) ) $.)

Let $ v ( t ) \defeq e^{ - i t ( P - i W ) /h } A u $, $ \| u \|_{ L^2
  } = 1 $. Our aim is to prove the following property: 
\begin{equation}
\label{eq:indh}  w_j ( t) \defeq ( I - Q_j ( t ) ) v ( t ) = {\mathcal O} (
h^{ j/ 2} ) _{ L^2 } , \ \  \text{ for } j=0,\ldots,N,\ \  0 \leq t \leq T . 
\end{equation}
Since $ A \in \Psi^\comp $,  \eqref{eq:Sob} shows
that this property holds for  $ j  = 0 $.  Let us now prove that, if true at the level $j$, it then holds at the level $j+1$.

Noting that 
\begin{equation}
\label{eq:Not} w_{ j+1} = ( I - Q_{j+1} ) w_j + {\mathcal O} (
h^\infty ) _{\CI}, \end{equation}
we have 
\[  ( h D_t + P - i W ) w_{j+1}   = ( I - Q_{j+1} ( t ) ) ( h D_t + P - iW ) 
w_{j}  - i [ W ,Q_{j+1} ] w_j + {\mathcal O} ( h^\infty
) _{L^2 } \]
Dividing by $ h/i $, taking the inner product with $ w_{j+1} $,
taking real parts and integrating gives 
\begin{equation}
\label{eq:hii}   \begin{split} 
& \| w_{ j+1} ( t ) \|_{L^2}^2 - \| w_{ j+1 } ( 0 ) \|^2_{L^2}  +
2 \int_0^t \langle W w_{j+1}( s)  , w_{ j+1} (s) \rangle ds  = 
\\
 &\ \ \ \ \ \ \ 
\frac 2 h  \int_0^t \Re \langle [ W,  Q_{j+1} (s) ] w_j(s)
,   w_{j+1}( s)   \rangle ds + {\mathcal O} (
h^\infty ) , 
\end{split} \end{equation}
Now, 
\begin{gather*}   ( I - Q_{ j+1}( s )  )  [ W,  ( I - Q_{j+1} (s)) ] = i h B_{j+1} (
s) +h^2
C_{ j+1} ( s)  , \\ B_{j+1} ( s ) , C_{j+1}( s )  \in \Psi^\comp ( X ) , \ \
B_{ j+1}( s ) = B_{ j+1} ( s  ) ^* .\end{gather*}
Hence, using \eqref{eq:Not} and the induction hypothesis \eqref{eq:indh}, the right hand side of \eqref{eq:hii} 
becomes
\[  2 h \int_0^t  \Re \langle C_{j+1} ( s)  w_j  ( s) , w_j ( s )
\rangle ds + {\mathcal O} ( h^\infty ) =  {\mathcal O} (h^{j+1} ) .\]
Returning to \eqref{eq:hii}  and using the non-negativity of $W$, we see that
\[ \| w_{ j+1 } ( t ) \|_{L^2}^2 \leq \| w_{j+1} ( 0 ) \|_{L^2 }^2 +
C h^{ j + 1} . \]
Since 
\[ w_{j+1} ( 0 ) = ( I - Q_{ j+1} ) A u = {\mathcal O} (
h^\infty ) _{ L^2}, \]
 we have established \eqref{eq:indh} with $ j $ replaced
by $ j+1$.  

The estimate \eqref{eq:AQA} then follows from 
\[ ( I - Q ( t ) ) v ( t ) = ( I - Q ( t ) ) w_j ( t ) + {\mathcal
  O}_{L^2}  ( h^\infty ),\]
the estimate \eqref{eq:indh} at the level $j=N$,
and the fact that $N$ could be taken arbitrary large.

To see \eqref{eq:AQW} we note that if $ A \in \Psi^\comp ( X ) $ then 
\[  \WFh ( A ) \subset w^{-1} ( [ \epsilon_1 , \infty ) \
\Longrightarrow \ \varphi_t( \WFh( A ))
 \subset w^{-1} ( [ \epsilon_1 /2  , \infty ) \   \text{ for $ 0 \leq t \leq
   \delta $.  } \]
Hence,  by \eqref{eq:AQA}, 
\[  \WF_h ( v ( t ) ) 
\subset w^{-1} ( [ \epsilon_1 /2 , \infty ) , \ \
v ( t) \defeq e^{ - i t ( P - i W ) /h } A u  , \ \ 
  \| u \|_{L^2} = 1, \ \ 0 \leq t \leq \delta .\]
This means that we can modify $ W $ into $W_1$, so that 
$$
\sigma (   W_1 ) 
  ( x, \xi ) \geq \langle \xi \rangle^k/C , \ \ \ 
  W_1 \geq c_0 ,  \text{ for $ 0 < h < h_0 $},
$$
while we have
$$
0 =   ( h D_t + P - i W ) v ( t ) = ( h D_t + P - i W_1 ) v (
t)  + {\mathcal O} ( h^\infty ) _{ \CI }\quad\text{uniformly for }0 \leq t \leq \delta\,.
$$
Taking the imaginary part of the inner product of the above expression with $ v ( t ) $ gives
\[ \frac h 2  \partial_t \| v ( t ) \|^2_{L^2}  = - \langle W_1 v ( t ) , v ( t
) \rangle + {\mathcal O} ( h^\infty ) \leq -c_0 \| v ( t) \|^2 +
{\mathcal O}  ( h^\infty ) , \]
and  hence 
\[   \| v ( t) \|_{L^2}^2 = {\mathcal O} ( h^\infty ) \quad\text{uniformly for } \delta/2\leq t \leq \delta\,.
\]
This proves the first part of \eqref{eq:AQW}. The second
part follows by taking a conjugate: $ A \,e^{ - i t ( P - i W ) /h } = 
\left( e^{ - it ( - P - i W ) / h } A^* \right)^* $, and all the
arguments remain valid for $ P $ replaced by $ - P $.
\end{proof}

The next lemma is needed in \S \ref{pr} and follows immediately from 
Lemma \ref{l:propa}:
\begin{prop}
\label{l:mod2}
Suppose that $ A \in \Psi^\comp ( X ) $ satisfies 
\begin{equation}
\label{eq:suppsi}  \WFh ( A ) \subset  p^{ -1} ( ( - \delta, \delta )) \cap w^{-1} ( [ 0 , \epsilon_1 ) ) , \end{equation}
for some $ \epsilon_1 > 0 $ and that $ T $ is independent of  $ h $.

Then there exists $ B \in \Psi^\comp ( X ) $ for which \eqref{eq:suppsi}
holds with $ B$ in place of $ A$, and
\begin{equation}
\label{eq:mod2}    e^{ -i t ( P - i W ) /h } A = B  e^{ - i t ( P - i W ) /h } A +
{\mathcal O} (  h^\infty  ) _{L^2 \to L^2} ,  \ \ 0 \leq t \leq T . \end{equation}
\end{prop}
\begin{proof}
Using again the operator $ Q ( t ) $ constructed in the proof of Lemma \ref{l:propa}, we
take a compact set $L$ containing $ \WFh ( Q ( t ) ) $ for all $ 0 \leq t \leq
T $. By taking $ \WFh ( Q ( 0 ) ) \subset p^{-1} ( ( - \delta, \delta
) ) $ (which is possible due the assumptions on $ A $) we see that we
can assume $ L \subset p^{-1} ( ( - \delta, \delta ) )$. We can now choose
$ B \in \Psi^\comp ( X ) $ such that
\[    \WFh ( I - B ) \cap L \cap w^{-1} ( [0, \epsilon_1/3] ) =
\emptyset , \ \  \WFh ( B ) \subset p^{-1} ( (- \delta, \delta) ) \cap
w^{-1} ( [ 0, \epsilon_1/2 ) . \]
This implies that  $ ( I - B ) Q ( t ) = C ( t ) $, where $ \WFh ( C(
t )) \subset w^{-1} ( [ \epsilon_1/3, \infty ) ) $, and hence, by 
\eqref{eq:AQA} and \eqref{eq:AQW}, 
\[  ( I - B ) e^{ - i t  ( P - i W ) / h} A= \big( C ( t )  +  ( I - B ) ( I
  - Q (  t) \big) e^{ - i t  ( P - i W ) / h} A= {\mathcal O}
    ( h^\infty ) _{L^2 \to L^2}, \]
proving \eqref{eq:mod2}.
\end{proof}

Finally we  present a modification of \cite[Lemma A.1]{NZ3}. The modification lies in slightly different
assumptions
on $ P $ and $ W $, and the proof also corrects a mistake in the proof 
given in \cite{NZ3}. 
From now on we work under the extra assumption \eqref{eq:aux} on the
CAP. We remark that in \cite{NZ3} we only needed Lemma \ref{l:propa}
and hence the assumption \eqref{eq:aux} was not required.
\begin{prop}
\label{l:mod1}
Suppose that $ X $ is a compact manifold, $ P $ is a 
self-adjoint operator, 
 $ P \in \Psi^m ( X) $, $ W \in \Psi^k ( X ) $, $ W \geq 0
$, and that \eqref{eq:CAPe} and \eqref{eq:aux} hold. 
Then 
for any $ t $ independent of $ h$, 
for $ A \in \Psi^\comp ( X )  $ satisfying \eqref{eq:suppsi}, we may write
\[ e^{ i t P/h } e^{ - i t ( P - i W ) / h } A  = V_A ( t ) +
{\mathcal O} ( h^\infty ) _{ L^2 \to L^2} ,\]
where 
\begin{gather}
\label{eq:VAt} 
\begin{gathered}   
V_A ( t) \in \Psi_{\gamma}^{\rm{comp} } ( X ) , \ \ 
\WFh ( V_A ( t ) ) \subset \bigcap_{ 0 \leq s \leq t } ( \varphi_{-s}
(w^{-1} ( 0 )))   \cap \WFh ( A ) ,  \\ 
\sigma ( V_A ( t ) ) = \exp \left( - \frac 1 h \int_0^t \varphi_s^* W ds
\right)\sigma ( A ) \,.
\end{gathered}
\end{gather}
The class of operators $ \Psi_{\gamma}^{\rm{comp}}$ 
was introduced in \S \ref{12c}.
\end{prop}

The proof is based on the following lemma inspired by the 
pseudodifferential approach to constructing parametrices for
parabolic equations presented in \cite{Iw}.
\begin{lem}
\label{l:iw}
Suppose that $ t \mapsto p(  t, z , h) $,  $ p ( t, \bullet, h ) 
\in \CIc ( \RR^{2n}; \RR ) $, 
is a family of functions satisfying 
\begin{equation}
\label{eq:ptx} 
\begin{gathered}  \partial_t^k \partial_z^\alpha p ( t, z , h ) = \mathcal O _{ k,
  \alpha} ( 1 ) , \ \ \ p \geq -C h  , \ \ \ 0 < h < h_0 , \\  
|  \partial_z^\alpha p ( t, z, h ) | = {\mathcal
    O}_\alpha  ( p^{ 1 -  \delta } ) , \ \ 0 < \delta < \textstyle{\frac12} . 
\end{gathered} \end{equation}
Then, for $ 0 \leq s \leq t $ there exists $ E ( t, s ) \in
\Psi_{\delta} ( \RR^n) $ such that
\[  ( h \partial_t + p^w ( t , x , h D_x , h ) ) E ( t , s ) = 0,  \ \ t
\geq s \geq 0  , \ \ E ( s, s ) =  \Id .\]
Moreover, $ E ( t, s ) = e^w  ( t , s, x, h D_x , h ) $ where $ e(t,s) \in
S_{\delta} ( \RR^{2n} ) $
has an explicit expansion given in \eqref{eq:iw1} below.
\end{lem} 
\begin{proof}
Replacing $ p $ by $ p + (C+1)h $, gives $ p \geq h $ and
$ p(t, \bullet,h )  \in (C+1)h + \CIc(\RR^{2n}_z ) $. The
multiplicative factor  $ e^{ ( C + 1 ) (t-s) } $ in the
evolution equation is irrelevant to our estimates.

For any $ N \geq 0 $ we try to approximate the symbol $e(t,s,x,\xi,h)$ by an expansion of the form
\be f_N ( t , s, z , h ) \defeq \sum_{ j=0}^N  h^j e_j ( t, s , z , h
) \,. \label{eq:fN-def}
\ee

The symbol of the operator $h\partial_t f_N^w + p^w f_N^w$ can be expanded using 
the standard notation $ a^w \circ b^w = ( a \# b )^w $ and the product formula (see for instance \cite[Theorem 4.12]{e-z}):
\begin{align}
& h \partial_t f_N ( t, s ) + \left[ p ( t ) \# f_N ( t, s) \right]    \nonumber \\
& =\sum_{ j=0}^N h^j \Big( h \partial_t e_j (t,s) 
+ \sum_{ k=0}^{N-j -1 } \frac 1 {
    k!} \left( \textstyle{\frac12}  i h  \omega ( D_z , D_w )
  \right)^k p (t, z) e_j (t,s, w ) |_{ z = w }  \nonumber 
+  h^{N-j} r_{N,j} 
\Big) 
\nonumber \\
& = \nonumber 
\sum_{ j=0}^N h^j \Big( ( h \partial_t + p ( t ))e_j ( t, s) + 
\sum_{ \ell=0}^{j -1 }  \frac 1 {
    (j-\ell)!} \left( \textstyle{\frac12}  i  \omega ( D_z , D_w )
  \right)^{j-\ell} p (t, z) e_\ell (t,s, w ) |_{ z = w }\Big) \\
& \ \ \ \ \ \ \ \ \ \ \ \ \ \ \  + h^{N} r_N ( t, s , z ) , \ \ \ \  \
\ \ \ \ \ \ \ \ \ r_N ( t, s, z ) \defeq  \sum_{ j=0}^{N-1} r_{N,j} (
t, s, z )  .\nonumber
\end{align}
The remainders satisfy the following bounds (see for instance \cite[(3.12)]{SZ10}):
\begin{equation}
\label{eq:312}
\begin{split}
& \sup_{z} | \partial_z^\alpha r_{N,j} ( t, s, z ) | \leq \\
& \ \ \ \ \ \ 
C_{\alpha, N , j}  \sum_{ \alpha_1 + \alpha_2 = \alpha } \sup_{ z, w } \sup_{ | \beta|
  \leq M , \beta \in \NN^{4n} } \left| ( h^{\frac12} \partial_{z,w})^\beta
( \sigma ( D_z, D_w) )^{N-j} \partial_z^{\alpha_1} p ( z
) \partial_w^{\alpha_2 } e_j ( w ) \right|. 
\end{split} \end{equation}
The standard strategy is now to iteratively construct the symbols $ e_j$ so that each term in the above
expansion vanishes. 
The term $j=0$ simply reads $(h\partial_t + p)e_0=0$. From the initial condition $e_0(s,s)\equiv 1$, it is solved by 
\begin{equation}
\label{eq:e0}  e_0 ( t , s , z ,  h ) = \exp \left( - \textstyle{\frac 1 h } \int_s^t p (
s' , z , h ) ds'  \right)\,.
\end{equation}
For $ j \geq 1 $, the symbol $e_j$ is obtained iteratively by solving
\begin{equation}
\label{eq:ejt} 
\begin{gathered}  e_j ( t, s , z ) \defeq  \frac 1 h \int_s^t e_0 ( t, s' , z )
q_j ( s', s , z ) ds', \ \ \ 
e_j ( t, s ,
  \bullet ) \in \CIc ( \RR^{2n} ) , 
\\  q_j ( t, s , z ) \defeq - \sum_{ \ell=0}^{j -1 } \frac 1 {
    (j-\ell)!} \left( \textstyle{\frac12}  i  \omega ( D_z , D_w )
  \right)^{j-\ell} p (t, z) e_\ell (t,s, w ) |_{ z = w }  \in \CIc (
  \RR^{2n}_z  ) \,.
\end{gathered}
\end{equation}
This construction formally leads to an approximate solution:
\begin{equation}
\label{eq:fN} h \partial_t f_N ( t, s , z ) + \left[ p ( t, \bullet  ) \# f_N (
 t, s, \bullet ) \right] ( z )  = h^N r_N ( t,s, z ) . 
\end{equation}
To make the approximation effective, we now need to check that the sum \eqref{eq:fN-def} is indeed an expansion in power of $h$. We thus need to estimate the $ e_j$'s and thereby the remainders $r_{N,j} $'s. 

We will prove the following estimate by induction:
\begin{equation}
\label{eq:iw3}
| \partial_z^\alpha e_j ( t, s , z ) | \leq C_{\alpha, j  }   h^{- 2
  \delta j - \delta   |\alpha|  } \left( 1 + 
\left( \textstyle{\frac 1 h} \int_s^t p ( s', z )    ds' \right)^{  2 j  + | \alpha| } \right)  e_0 ( t, s , z ) . 
\end{equation}
For that we first note that, as $ p \geq h $,  and 
$| \partial^\alpha
p | \leq C_\alpha p^{1-\delta}  $, we have
\begin{equation}
\label{eq:stand}   | \partial^\alpha p | \leq  C_\alpha h^{ - \delta } p.
\end{equation}
Consequently, for $ j = 0 $ we have
\be\begin{split} 
\label{eq:e0-estimate}
 |\partial_z^\alpha e_0 ( t, s , z ) | & \leq 
 \sum_{ \sum_{\ell=1}^k \alpha_\ell = \alpha}
 \prod_{\ell=1}^k \left( \textstyle{\frac 1 h }
   \int_s^t |\partial^{\alpha_\ell} p ( s' , z )|  ds'  \right) e_0 ( t,
 s, z ) 
 \\
& \leq C_\alpha 
\sum_{ \sum_{\ell=1}^k  \alpha_\ell = \alpha} 
\prod_{\ell =1}^k  \left( h^{ -\delta }
\textstyle{\frac 1 h }    \int_s^t  p ( s' , z )   ds'  \right) 
e_0 ( t,  s, z ) 
 \\
& \leq C_\alpha'  h^{ - \delta | \alpha | }  \left( 1 + 
\left(\textstyle{\frac 1 h }    \int_s^t  p ( s' , z )   ds'
\right)^{|\alpha| }\right) e_0 ( t, s , z ) , 
\end{split}
\ee
Here we used the fact that $ k \leq | \alpha | $ and that
\[   A^{k} \leq c_\alpha ( 1 + A^{|\alpha | } ) , \ \ \ A =
\textstyle{\frac 1 h }    \int_s^t  p ( s' , z )   ds' \geq 0 .\]
This gives \eqref{eq:iw3} for $ j = 0 $.

To proceed with the induction we put
\[   a_{j , \alpha} ( t, s , z ) \defeq { 
 \partial^\alpha_z e_j ( t, s , z ) }   / { e_0 ( t, s, z )} , \ \ \ 
b_{ j , \alpha } ( t,s , z ) \defeq { 
 \partial^\alpha_z q_j ( t, s , z ) }   / { e_0 ( t, s, z )} , \]
noting that, for some coefficients, $ c_\bullet $, 
\begin{equation}
\label{eq:bjaj}
\begin{split}
& b_{j, \alpha }( t, s, z )  = \sum_{ \ell=0}^{j-1} 
  \sum_{    \beta_1 + \beta_2 = \alpha } 
c_{ \beta_1, \beta_2, \ell, j}   \omega ( D_z, D_w )
^{ j - \ell }  \partial_z^{\beta_1 } p ( t, z ) a_{ \ell , \beta_2 } ( t ,s,w )
  |_{z = w } , \\
& a_{ j, \alpha} ( t ,s, z ) = \frac 1 h \sum_{ \beta_1 + \beta_2 = \alpha} c_{  \beta_1, \beta_2, j} \int_s^t a_{ 0 , \beta_1 } ( t, s' ,z ) b_{ j,  \beta_2 } ( s', s) d s' ,
\end{split}
\end{equation}
  where the last equality follows from $ e_0 ( t, s' ,z ) e_0 ( s',
  s, z ) = e_0( t, s, z ) $, $ s \leq s' \leq t $.

Our aim is to show
\begin{equation}
\label{eq:ind1} | b_{j, \alpha} ( t, s , z) | \leq 
C_{\alpha, j  }   h^{- 2 \delta  j -\delta |\alpha| }  p ( t ,  z ) 
 \left( 1 + 
\left( \textstyle{\frac 1 h} \int_s^t p ( s', z )
   ds' \right)^{2 j + | \alpha| -1 } \right)  ,
\end{equation}
and
\begin{equation}
\label{eq:ind2} | a_{j, \alpha} ( t, s , z) | \leq 
C'_{\alpha, j  }   h^{-2  \delta  j -  \delta |\alpha| } \left( 1 + 
\left( \textstyle{\frac 1 h} \int_s^t p ( s', z )
   ds' \right)^{ 2j + | \alpha| } \right) ,
\end{equation}
assuming the statements are true for $ j $ replaced by smaller
values. 

We note that the case of $ j =0 $ has been shown in \eqref{eq:e0-estimate}, and since $b_{0, \alpha} \equiv 0$.

The first estimate \eqref{eq:ind1} follows immediately from the
inductive hypothesis on $ a_{\ell, \alpha} $, $ 0 \leq \ell \leq j-1
$ and the estimates on $ p $ in \eqref{eq:stand}.
The second estimate \eqref{eq:ind2} follows from \eqref{eq:e0-estimate}, \eqref{eq:ind1} 
and the obvious fact that $ \int_{s_1}^{s_2} p (s') ds' \leq \int_s^t p
(s') ds' $, $ s\leq s_1 \leq s_2 \leq t$. 

We note that \eqref{eq:iw3} and the definition of $ e_0$ given in 
\eqref{eq:e0}  imply that 
$$ \partial^\alpha_z e_j (t,s, z ) = {\mathcal O} ( h^{ - \delta |\alpha| - 2 \delta j })\,,\quad j\geq 0\,.
$$
so from \eqref{eq:fN-def} we see that the symbol $f_N(t,s)\in S_{\delta}(\IR^{2n})$.

The bounds  \eqref{eq:312} then show that the remainders 
satisfy
 \[ | \partial^\alpha r_N (t, s , z ) | \leq C_{N, \alpha}  
h^{ - 2 \delta N - \delta | \alpha |   } . \]
Going back to \eqref{eq:fN} we get the expression
\begin{equation}
\label{eq:Duh} E( t, s ) = f_N^w ( t,s , x , h D_x ) + h^{N-1}  \int_s^t E (
t, s' ) r_N^w  ( s', s, x , h D_x ) . 
\end{equation}
(We note that, since $ p^w ( t, x,h D_x ) \geq - Ch $ by the sharp G{\aa}rding
inequality \cite[Theorem 4.32]{e-z}, and since $ p^w  $ is bounded on $
L^2$, the operator $ E ( t, s) $ exists and is bounded on $ L^2 $, uniformly in $
h$.) Since operators in $ \Psi_{\delta} $ are uniformly bounded on $ L^2$
\cite[Theorem 4.23]{e-z}, it follows that
\[  E ( t , s ) = f^w_N ( t, s ,x , h D_x ) +  {\mathcal O}( h^{( 1 -  2 \delta ) N } ) _{L^2 \to L^2}
. \]  
To show that $ E ( t , s ) - e^w_0 ( s, t, x , h D_x) 
\in \Psi^\comp_{\delta} ( \RR^n ) $, we use \eqref{eq:Duh} and Beals's
lemma in the form given in \cite[Lemma 3.5, $ \tilde h =1 $]{SZ10}:
$ \ell_j $ are linear functions on $ \RR^{2n} $, $ \ell_j^w = \ell_j^w ( x, h D) $, then 
\[\begin{split} &    \ad_{\ell_1^w } \cdots 
\ad_{ \ell_J^w } E ( s, t ) = \\
& \ \ \   \ad_{\ell_1^w } \cdots \ad_{ \ell_J^w }  f_N^w ( s, t, x , h D_x) + h^{N-1} 
\int_s^t \ad_{\ell_1^w } \cdots \ad_{ \ell_J^w }  
\left( E ( s, s' ) r^w_N ( s', s, x, h D_x ) \right) ds' \\
&\ \ \ = {\mathcal O} ( h^{ ( 1 - 2 \delta)  J} ) _{ L^2 \to L^2 } + 
{\mathcal O} ( h^{ ( 1 - 2 \delta ) N } ) _{ L^2 \to L^2 }
= \mathcal O ( h^{( 1 - 2 \delta ) J } ) _{ L^2 \to L^2 } , \end{split} \]
if $ N $ is large enough. Here we used the fact that $ f_N, r_N  \in S_{\delta} $ and that 
$   \ad_{\ell_1^w } \cdots \ad_{ \ell_J^w }  
E ( s, t )= {\mathcal O} (1 ) _{L^2 \to L^2 }$, which follows from considering the evolutions equation for
the operator on the left hand side.

In conclusion we have shown that $ E ( t, s ) = e^w ( t, s , x , h
D_x) $, where $ e \in S_{\delta} ( \RR^n ) $ admits the expansion
\begin{equation}
\label{eq:iw1}
\begin{gathered}
e( t, s , z , h ) \sim \sum_{ j \geq 0 } h^j e_j ( t , s, z, h ) , \ \
\quad
e_j(t,s) \in h^{- 2 \delta  j} S^\comp_{\delta} (\IR^{2n}) ,  \ \ j \geq 1, 
\end{gathered}
\end{equation}
with $ e_0 $ given by \eqref{eq:e0}. 
\end{proof}

\begin{proof}[Proof of Proposition \ref{l:mod1}]
We first observe that Lemma \ref{l:propa} (applied both to propagators
for  $ P - i W $ and for $ P $) shows that
for  $ B \in \Psi^\comp ( X ) $ satisfying $ \WFh ( I - B ) \cap \WFh
( A ) = \emptyset $, 
\[    e^{ i t P/h } e^{ - i t ( P - i W ) / h } A = B   e^{ i t P/h }
e^{ - i t ( P - i W ) / h } A + {\mathcal O} (h^\infty
) _{L^2 \to L^2}. \]
We can choose  $ B = B^* $.
Since
\[ \begin{split}  h \partial_t \left( B   e^{ i t P/h }
e^{ - i t ( P - i W ) / h } A \right) & = - B e^{it P/h } W e^{ -i t P/h
}  e^{ i t P/h }  e^{ - i t ( P - i W ) / h } A \\ & = 
- \left( B e^{it P/h } W e^{ -i t P/h } B  \right) \left(  B e^{ i t P/h }  e^{
  - i t ( P - i W ) / h } A\right) +  {\mathcal O}
(h^\infty ) _{L^2 \to L^2}, \end{split} \]
it follows that 
\begin{equation}
\label{eq:ABe}
B\,e^{ i t P/h } e^{ - i t ( P - i W ) / h } A = V^{B} ( t )  +
 {\mathcal O} (h^\infty ) _{L^2 \to L^2}, 
\end{equation}
where
\begin{equation}
h \partial_t V^B ( t) = -  W_B ( t) V^B ( t ) , \quad  W_B ( t ) \defeq B
e^{i t P/h } W e^{- it P/h} B . 
\end{equation}
We note that $ W_B ( t ) \in \Psi^\comp ( X)$, $ \WFh ( W_B ( t ) )
\subset \WFh ( B ) $, and that $W_B(t)\geq 0$. 
Hence 
$ V^B ( t ) = {\mathcal O}( 1 ) _{L^2 \to L^2}  $ and \eqref{eq:ABe} 
follows from Duhamel's formula. 

By decomposing $ A $ as a sum of operators, we can assume that $ \WFh
( A ) $ is supported in a neighbourhood of a fiber of a point in $
X$. 
Hence, by choosing $ B$ with a sufficiently small wave front set,
we only need to prove that $ V^B ( t ) \in \Psi_{\delta} $ for $ X =
\RR^n $;  that follows from Lemma \ref{l:iw}, since the
  symbol of $W_B(t)$ satisfies the assumptions \eqref{eq:ptx}. 
The second and  third 
properties
in \eqref{eq:VAt} follows from \eqref{eq:AQW} and 
\eqref{eq:iw1}.
\end{proof}

\medskip

\smallsection{Acknowledgements}
We would like to thank Kiril Datchev,
Semyon Dyatlov, Fr\'ed\'eric Faure and Andr\'as Vasy for helpful discussions of the
material in \S\S \ref{recap}, \ref{decay} and the Appendix, and of 
connections with previous works. We are particularly grateful to 
the anonymous referee for the careful reading for the manuscript
and for many useful suggestions.
The partial supports by the Agence Nationale de la Recherche under
grant ANR-09-JCJC-0099-01 (SN) 
and by the National Science Foundation under the
grant DMS-1201417 (MZ), are also acknowledged.

\end{document}